\documentclass[12pt]{article}
\usepackage{arxiv}
\usepackage[utf8]{inputenc} 
\usepackage[T1]{fontenc}    
\usepackage{hyperref}       
\usepackage{url}            
\usepackage{booktabs}       
\usepackage{amsfonts}       
\usepackage{nicefrac}       
\usepackage{microtype}      
\usepackage{lipsum}
\usepackage{stix}
\usepackage{xspace}
\usepackage{amsmath, amsthm}
\usepackage{enumerate}
\usepackage{comment}

\usepackage{listings}
\usepackage{float}
\usepackage{graphicx}
\usepackage{mathrsfs} 
\usepackage{hyperref}
\usepackage{cite}
\usepackage{color}
\usepackage{caption}
\usepackage{subcaption}
\usepackage{color}
\usepackage{cancel}
\usepackage{bbm}
\usepackage{tikz}
\usepackage{fullpage} 
\usepackage{multirow}

\newcommand{\appn}{\;\shortstack{$ \approx $ \\ $\scriptstyle{ n \to \infty}$}\;}
\newcommand{\preceqnh}{\;\shortstack{$\preceq$ \\ $\scriptstyle{n\to \infty}$ \\ $\scriptstyle{h\to 0}$}\;} 
\newcommand{\preceqh}{\;\shortstack{$ \approx $ \\ $\scriptstyle{ h \to 0}$}\;}

\newtheorem*{assumption*}{\assumptionnumber}
\providecommand{\assumptionnumber}{}
\makeatletter
\newenvironment{assumption}[2]
 {%
  \renewcommand{\assumptionnumber}{Assumption #1#2}%
  \begin{assumption*}%
  \protected@edef\@currentlabel{\textbf{#1#2}}%
 }
 {%
  \end{assumption*}
 }
 
\newtheorem{theorem}{Theorem}[section]

\newtheorem{remark}{Remark}[section]
\newtheorem{lemma}{Lemma}[section]

\def\car{\mbox{\rm{1\hspace{-0.08 cm }I}}}
\def\W{\mathcal{W}}

\def\R{\mathbb{R}}
\def\Rd{\mathbb{R}^d}
\def\R+{\mathbb{R}^+}
\def\E{\mathbb{E}}

\title{Spatio - Temporal Weighted Regression Model with Fractional-Colored
Noise:  Parameter estimation and consistency}

\author{
  Lisandro J. Fermín\\ 
  Instituto de Ingeniería Matemática\\
  Facultad de Ingeniería\\
  Universidad de Valparaíso\\
  \texttt{lisandro.fermin@uv.cl} \\
   \And
 Silfrido Goméz \\
 Instituto de Estadística\\
 Facultad de Ciencias\\
  Universidad de Valparaíso\\
  \texttt{silfrido.gomez@uv.cl} \\
  \And
 Soledad Torres \\
  Instituto de Ingeniería Matemática\\
  Facultad de Ingeniería\\
  Universidad de Valparaíso\\
  \texttt{soledad.torres@uv.cl} \\
   \And
Héctor Araya \\
  Facultad de Ingeniería y Ciencias\\
  Universidad Adolfo Ibáñez\\
  \texttt{hector.araya@uai.cl} \\
   \And
Tania Roa \thanks{corresponding author} \\
  Facultad de Ingeniería y Ciencias\\
  Universidad Adolfo Ibáñez\\
  Data Observatory\\
  \texttt{tania.roa@uai.cl} \\
}

\begin{document}
\maketitle
\begin{abstract} 
Geographical and Temporal Weighted Regression (GTWR) model is an important local technique for exploring spatial heterogeneity in data relationships, as well as temporal dependence due to its high fitting capacity when it comes to real data. In this article, we consider a GTWR model driven by a spatio-temporal noise, colored in space and fractional in time. Concerning the covariates, we consider that they are correlated, taking into account two interaction types between covariates, weak and strong interaction. Under these assumptions, Weighted Least Squares Estimator (WLS) is obtained, as well as its rate of convergence. In order to evidence the good performance of the estimator studied, it is provided a simulation study of four different scenarios, where it is observed that the residuals oscillate with small variation around zero. The STARMA package of the R software allows obtaining a variant of the $R^{2}$ coefficient, with values very close to 1, which means that most of the variability is explained by the model.
\end{abstract}

\keywords{Geographically and Temporally Weighted Regression \and Fractional Colored Noise \and Consistency} 

\textbf{MSC:} Primary 62M30, Secondary 62M10.

\section{Introduction}

Spatio-temporal weighted regression models have been  widely used to analyze and  visualize geo-referenced information in many research areas. Some examples, can be evidenced in the exploration of spatio-temporal patterns of human behavior \cite{Chen2011,Kwan2000} , modeling the variation of housing prices as a function of their georeferencing \cite{Foth2003}, criminal activities \cite{Brundson2007,Naka2010}; , disease outbreaks \cite{Taka2008} and, in methods for analyzing and visualizing data in space and time \cite{Andrei2010,Demsar2010,Rey2009}. 
Within the theory of geospatial statistics, these models have allowed the deepening  of  environmental variables analysis such as  the temperature present in certain locations and soil moisture, among others,  through satellite images captured over the earth's surface at different moments in time and which, by means of strategically located temperature sensors, allow modeling the spatio-temporal behavior of the ground surface temperature. Such is the case of the work done by the authors in \cite{Peng2019}; where they propose a new algorithm based on a geographically and temporally weighted regression model, for the spatial downscaling of the radiometric spectrum of moderate resolution images from 1000 to 100 meters, in data related to ground surface temperature. It is worth mentioning that the use and implementation of these spatio-temporal weighted regression models, is largely due to the high fitting capacity it has with respect to real data, both globally and locally. Furthermore, the recommended use of spatio-temporal weighted regression models on georeferenced data is always advisable when the data present heterogeneity and stationarity at the spatio-temporal level. For example, the authors in \cite{Shol2017}; where economic growth is compared between regions in India, using two different models, a global spatio-temporal regression model and a spatio-temporal weighted regression model. Thus, the researchers show through the results obtained a better fit by the spatio-temporal weighted regression model than that obtained with the global model.

To study these models, it is necessary to understand the complexity of the spatio-temporal covariance structure between the explanatory variables, and the behavior of the error considered within the model.  Thus, given a spatio-temporal weighted regression model, which is specifically inspired by the geographically and temporally weighted regression model proposed by the authors in \cite{Crespo2015}, we state the regression model: 

\begin{equation}\label{model}
    Y_i=\beta_0(z_i) + \sum_{j=1}^p \beta_j(z_i)X_{i,j} + \epsilon_i,\quad i=1,\ldots,n \, ; \,\,\,\,\,p \in \mathbb{N}
\end{equation}

 where $z_i=(t_i, u_i)$ denotes the coordinates of the 
observation at point $z_i$, in the space $u_i  \in \mathbb{R}^d$ space at a time $t_i$ in $t_i \in \mathbb{R}^+$, $\beta_0(z_i)$ denotes the value of the intercept, $\beta_j(z_i)$ 
denotes the parameter associated with the $j_{th}$ covariate $X_j$ at point $z_i$, and $\epsilon_i$ is the colored fractional noise at point $z_i$, defined in \cite{Man1968}, i.e., $\epsilon=(\epsilon_i)_{i=1,\ldots,n}$ is a Gaussian noise that behaves like a fractional Brownian motion (fBm) in time and has white or colored spatial covariance. Basically, $\epsilon_i$ by presenting these characteristics, it intuitively gives us an idea of the level of irregularity or variability that can present the spatio - temporal information that is known, in relation to  what we want to estimate.

In this work, our main result, proves the strong consistency of the spatio-temporal weighted least squares estimator (WLSE), under certain H\"older-type regularity conditions on the continuity of the spatio-temporal trajectories described by the covariates. This estimator is expressed as: 
\begin{eqnarray*}
\hat{\beta}(z_{i}) &=&(X^{T}\W(z_i)X)^{-1}X^{T}\W(z_i)(X\beta(z_{i})+\epsilon),\\
\end{eqnarray*}
where $z_i$ denotes the $z_i^{th}$-th spatio-temporal observation, $X$ is a $n \times (p+1)$ - order matrix corresponding to the covariate entries, $Y$ is the $n$-dimensional vector of spatio - temporal observations, $\W(z_i)$ is a positive definite symmetric $n \times n$ - order matrix, known as the weights matrix, and $\epsilon = (\epsilon_l)_{l=1: n}$ has as its associated covariance function:
\begin{eqnarray}
\mathbb{E}(\epsilon_l \epsilon_{l'}) &=& \frac{1}{2}\left( \int_{t^{-}_l}^{t_l^+} \int_{t_{l'}^-}^{t_{l'}^+} 2H(2H-1)|t-t'|^{2H-2} dt' dt
\right) \nonumber \\
& \times & \left(\int_{\Rd} \int_{\mathbb{R}^d} \car_{V(u_l)}(u) \gamma_{\alpha,d}\|u-v\|^{-d+\alpha} \car_{V(u_{l'})}(v) du dv\right). \label{colored_cova}
\end{eqnarray}

The above expression \eqref{colored_cova}, is derived from the work \cite{Torres2014}, where the Riesz kernel function of order $\alpha$, given by $\gamma_{\alpha,d}\|u - v\|^{-d+\alpha}$, is considered as the spatial covariance of the noise. Our main result is the convergence in $L^{2}$ and in probability of the spatio -temporal weighted least squares estimator. Finally, in this work we perform a simulation study on possible scenarios that can be considered for estimating under regularity conditions assigned to the Hurst $H$ index chosen for the spatial and temporal covariance. These considered cases are accompanied by a graphical display of the stability of the Mean Squared Error (MSE) of the spatio -temporally weighted least squares estimator in each situation.

The organization of the present work was structured as follows: Section 2 presents the weighted regression model considered in this work; the fractional colored noise, as well as the spatio-temporal point measure of the noise over the observations $z_{i}$ . We also derive the explicit form of the spatio-temporal colored fractional noise covariance function, and the correlation type between the explanatory variables along with the assumptions to be considered, are defined. The spatio-temporal weighted least squares estimator of the proposed model is shown. In Section 3, the convergence in quadratic norm of the weighted least-squares estimator is proven, via an auxiliary lemma that proves the convergence of the least - squares estimator in probability. Section 4 presents the results and simulation work performed based on the different scenarios considered and finally Section 5 includes an appendix showing the details of the proof for the auxiliary lemma that is considered in the proof of the paper’s most important result.


\section{The model}

\subsection{Weighted regression model}\label{wrm}

The geographically and temporally weighted regression (GTWR) model is a spatio-temporal varying coefficient regression approach for exploring spatial nonstationarity and temporal dependence of a regression relationship for spatio-temporal data. The GTWR model can be expressed as follows:

\begin{equation}\label{model}
    Y_i=\beta_0(z_i) + \sum_{j=1}^p \beta_j(z_i)X_{i,j} + \epsilon_i,\quad i=1,\ldots,n \, ; \,\,\,\,\,p \in \mathbb{N}
\end{equation}
where $z_i=(t_i, u_i)$ denotes the coordinates of the observation point $z_i$, in space $u_i\in\mathbb{R}^d$ at time $t_i\in \mathbb{R}^+$, $\beta_0(z_i)$ indicates the intercept value, $\beta_j(z_i)$ indicates the parameter associated with the $j_{th}$ covariate $X_j$ at point $z_i$, and $\epsilon_i=\Delta W^H(z_i)$ is the fractional colored noise at $z_i$; i.e. $\epsilon=(\epsilon_i)_{i=1,\ldots,n}$ is a Gaussian noise which behaves like fractional Brownian motion (fBm) in time and has white or colored spatial covariance in space.

\begin{assumption}{M}{1}\label{M1}
The noise $\epsilon=(\epsilon_i)_{i=1,\ldots,n}$, is independent of the covariates $(X_1,\ldots, X_p)$, where {$X_j \in \mathbb{R}^{n}$} for every $j=1, \ldots , p$.
\end{assumption}

\subsection{Fractional-Colored Noise}

We begin by describing the spatial covariance of the noise. Let us recall the frame-work from \cite{Torres2014}. Let $\mu$ be a non-negative tempered measure on $\mathbb{R}^d$, i.e. a non-negative which satisfies the following condition

\begin{assumption}{N}{1}\label{N1}
$\int_{\mathbb{R}^d}(1+\|\xi\|^2)^{-\ell} \mu(d\xi) < \infty, \quad \mbox{ for some } \quad \ell >0$.
\end{assumption}

Let $f:\mathbb{R}^d \rightarrow \mathbb{R}^+$ be the Fourier transform of $\mu$ in $\mathcal{S}(\mathbb{R}^d)$ (Schwarz space of rapidly decreasing $C^{\infty}$ functions on $\mathbb{R}^{d}$, see. \cite{tudor, tudor2} for details), i.e.
\begin{equation}\label{def_Fourier}
    \int_{\mathbb{R}^d} f(u)\varphi(u)du = \int_{\mathbb{R}^d} \mathcal{F}\varphi(\xi)\mu(d\xi), \quad \mbox{ for all } \quad \varphi\in \mathcal{S}(\mathbb{R}^d).
\end{equation}

Let the Hurst parameter $H$ be fixed in $(1/2,1)$. On a complete probability space $(\Omega,\mathcal{F},\mathbb{P})$, we consider a zero-mean Gaussian field $W^H = \left\{W_t^H(A): t \geq 0, A \in \mathcal{B}_b(\mathbb{R}^d)\right\}$,  defined on the set of bounded Borel measurable functions $\mathcal{B}_b(\mathbb{R}^d)$, with covariance

\begin{eqnarray}\label{cov_noise_0}
\mathbb{E}\left( W^H_t(A) W^H_s(B) \right)
&=& R_H(t,s) \int_{\mathbb{R}^d} \int_{\mathbb{R}^d} \car_{A}(u)f(u-v)\car_{B}(v)du dv\\
& := & \left< \car_{[0,t]\times A}, \car_{[0,s]\times B} \right>_{\mathcal{H}}, \nonumber 
\end{eqnarray}
where $R_H$ is the covariance of the fBm 

\begin{equation}\label{cov_fbm}
R_H(t,s)
=  \frac{1}{2} \left( t^{2H}+s^{2H}-|t-s|^{2H}\right),  \quad \mbox{ for } \quad s,t\geq 0,
\end{equation}

and $\mathcal{H}$ is the canonical Hilbert space associated with the Gaussian process $W$ is defined as the closure of the linear span generated by the indicator functions $\car_{[0,t]\times A}$,   $t \in [0,T]$, $A \in \mathcal{B}_b(\mathbb{R}^d)$ with respect to inner product given by the right hand side of \eqref{cov_noise_0}.

 This can be extended to a Gaussian noise measure on $\mathcal{B}_b(\mathbb{R^+} \times \mathbb{R}^d)$  by setting 

\begin{equation}
    W^H ((s, t] \times A) := W_t^H (A) - W_s^H (A).
\end{equation}

We suppose that the spatial covariance is given by a Riesz kernel $f$ of order $\alpha$ satisfying the following condition
\begin{assumption}{N}{2}\label{N2}
We consider $f$ as following   $f_{\alpha}(u):=\gamma_{\alpha,d}\|u\|^{-d+\alpha}$,  for  $-d<\alpha < d$ and $\gamma_{\alpha,d}=2^{d-\alpha}\pi^{d/2}\Gamma((d -\alpha)/2) / \Gamma(\alpha/2)$. In this case, $\mu(d\xi)=\|\xi\|^{-\alpha}d\xi$.
\end{assumption}

\begin{remark}
Under \ref{N2} condition \ref{N1} is satisfied for $d-\alpha<2\ell$. The special case of white noise in space is identical to the particular case of condition \ref{N2} with $\alpha=0$, in which case $\mu$ is the Lebesgue measure. 
\end{remark}

Fractional colored noise at observation point $z_l$ is defined as $\epsilon_l=\Delta W^H(z_l)$, this represents the noise measured in the neighborhood 
$$V(z_l)=\left\{ z\in \mathbb{R}^{+}\times \mathbb{R}^{d}: \|z-z_l\|=\max\{|t-t_l|,\|u-u_l\|\}\leq \delta_n\right\},$$
where $\delta_n$ is such that the volume of $V(z_l)$ is $1/n$; i.e.,
\begin{eqnarray*}
 \lambda(V(z_l)) &=&\int_{\mathbb{R}^+} \int_{\mathbb{R}^d} \car_{\{|t-t_l|\leq \delta_n\}} \car_{\{\|u-u_l\|\leq \delta_n\}} dt du \\
 &=& 2\delta_n(\delta_n)^d\lambda(\mathcal{S}^{d-1})=\frac{1}{n},
\end{eqnarray*}
with $\lambda(\mathcal{S}^{d-1})$ the volume of a d-dimensional hypersphere of unit radius.

We can rewrite $V(z_l)=(t_l^-,t_l^+]\times V(u_l)$, where $t_l^{\pm}=t_l\pm \delta_n$ and $V(u_l)= \{u\in \mathbb{R}^d : \|u-u_l\|\leq \delta_n\}$. Then,
\begin{equation}\label{def_noise}
    \epsilon_l =\Delta W^H(z_l):= W^H(V(z_l))= W^H_{t_l^+}(V(u_l)) - W^H_{t_l^-}(V(u_l)). 
\end{equation}

\begin{remark}\label{grilla}
In this paper we consider the discrete grid in $\mathbb{R}^{d}$  of distance $\delta_{n}$, i.e., each $u_{i}$  in the grid has $2^{d}$ neighbors that are at distance less than or equal to $\delta_{n}$.
\end{remark}

Next, we show an important result related to the covariance of the noise 

\begin{lemma}\label{cov_noise} The covariance function of fractional colored noise $\epsilon=(\epsilon_l)_{l=1:n}$ is given by
\begin{eqnarray}
\mathbb{E}(\epsilon_l \epsilon_{l'})
& =&\frac{1}{2}\left( \int_{t^{-}_l}^{t_l^+} \int_{t_{l'}^-}^{t_{l'}^+} 2H(2H-1)|t-t'|^{2H-2} dt' dt \right) \nonumber \\
& \times &\left(\int_{\Rd} \int_{\mathbb{R}^d} \car_{V(u_l)}(u) \gamma_{\alpha,d}\|u-v\|^{-d+\alpha} \car_{V(u_{l'})}(v) du dv\right) \label{cov_epsilon}
\end{eqnarray}

and the variance is 
$\mathbb{E}(\epsilon_l^2)  = \sigma^2 2^{2H} (\delta_n)^{2H+d+\alpha}$,
with
$\sigma^2=Var\left(W^H\left(\car_{ \{0\leq t \leq 1, \|u\|\leq 1\}}\right)\right)$.
\end{lemma}

\begin{proof}
The proof of Lemma \ref{cov_noise} is left in Appendix, Section \ref{ap-cov-noise}.
\end{proof}

\subsection{Correlated covariates}

We assume that the covariates $X_j$, for $j=1,\ldots,p$, of the regression model are centered and locally correlated. The covariance function is $\chi(z,z')=\left(\chi_{jk}(z,z')\right)_{j,k=1:p}$, with
\begin{equation}\label{cov_chi}
    \chi_{jk}(z,z')=\E\left(X_j(z)X_k(z')\right).
\end{equation}

\begin{assumption}{C}{1}\label{C1}
\begin{enumerate}
\item[i)] The covariance function $\chi$ is positive definite.
\item[ii)] $\chi$ is $\alpha_{\chi}-$Hölder continuous; i.e. there exist $C_{\chi}>0$ such that
     $$|\chi(z,z') - \chi(z_l,z_{l'})|\leq C_{\chi}(\|z-z_l\|+\|z'-z_{l'}\|)^{\alpha_{\chi}}.$$
\end{enumerate}
\end{assumption}

We consider the covariance function $\Gamma(z,z')=\left(\Gamma_{jk}(z,z')\right)_{j,k=1:n}$ defined by

\begin{equation}\label{cov_gamma}
    \Gamma_{jk}(z,z')= Cov\left(X_{j}(z)X_{k}(z), X_{j}(z')X_{k}(z')\right).
\end{equation}

We suppose that $\Gamma$ satisfy the following conditions.
\begin{assumption}{C}{2}\label{C2} 
\begin{enumerate}
\item[i)] The covariance function $\Gamma$ is positive definite.
\item[ii)] $\Gamma$ is $\alpha_{\Gamma}-$Hölder continue; i.e. there exist $C_{\Gamma}>0$ such that
     $$|\Gamma(z,z') - \Gamma(z_l,z_{l'})|\leq C_{\Gamma}(\|z-z_l\|+\|z'-z_{l'}\|)^{\alpha_{\Gamma}}.$$
\item[iii)] Furthermore, $\Gamma$ is such that 
$$\left|\Gamma(z,z')\right| \leq C_{k,d,\theta}\delta^{d+1+\theta},$$ for $z,z'\in\mathbb{R}^{d+1}$ such that $\|z-z'\|>k\delta$, for $\delta >0$ and some $k\in\mathbb{N}$, $-(d+1)<\theta<d+1$, and $C_{k,d,\theta}\geq 0$.
\item[iv)] $\Gamma$ is such that 
$$\left|\Gamma(z,z)\right| \leq C_{k,d},$$ for $z\in\mathbb{R}^{d+1}$ and  some $k\in\mathbb{N}$.
\end{enumerate}
\end{assumption}

\begin{remark}
Under assumption \ref{C1} we consider two interaction types between covariates $X_j's$:
\begin{description}
    \item[-] \textbf{Weak interaction:} when the parameter $\theta\leq 0$. For instance, the independent case is obtained for $\theta=0$, the $k$-dependent covariates case correspond to $C_{k,d,\theta}=0$.    
    \item[-] \textbf{Strong interaction:} when the parameter $\theta>0$, then the spectral density of covariance function $\Gamma$ is singular at zero, so $\Gamma$ has heavy tails. The fractional time dependence correspond to $\theta=2H-1$ and this has long-range dependence when $\theta>0$ i.e. if $H>\frac{1}{2}$. The fractional-colored spatial-temporal dependence corresponds to $\theta=2H-1+\alpha$.
\end{description}
\end{remark}

\subsection{The weighted least square estimator}

For a given data set, the local parameters of weighted regression model \eqref{model} are estimated using the weighted least square procedure.
Let be $\beta(z_i)$ the vector of the local parameters for the space-time point $z_i$,
\begin{equation}\label{parameters}
    \beta(z_i)=\left( \beta_0(z_i), \beta_1(z_i), \ldots, \beta_p(z_i)\right)^T.
\end{equation}
Here, the superscript $T$ represents the transpose of a vector or matrix.

The local parameters $\beta(z_i)$ at point $z_i$ is estimated by
\begin{equation}\label{estimator}
    \hat{\beta}(z_i)=[X^T\W(z_i)X]^{-1}X^T\W(z_i)Y,
\end{equation}
where $X$ is the $n\times (p+1)$ matrix of input covariables, $Y$ is the $n$-dimensional vector of output observed variable, and $\W(z_i)$ is an $n\times n$ weighting matrix of the form
\begin{equation}\label{weights_matrix}
    \W(z_i)=diag(\W_{i1}, \ldots, \W_{in}).
\end{equation}
The weights $\W_{ij}$, for $j = 1, \ldots, n$, are obtained through an adaptive kernel function $K$ in terms of the proximity of each data point to the point $z_i$;
i.e.
\begin{equation}\label{weights}
    \W_{il}= K_h\left(z_l-z_i\right),
\end{equation}
with $K_h(z)=K(z/h)$. Here, $K:\mathbb{R}^{d+1} \rightarrow \mathbb{R}$ is  positive, symmetric such that $\int_{\mathbb{R}^{d+1}}K(z)dz = 1$, and  $h$ is nonnegative parameter known as bandwidth, which produces a decay of influence with distance. The observations $z_l$ near $z_i$ have the
largest influence on the estimate of the local parameters at point $z_i$. 

We suppose that the kernel $K$ satisfies the additional following conditions:

\begin{assumption}{K}{1}\label{K1} 
\ \\
\begin{enumerate}
\item[i)] The kernel $K$ is bounded, i.e. $\|K\|_{\infty} <\infty$.  
\item[ii)] $K$ is $\alpha_K$-Hölder continuous, i.e. there exist $C_K>0$ such that 
     $$|K(z)-K(z')| < C_K \|z-z'\|^{\alpha_K}.$$
\item[iii)] $\int_{\mathbb{R}^{d+1}} \max{\left(\|z\|^{\alpha_{K}}, \|z\|^{\alpha_{\chi}}, \|z\|^{\alpha_{\Gamma}} \right)}  K(z) dz <\infty$.
\item[iv)]  If  $\Vert z \Vert  \geq \delta_n$, then  $$K(z) =  f_{K}\left( \Vert z \Vert \right)= \mathcal{O}(n^{-\gamma} L(n) ),$$
with L a slowly varying function at infinity and $\gamma >0$.
\end{enumerate}
\end{assumption}

Under condition \ref{K1}, the kernel $K$ is such that $\int_{\mathbb{R}^{d+1}} z K(z)dz = 0$. 

The most commonly used adaptive kernel is the Gaussian function $K(z)=\frac{1}{\sqrt{2\pi}}e^{-(d^{s,t}(z))^2/2}$, where the space-time  distance $d^{s,t}$  is given as a function of the temporal distance $d^t=|t|$ and the spatial distance $d^u=\|u\|$; for instance, $(d^{s,t}(z))^2=\mu^t(d^t)^2+\mu^s(d^u)^2$ where $\mu^t$ and $\mu^s$ are temporal and spatial scale factors respectively. 

\section{Consistency}
We study the consistency  for the local weighted least square estimator $\hat{\beta}(z_{i})$ obtained in \eqref{estimator} from \eqref{model}. If we substitute $Y=X\beta(z_i)+\epsilon$ on \eqref{estimator} we obtained that:  
\begin{eqnarray*}
\hat{\beta}(z_{i})&=&(X^{T}\W(z_i)X)^{-1}X^{T}\W(z_i)(X\beta(z_{i})+\epsilon)\\
&=& (X^{T}\W(z_i)X)^{-1}(X^{T}\W(z_i)X)\beta(z_{i})+(X^{T}\W(z_i)X)^{-1}X^{T}\W(z_i)\epsilon\\
&=& \beta(z_{i})+(X^{T}\W(z_i)X)^{-1} X^{T}\W(z_i)\epsilon.
 \end{eqnarray*}
Then,
\begin{eqnarray*}
\E\left(\hat{\beta}(z_{i}) \right) &=&\beta(z_{i})+ \E\left( X^{T}\W(z_i)X)^{-1} X^{T}\W(z_i) \E(\epsilon | X) \right) = \beta(z_{i}),
 \end{eqnarray*}
since from assumption \ref{M1}  we have  $\E(\epsilon | X)= \E(\epsilon ) =0 $. Thus, the estimator $\hat{\beta}(z_{i})$ is unbiased, and the estimation error is written as:
\begin{equation}\label{residual}
 \hat{\beta}(z_{i}) -\beta(z_{i}) =(X^{T}\W(z_i)X)^{-1}(X^{T}\W(z_i)\epsilon).
\end{equation}

\begin{remark} \label{defis}
We define the following notation
\begin{enumerate}
\item $f_{n,h}\approx\tilde{f} _h$, which is equivalent to $\lim_{h\to 0}\lim_{n\to \infty} f_{n,h} = \lim_{h\to 0}\tilde{f}_h$, i.e. for $n$ large enough  and $h$ small enough,  $f_{n,h}$ is approximately equal to $\tilde{f}_h$. \label{def1}

\item $f_{n,h} \preceq \tilde{f}_h$, which is equivalent to $\lim_{h\to 0} \lim_{n\to \infty} f_{n,h} \leq \lim_{h\to 0} \tilde{f}_h$. Particularly, we write
$f_{n,h} \preceq C$ to state that $C$ is a bound for the sequence $f_{n,h}$, for $n$ large enough and $h$ small enough. \label{def2}

\item $f_{h}\approx\tilde{f}$, which is equivalent to $\lim_{h\to 0} f_{h} = \tilde{f}$, $f_{n}\approx\tilde{f}$ when $\lim_{n\to \infty} f_{n} = \tilde{f}$, and
 $\tilde{f}_h \preceq C$ to state that $C$ is a bound for the sequence $\tilde{f}_{h}$, for $h $small enough. \label{def3}
\end{enumerate}
This notation will be used along our work.
\end{remark}
 
In order to study the consistency of the estimator $\hat{\beta}(z_{i})$ given by \eqref{estimator}, we will prove that there exists an appropriated normalization sequence $(b_{n,h})_{n\geq1,h>0}$ of positive constants with $b_{n,h} \to \infty$ as $n\to \infty$ and $h\to 0$, and such that   
  
\begin{enumerate}
\item[i)] $b_{n,h}^{-1}(X^{T}\W(z_i)X) \rightarrow \chi(z_i,z_i)=\mathbb{E}[X^{T}\W(z_i)X]$, as $n \to +\infty$ and $h\to 0$.
\item [ii)] $b_{n,h}^{-1}(X^{T}\W(z_{i})\epsilon)\rightarrow 0$, $n \to +\infty$ and $h\to 0$.
\end{enumerate}

To prove $i)$ we need an auxiliary lemma related to the almost sure convergence of the term  $(X^{T}\W(z_i)X)$ in \eqref{residual}. 

\begin{lemma}\label{conv_denom}
Under assumptions \ref{C1}-\ref{C2} and \ref{K1}, $\theta>0$ and $\gamma>\frac{\theta}{d+1}$, we have that
\begin{equation*}
 \frac{1}{nh^{d+1}}(X^{T}\W(z_i)X) \xrightarrow[n \to \infty]{a.s.}  \chi(z_i,z_i)=\mathbb{E}[X^{T}\W(z_i)X].
\end{equation*}
\end{lemma}
\begin{proof}
The proof of  Lemma \ref{conv_denom} is left in Appendix, Section \ref{ap-conv-denom}.
\end{proof}
We are ready to present our main result.

\begin{theorem} Assume that the regression model \eqref{model} satisfies the hypothesis \ref{M1}, \ref{N1}, \ref{N2}, \ref{C1}-\ref{C2}and  \ref{K1}. Then, the local weighted least square estimator $\hat{\beta}(z_{i})$ obtained in \eqref{estimator} is 
strongly consistent for $2H+\alpha > 1$, $\theta >0$, and $\frac{\theta}{d+1}<\gamma<1+\frac{\theta}{d+1}$ that is
\begin{equation*}
\hat{\beta}(z_i) - \beta(z_i) \xrightarrow[n \to \infty]{a.s.}  0
\end{equation*}
and, for $2H+d+\alpha >0$ and $d+1+\theta>0$ the convergence in probability is ensured.
\end{theorem}

\begin{proof}
By Lemma \ref{conv_denom}, it remains to study the asymptotic behavior of $ \left( X^{T}\W(z_{i})\epsilon \right)$ as $n\to\infty$.
The $j_{th}$ component of $ \left( X^{T}\W(z_{i})\epsilon \right)$ is 
\begin{equation}
\left( \left( X^{T}\W(z_{i})\epsilon \right) \right)_j =  \sum_{l=1}^n X_{lj}\W_{il}\epsilon_l.
\end{equation}
It is quite easy to see, from assumption \ref{M1}, that $\E\left[  \left( \left( X^{T}\W(z_{i})\epsilon \right) \right)_j  \right]  = 0$.
Let us compute the variance of $\left( \left( X^{T}\W(z_{i})\epsilon \right) \right)_j $,
\begin{eqnarray}\label{split_var}
\E\left(\left( \left( X^{T}\W(z_{i})\epsilon \right) \right)_j^2\right)  &=&  \sum_{l,l'=1}^n \chi_{jj}(z_l,z_{l'})\W_{il}\W_{il'}\E(\epsilon_l \epsilon_{l'} ) \nonumber\\
    &=&  \sum_{l=1}^n \chi_{jj}(z_l,z_l)\W_{il}^2\E(\epsilon_l^2) \nonumber\\
    &+&   \sum_{\substack{1\leq l\neq l'\leq n \\ \|z_l-z_{l'}\|\leq 3\delta_n}} \chi_{jj}(z_l,z_{l'})\W_{il}\W_{il'}\E(\epsilon_l \epsilon_{l'} ) \nonumber\\ 
    &+&  \sum_{\substack{1\leq l\neq l'\leq n \\ \|z_l-z_{l'}\| > 3\delta_n}} \chi_{jj}(z_l,z_{l'})\W_{il}\W_{il'}\E(\epsilon_l \epsilon_{l'} ) \nonumber\\
    &:=& A^{(1)}_{j,n}(z_i) + A^{(2)}_{j,n}(z_i)  + A^{(3)}_{j,n}(z_i),
\end{eqnarray}
where we split the sum into three terms associated with the distance between the observed points $z_l$ and $z_{l'}$. 

First, we study the term $A^{(1)}_{j,n}(z_i) $ in \eqref{split_var}
\begin{equation}\label{A1}
\begin{split}
\lefteqn{ \frac{1}{nh^{d+1} (\delta_n)^{2H+d+\alpha}} A^{(1)}_{j,n}(z_i) }\\
& =  \frac{1}{nh^{d+1}(\delta_n)^{2H+d+\alpha}}\sum_{l=1}^n \chi_{jj}(z_l,z_l)\W_{il}^2\E(\epsilon_l^2)\\ 
& =   \frac{ 2^{2H}\sigma^2}{nh^{d+1}} \sum_{l=1}^n \chi_{jj}(z_l,z_l)K^2_h\left(z_l-z_i\right)\\
& =  2^{2H} \sigma^2 \int_{\mathbb{R}^{d+1}} \sum_{l=1}^n \chi_{jj}(z_l,z_l)\frac{1}{h^{d+1}} K^2_h\left(z_l-z_i\right) \car_{V(z_l)}(z) dz\\
& \preceq  2^{2H} \sigma^2 \int_{\mathbb{R}^{d+1}} \chi_{jj}(z,z)\frac{1}{h^{d+1}}K^2_h\left(z-z_i\right) dz\\
& =   2^{2H} \sigma^2 \int_{\mathbb{R}^{d+1}} \chi_{jj}(z_i+hz,z_i+hz)K^2\left(z\right) dz\\
& \approx \, C_1(H)\chi_{jj}(z_i,z_i) + \mathcal{O}\left(|h|^{\alpha_{\chi}}\right),
\end{split}
\end{equation}
where $C_1(H)= 2^{2H}\sigma^2 \|K\|^2_2$. The last inequality comes from the regularity of $\chi$ from Condition \ref{C1} and notations defined in Remark \ref{defis}.  

Secondly, we consider the term $A^{(2)}_{j,n}(z_i)$ in \eqref{split_var}, i.e. when $0<\|z_l-z_{l'}\|\leq 3\delta_n$

\begin{equation}\label{A2}
\begin{split}
\lefteqn{ \frac{1}{nh^{d+1} (\delta_n)^{2H+d+\alpha}} A^{(2)}_{j,n}(z_i) }\\
 & =   \frac{1}{nh^{d+1} (\delta_n)^{2H+d+\alpha}} \sum_{\substack{1\leq l\neq l'\leq n \\ \|z_l-z_{l'}\|\leq 3\delta_n}} \chi_{jj}(z_l,z_{l'})\W_{il}\W_{il'}\E(\epsilon_l \epsilon_{l'} ) \\
 & =   \frac{1}{nh^{d+1} (\delta_n)^{2H+d+\alpha}} \sum_{\substack{1\leq l\neq l'\leq n \\ \|z_l-z_{l'}\|\leq 3\delta_n}} \chi_{jj}(z_l,z_{l'})K_h(z_l-z_i)K_h(z_{l'}-z_i)\E(\epsilon_l \epsilon_{l'} ) \\
\end{split}
\end{equation}

We can bound the covariance term $\E(\epsilon_l \epsilon_{l'} )$ when $0< \|z_l-z_{l'}\|\leq 3\delta_n$ by
\begin{equation}\label{bound_cov_epsilon_c2}
\begin{split}
\E(\epsilon_l \epsilon_{l'} ) & = 
 \frac{1}{2}\left(|t_l-t_{l'}+2\delta_n|^{2H} + |t_l-t_{l'}-2\delta_n|^{2H}- 2|t_l-t_{l'}|^{2H}\right) (\delta_n)^{d+\alpha}\\
 & \times Cov\left( W^{H} \left(\car_{\{\|u-u_l/\delta_n\|\leq 1\}} \right), W^{H} \left( \car_{\{\|u-u_{l'}/\delta_n\|\leq 1\}} \right) \right)\\
 & \leq \left(2\delta_n\right)^{2H}(\delta_n)^{d+\alpha} Var^{1/2} \left( W^{H} \left( \car_{\|u-u_l/\delta_n\|\leq 1} \right) \right) Var^{1/2} \left( W^{H} \left( \car_{\|u-u_{l'}/\delta_n|\leq 1} \right) \right) \\
  & \leq 2^{2H}(\delta_n)^{2H+d+\alpha} Var \left( W^{H} \left( \car_{\|u\|\leq 1} \right) \right)\\
  &= 2^{2H}\sigma^2 (\delta_n)^{2H+d+\alpha}.
\end{split}
\end{equation}

Plugging inequality \eqref{bound_cov_epsilon_c2} into the equation \eqref{A2} yields 
\begin{equation}\label{A2_0}
\begin{split}
\lefteqn{ \frac{1}{nh^{d+1} (\delta_n)^{2H+d+\alpha}} A^{(2)}_{j,n}(z_i) }\\
 & \leq  \frac{ 2^{2H}\sigma^2}{nh^{d+1}}   \sum_{\substack{1\leq l\neq l'\leq n \\ \|z_l-z_{l'}\|\leq 3\delta_n}} \chi_{jj}(z_l,z_{l'})K_h(z_l-z_i)K_h(z_{l'}-z_i)\\
 & = \frac{ 2^{2H}\sigma^2 }{nh^{d+1}}  \left[ \sum_{\substack{1\leq l\neq l'\leq n \\ \|z_l-z_{l'}\|\leq 3\delta_n}} \chi_{jj}(z_l,z_{l})K_h(z_l-z_i)K_h(z_{l}-z_i)\right.\\
 & + \sum_{\substack{1\leq l\neq l'\leq n \\ \|z_l-z_{l'}\|\leq 3\delta_n}}\chi_{jj}(z_l,z_{l}) K_h(z_l-z_i)\left(K_h(z_{l'}-z_i) - K_h(z_{l}-z_i)\right)\\
& + \sum_{\substack{1\leq l\neq l'\leq n \\ \|z_l-z_{l'}\|\leq 3\delta_n}} \left(\chi_{jj}(z_l,z_{l'})-\chi_{jj}(z_l,z_l)\right)K_h(z_l-z_i)K_h(z_{l}-z_i)\\
&+ \left.\sum_{\substack{1\leq l\neq l'\leq n \\ \|z_l-z_{l'}\|\leq 3\delta_n}} \left(\chi_{jj}(z_l,z_{l'})-\chi_{jj}(z_l,z_l)\right)K_h(z_l-z_i)\left(K_h(z_{l'}-z_i)-K_h(z_{l}-z_i)\right)\right]
\end{split}
\end{equation}

From assumption  \ref{C1} and \ref{A2}
\begin{equation}\label{A2_1}
\begin{split}
\lefteqn{ \frac{1}{nh^{d+1} (\delta_n)^{2H+d+\alpha}} A^{(2)}_{j,n}(z_i) }\\
& \leq \frac{ 2^{2H}\sigma^2 }{nh^{d+1}}  \left[ \sum_{\substack{1\leq l\neq l'\leq n \\ \|z_l-z_{l'}\|\leq 3\delta_n}} \chi_{jj}(z_l,z_{l})K^2_h(z_l-z_i)\right.
 \\
 & +  C_K\sum_{\substack{1\leq l\neq l'\leq n \\ \|z_l-z_{l'}\|\leq 3\delta_n}}\chi_{jj}(z_l,z_{l}) K_h(z_l-z_i)\left\|z_{l} - z_{l'}\right\|^{\alpha_{K}}\\
& + C_{\chi} \sum_{\substack{1\leq l\neq l'\leq n \\ \|z_l-z_{l'}\|\leq 3\delta_n}} \left\|z_l-z_{l'}\right\|^{\alpha_{\chi}}K_h^2(z_l-z_i)\\
& +C_k C_{\chi}\left. \sum_{\substack{1\leq l\neq l'\leq n \\ \|z_l-z_{l'}\|\leq 3\delta_n}} \left\|z_l-z_{l'}\right\|^{\alpha_{\chi}}K_h(z_l-z_i)\left\|z_{l} - z_{l'}\right\|^{\alpha_{K}}\right]\\
&= A^{(2,1)}_{j,n}+A^{(2,2)}_{j,n}+A^{(2,3)}_{j,n}+A^{(2,4)}_{j,n}.
\end{split}
\end{equation}

Note that 
\begin{equation}\label{volume_V_zl}
\frac{1}{n}\sum_{l'=1}^n\car_{\{0<\|z_l-z_{l'}\|\leq 3 \delta_n\}}\quad \appn  \quad 3^{d+1}\int_{\mathbb{R}^{d+1}}\car_{V(z_l)}(z')dz'=3^{d+1}\lambda(V(z_l))=\frac{3^{d+1}}{n},
\end{equation}
then using \eqref{A2_1} and \eqref{volume_V_zl} we have
\begin{equation}\label{A21}
\begin{split}
A^{(2,1)}_{j,n} & = \frac{ 2^{2H}\sigma^2 }{nh^{d+1}}  \sum_{l=1}^n \chi_{jj}(z_l,z_{l})K^2_h(z_l-z_i)\sum_{l'=1}^n\car_{\{0<\|z_l-z_{l'}\|\leq 3\delta_n\}}\\
& \preceq \frac{  2^{2H}3^{d+1}\sigma^2}{h^{d+1}} \int_{\mathbb{R}^{d+1}} \chi_{jj}(z,z)K^2_h(z-z_i)dz\\
&=   2^{2H}3^{d+1} \sigma^2\int_{\mathbb{R}^{d+1}} \chi_{jj}(z_i+hz,z_i+hz)K^2(z)dz\\
& \approx \,  2^{2H}3^{d+1}\sigma^2 \chi_{jj}(z_i,z_i) \|K\|_2^2 + \mathcal{O}(|h|^{\alpha_{\chi}}).
\end{split}
\end{equation}
As before, the last inequality comes from the regularity of $\chi$ from Condition \ref{C1} and notations defined in Remark \ref{defis}. 

\begin{equation}\label{A22}
\begin{split}
\frac{1}{(\delta_n)^{\alpha_k}}A^{(2,2)}_{j,n}& = \frac{2^{2H}\sigma^2  C_K}{nh^{d+1}(\delta_n)^{\alpha_k}}\sum_{l=1}^n \chi_{jj}(z_l,z_{l}) K_h(z_l-z_i)\sum_{l'=1}^n\left\|z_{l} - z_{l'}\right\|^{\alpha_{K}}\car_{\{0<\|z_l-z_{l'}\|\leq 3\delta_n\}}\\
&\leq \frac{  2^{2H}\sigma^2 C_K (3\delta_n)^{\alpha_k}}{nh^{d+1}(\delta_n)^{\alpha_k}}\sum_{l=1}^n \chi_{jj}(z_l,z_{l}) K_h(z_l-z_i)\sum_{l'=1}^n\car_{\{0<\|z_l-z_{l'}\|\leq 3\delta_n\}}\\
& \preceq  \frac{  2^{2H}3^{\alpha_K+d+1}\sigma^2C_K}{h^{d+1}} \int_{\mathbb{R}^{d+1}} \chi_{jj}(z,z) K_h(z-z_i) dz\\
& =  2^{2H}3^{\alpha_K+d+1}\sigma^2C_K\int_{\mathbb{R}^{d+1}} \chi_{jj}(z_i+hz,z_i+hz) K(z) dz\\
& \approx  2^{2H}3^{\alpha_K+d+1}\sigma^2 C_K \chi_{jj}(z_i,z_i) +
\mathcal{O}(|h|^{\alpha_{\chi}}).
\end{split}
\end{equation}
Again, the last inequality comes from the regularity of $\chi$ from Condition \ref{C1} and notations defined in Remark \ref{defis}. 

\begin{equation}\label{A23}
\begin{split}
\frac{1}{(\delta_n)^{\alpha_{\chi}}} A^{(2,3)}_{j,n} & = \frac{  2^{2H}\sigma^2C_{\chi}}{nh^{d+1}(\delta_n)^{\alpha_{\chi}}} \sum_{l=1}^n K_h^2(z_l-z_i) \sum_{l'=1}^n \left\|z_{l} - z_{l'}\right\|^{\alpha_{\chi}}\car_{\{0<\|z_l-z_{l'}\|\leq 3\delta_n\}}\\
&\leq \frac{  2^{2H}\sigma^2C_{\chi}(3\delta_n)^{\alpha_{\chi}} }{nh^{d+1}(\delta_n)^{\alpha_{\chi}}} \sum_{l=1}^n K_h^2(z_l-z_i) \sum_{l'=1}^n \car_{\{0<\|z_l-z_{l'}\|\leq 3\delta_n\}}\\
& \preceq \frac{ 2^{2H}3^{\alpha_{\chi}+d+1}\sigma^2C_{\chi}}{h^{d+1}} \int_{\mathbb{R}^{d+1}} K^2_h(z-z_i) dz\\
&=  2^{2H}3^{\alpha_{\chi}+d+1}\sigma^2C_{\chi} \|K\|_2^2
\end{split}
\end{equation}

\begin{equation}\label{A24}
\begin{split}
\frac{1}{(\delta_n)^{\alpha_{\chi}+\alpha_k}} A^{(2,4)}_{j,n}& = \frac{  2^{2H}\sigma^2C_K C_{\chi}}{nh^{d+1}(\delta_n)^{\alpha_{\chi}+\alpha_k}} \sum_{l=1}^n K_h(z_l-z_i) \sum_{l'=1}^n \left\|z_{l} - z_{l'}\right\|^{\alpha_{\chi}+\alpha_k}\car_{\{0<\|z_l-z_{l'}\|\leq 3\delta_n\}}\\
&\leq \frac{  2^{2H}\sigma^2C_K C_{\chi}(3\delta_n)^{\alpha_{\chi}+\alpha_K} }{nh^{d+1}(\delta_n)^{\alpha_{\chi}+\alpha_K}} \sum_{l=1}^n K_h(z_l-z_i) \sum_{l'=1}^n \car_{\{0<\|z_l-z_{l'}\|\leq 3\delta_n\}}\\
& \preceq \frac{ 2^{2H}3^{\alpha_{\chi}+\alpha_K+d+1}\sigma^2C_K C_{\chi}
}{h^{d+1}} \int_{\mathbb{R}^{d+1}} K_h(z-z_i) dz\\
&=  2^{2H}3^{\alpha_{\chi}+\alpha_K+d+1}\sigma^2C_K C_{\chi}.
\end{split}
\end{equation}

Thus, from \eqref{A21}, \eqref{A22}, \eqref{A23}, and \eqref{A24} we have
\begin{equation}\label{A2_2}
\begin{split}
\lefteqn{\frac{1}{nh^{d+1} (\delta_n)^{2H+d+\alpha}} A^{(2)}_{j,n}(z_i) }\\
&\preceq 2^{2H}3^{d+1}\sigma^2\chi_{jj}(z_i,z_i) \|K\|_2^2 
  + \mathcal{O}\left(|h|^{\alpha_{\chi}}\vee (\delta_n)^{\alpha_K}\vee (\delta_n)^{\alpha_{\chi}}\vee (\delta_n)^{\alpha_{\chi}+\alpha_K}\right)\\
 & = C_2(H)\chi_{jj}(z_i,z_i) + \mathcal{O}\left(|h|^{\alpha_{\chi}}\vee (\delta_n)^{\alpha_K}\vee (\delta_n)^{\alpha_{\chi}}\right),
 \end{split}
\end{equation}
where $C_2(H)= 2^{2H}3^{d+1}\sigma^2 \|K\|_2^2 $.

Finally we consider the case $\|z_l-z_{l'}\|>3\delta_n$, and we split the term $A^{(3)}_{j,n}(z_i)$ in three term:
\begin{equation}\label{A3}
\begin{split}
\lefteqn{ \frac{1}{n^2h^{2(d+1)} (\delta_n)^{2H+d+\alpha}} A^{(3)}_{j,n}(z_i) }\\
 & =  \frac{1}{n^2h^{2(d+1)}(\delta_n)^{2H+d+\alpha}} \sum_{1\leq l\neq l'\leq n} \chi_{jj}(z_l,z_{l'})K_h(z_l-z_i)K_h(z_{l'}-z_i)\\
 & \times \E(\epsilon_l\epsilon_{l'})\car_{\{ |t_l-t_{l'}| \leq 3\delta_n, \|u_l-u_{l'}\|> 3\delta_n\}}\\
 & + \frac{1}{n^2h^{2(d+1)}(\delta_n)^{2H+d+\alpha}} \sum_{1\leq l\neq l'\leq n} \chi_{jj}(z_l,z_{l'})K_h(z_l-z_i)K_h(z_{l'}-z_i)\\
 & \times \E(\epsilon_l\epsilon_{l'})\car_{\{|t_l-t_{l'}|> 3\delta_n, \|u_l-u_{l'}\|\leq 3\delta_n\}}\\
 & +
 \frac{1}{n^2h^{2(d+1)}(\delta_n)^{2H+d+\alpha} } \sum_{1\leq l\neq l'\leq n} \chi_{jj}(z_l,z_{l'})K_h(z_l-z_i)K_h(z_{l'}-z_i)\\
 &\times \E(\epsilon_l\epsilon_{l'})\car_{\{|t_l-t_{l'}|> 3\delta_n, \|u_l-u_{l'}\| > 3\delta_n\}}\\
 &= A^{(3,1)}_{j,n} + A^{(3,2)}_{j,n} + A^{(3,3)}_{j,n}.
\end{split}
\end{equation}

 In the case $|t_l-t_{l'}| \leq 3\delta_n$ and  $\|u_l-u_{l'}\|> 3\delta_n$ we can bond the covariance $\E(\epsilon_l \epsilon_{l'})$  as follows
\begin{equation}\label{bound_cov_epsilon_c3}
\begin{split}
\E(\epsilon_l \epsilon_{l'} ) & = 
 \frac{1}{2}\left(|t_l-t_{l'}+2\delta_n|^{2H} + |t_l-t_{l'}-2\delta_n|^{2H}- 2|t_l-t_{l'}|^{2H}\right)\\
 & \times \int_{V(u_l)}\int_{V(u_{l'})} \gamma_{\alpha,d}\|u-u'\|^{-d+\alpha}dudu'\\
 & \leq \left(2\delta_n\right)^{2H}\gamma_{\alpha,d}(\delta_n)^{-d+\alpha}\lambda(V(u_l))\lambda(V(u_{l'})) \\
 & \leq 2^{2H}\gamma_{\alpha,d} \lambda^2(\mathcal{S}^{d-1})\left(\delta_n\right)^{2H+d+\alpha}.
\end{split}
\end{equation}

Then, from \eqref{A3} and \eqref{bound_cov_epsilon_c3} we have
\begin{equation}\label{A31}
\begin{split}
A^{(3,1)}_{j,n} & = \frac{1}{n^2h^{2(d+1)}(\delta_n)^{2H+d+\alpha}} \sum_{1\leq l\neq l'\leq n} \chi_{jj}(z_l,z_{l'})K_h(z_l-z_i)K_h(z_{l'}-z_i)\\
& \times \E(\epsilon_l\epsilon_{l'})\car_{\{ |t_l-t_{l'}| \leq 3\delta_n, \|u_l-u_{l'}\|> 3\delta_n\}}\\
& \leq \frac{2^{2H}\gamma_{\alpha,d} \lambda^2(S^{d-1}) }{n^2h^{2(d+1)}} \sum_{1\leq l\neq l'\leq n} \chi_{jj}(z_l,z_{l'})K_h(z_l-z_i)K_h(z_{l'}-z_i)\\
& \preceq \frac{2^{2H} \gamma_{\alpha,d}\lambda^2(\mathcal{S}^{d-1}) }{h^{2(d+1)}} \int_{\mathbb{R}^{d+1}} \int_{\mathbb{R}^{d+1}} \chi_{jj}(z,z')K_h(z-z_i)K_h(z'-z_i)dzdz'\\
& = 2^{2H} \gamma_{\alpha,d} \lambda^2(\mathcal{S}^{d-1}) \int_{\mathbb{R}^{d+1}} \int_{\mathbb{R}^{d+1}} \chi_{jj}(z_i+hz,z_i+hz')K(z)K(z')dzdz'\\
& \approx 2^{2H}\gamma_{\alpha,d}\lambda^2(\mathcal{S}^{d-1})\chi_{jj}(z_i,z_i) + 2^{2H}\gamma_{\alpha,d}\lambda^2(\mathcal{S}^{d-1})|h|^{\alpha_{\chi}}\\
& \times \int_{\mathbb{R}^{d+1}} \int_{\mathbb{R}^{d+1}} (\|z\|+\|z'\|)^{\alpha_{\chi}}K(z)K(z')dzdz'\\
&= 2^{2H}\gamma_{\alpha,d} \lambda^2(\mathcal{S}^{d-1})\chi_{jj}(z_i,z_i) + \mathcal{O}(|h|^{\alpha_{\chi}}). 
\end{split}
\end{equation}
Again, the last inequality comes from the regularity of $\chi$ from Condition \ref{C1} and notations defined in Remark \ref{defis}. Now, we study the case $|t_l-t_{l'}|> 3\delta_n$ and  $\|u_l-u_{l'}\|\leq 3\delta_n$. From \eqref{cov_epsilon} and \eqref{cov_colored}  we bond the covariance $\E(\epsilon_l \epsilon_{l'})$  as follows

\begin{equation}\label{bound_cov_epsilon_c4}
\begin{split}
\E(\epsilon_l \epsilon_{l'} ) & = 
 \frac{1}{2}\left(|t_l-t_{l'}+2\delta_n|^{2H} + |t_l-t_{l'}-2\delta_n|^{2H}- 2|t_l-t_{l'}|^{2H}\right) (\delta_n)^{d+\alpha}\\
 & \times Cov\left( W^{H} \left( \car_{\{\|u-u_l/\delta_n\|\leq 1\}} \right), W^{H} \left( \car_{\{\|u-u_{l'}/\delta_n\|\leq 1\}}  \right) \right)\\
 &= \frac{1}{2}\left( \int_{t_l-\delta_n}^{t_l+\delta_n} \int_{t_{l'}-\delta_n}^{t_{l'}+\delta_n} 2H(2H-1)|t-t'|^{2H-2} dt' dt \right)(\delta_n)^{d+\alpha} \\
 & \times Cov\left( W^{H} \left( \car_{\{\|u-u_l/\delta_n\|\leq 1\}} \right), W^{H} \left( \car_{\{\|u-u_{l'}/\delta_n\|\leq 1\}}  \right) \right) \\
 & \leq H(2H-1) (\delta_n)^{2H-2}\left( \int_{t_l-\delta_n}^{t_l+\delta_n} \int_{t_{l'}-\delta_n}^{t_{l'}+\delta_n} dt' dt \right) (\delta_n)^{d+\alpha} \\
 & \times Var^{1/2} \left( W^{H} \left(  \car_{\|u-u_l/\delta_n\|\leq 1} \right) \right) Var^{1/2} \left( W^{H} \left(  \car_{\|u-u_{l'}/\delta_n\|\leq 1} \right) \right)\\
  & \leq  4H(2H-1)\sigma^2(\delta_n)^{2H+d+\alpha} .
\end{split}
\end{equation}

Then, from \eqref{A3} and \eqref{bound_cov_epsilon_c4} we have 
\begin{equation}\label{A32}
\begin{split}
A^{(3,2)}_{j,n} & = \frac{1}{n^2h^{2(d+1)}(\delta_n)^{2H+d+\alpha}} \sum_{1\leq l\neq l'\leq n} \chi_{jj}(z_l,z_{l'})K_h(z_l-z_i)K_h(z_{l'}-z_i)\\
& \times \E(\epsilon_l\epsilon_{l'})\car_{\{ |t_l-t_{l'}| > 3\delta_n, \|u_l-u_{l'}\|\leq 3\delta_n\}}\\
& \leq \frac{4H(2H-1)\sigma^2} 
{n^2h^{2(d+1)}} \sum_{1\leq l\neq l'\leq n} \chi_{jj}(z_l,z_{l'})K_h(z_l-z_i)K_h(z_{l'}-z_i)\\
& \preceq \frac{2\sigma^2 2H(2H-1) }{h^{2(d+1)}} \int_{\mathbb{R}^{d+1}} \int_{\mathbb{R}^{d+1}} \chi_{jj}(z,z')K_h(z-z_i)K_h(z'-z_i)dzdz'\\
& \approx 4H(2H-1)\sigma^2\chi_{jj}(z_i,z_i) + \mathcal{O}(|h|^{\alpha_{\chi}}). 
\end{split}
\end{equation}
Again, the last inequality comes from the regularity of $\chi$ from Condition \ref{C1} and notations defined in Remark \ref{defis}. For the case $|t_l-t_{l'}|> 3\delta_n$ and  $\|u_l-u_{l'}\|> 3\delta_n$, we proceed analogously to the previous cases
\begin{equation}\label{bound_cov_epsilon_c5}
\begin{split}
\E(\epsilon_l \epsilon_{l'} ) & = 
 \frac{1}{2}\left(|t_l-t_{l'}+2\delta_n|^{2H} + |t_l-t_{l'}-2\delta_n|^{2H}- 2|t_l-t_{l'}|^{2H}\right)\\
 & \times \left( \int_{V(u_l)}\int_{V(u_{l'})} \gamma_{\alpha,d}\|u-u'\|^{-d+\alpha}dudu'\right)\\
 &= \frac{1}{2}\left( \int_{t_l-\delta_n}^{t_l+\delta_n} \int_{t_{l'}-\delta_n}^{t_{l'}+\delta_n} 2H(2H-1)|t-t'|^{2H-2} dt' dt \right)\\
 & \times \left( \int_{V(u_l)}\int_{V(u_{l'})} \gamma_{\alpha,d}\|u-u'\|^{-d+\alpha}dudu'\right)\\
 & = 4H(2H-1) (\delta_n)^{2H} \gamma_{\alpha,d}\left(\delta_n\right)^{d+\alpha} \lambda^2(\mathcal{S}^{d-1})\\
  & \leq 4H(2H-1) \gamma_{\alpha,d} \lambda^2(\mathcal{S}^{d-1})(\delta_n)^{2H+d+\alpha}.
\end{split}
\end{equation}

Thus, from \eqref{A3} and \eqref{bound_cov_epsilon_c5} 
\begin{equation}\label{A33}
\begin{split}
A^{(3,3)}_{j,n} & = \frac{1}{n^2h^{2(d+1)}(\delta_n)^{2H+d+\alpha}} \sum_{1\leq l\neq l'\leq n} \chi_{jj}(z_l,z_{l'})K_h(z_l-z_i)K_h(z_{l'}-z_i)\\
& \times \E(\epsilon_l\epsilon_{l'})\car_{\{ |t_l-t_{l'}| > 3\delta_n, \|u_l-u_{l'}\|> 3\delta_n\}}\\
& \leq \frac{4H(2H-1)\gamma_{\alpha,d}\lambda^2(\mathcal{S}^{d-1})} 
{n^2h^{2(d+1)}} \sum_{1\leq l\neq l'\leq n} \chi_{jj}(z_l,z_{l'})K_h(z_l-z_i)K_h(z_{l'}-z_i)\\
& \preceqnh\frac{4H(2H-1)\gamma_{\alpha,d}\lambda^2(\mathcal{S}^{d-1}) }{h^{2(d+1)}} \int_{\mathbb{R}^{d+1}} \int_{\mathbb{R}^{d+1}} \chi_{jj}(z,z')K_h(z-z_i)K_h(z'-z_i)dzdz'\\
& \preceqh 4H(2H-1)\gamma_{\alpha,d}\lambda^2(\mathcal{S}^{d-1})\chi_{jj}(z_i,z_i) + \mathcal{O}(|h|^{\alpha_{\chi}}). 
\end{split}
\end{equation}

Then, from \eqref{A3}, \eqref{A31}, \eqref{A32} and \eqref{A33} we have 

\begin{equation}\label{A3_1}
\begin{split}
\lefteqn{\frac{1}{n^2h^{2(d+1)} (\delta_n)^{2H+d+\alpha}} A^{(3)}_{j,n}(z_i) }\\
&
\preceqnh
\;
 \left(2^{2H}\gamma_{\alpha,d}\lambda^2(\mathcal{S}^{d-1}) + 4H(2H-1)\sigma^2\right) \chi_{jj}(z_i,z_i)\\
 & + \left(4H(2H-1)\gamma_{\alpha,d}\lambda^2(\mathcal{S}^{d-1}) \right) \chi_{jj}(z_i,z_i)  + \mathcal{O}\left(|h|^{\alpha_{\chi}}\right)\\
  & \preceqh C_3(H) + \mathcal{O}\left(|h|^{\alpha_{\chi}}\right),
 \end{split}
\end{equation}
where $ C_3(H)=2^{2H}\gamma_{\alpha,d}\lambda^2(\mathcal{S}^{d-1}) + 4H(2H-1)\sigma^2+ 4H(2H-1)\gamma_{\alpha,d}\lambda^2(\mathcal{S}^{d-1})$.
Substituting \eqref{A1}, \eqref{A2_2} and \eqref{A3_1} into the equation \eqref{split_var}, and using that $2\lambda(\mathcal{S}^{d-1})(\delta_n)^{d+1}=1/n$ we obtain
\begin{eqnarray*}
\frac{1}{n^2h^{2(d+1)} }\E\left(\left(X^{T}\W(z_{i})\epsilon\right)_j^2\right)  & \preceq & \frac{\left(C_1(H)+C_2(H)\right)(\delta_n)^{2H+d+\alpha}}{nh^{d+1}} + C_3(H)(\delta_n)^{2H+d+\alpha} \\
&\approx & \frac{C(H)}{n^{1+\nu'}},
\end{eqnarray*}
where $\nu' = \frac{2H+\alpha-1}{d+1}>0$ if $2H+\alpha-1 >0$, and also we should consider $\gamma<1+\frac{\theta}{d+1}$ to obtain $nh^{d+1} \to \infty$. Thus, the convergence in $L^2$, and therefore in probability, is ensured for $2H+d+\alpha>0$.
For $2H+\alpha-1>0$,  the $L^2$ rate of $\frac{1}{nh^{d+1}}\left(X^{T}\W(z_{i})\epsilon\right)_j$ is faster than $1/n$; for instance, when $H>\frac{1}{2}$ and $\alpha\geq 0$. A direct application of Borell-Cantelli lemma allow us to obtain 
$$\frac{1}{nh^{d+1}}\left(X^{T}\W(z_{i})\epsilon\right)_j\xrightarrow[n \to \infty]{a.s.}  0.$$

By Slutsky Theorem and Lemma \ref{conv_denom}, the convergence of $|\hat{\beta}(z_i)- \beta(z_i)| \to 0$ is:

\begin{itemize}
    \item In probability, for $2H+d+\alpha>0$ and $d+1+\theta>0$. In particular for 
$H>0$, $\alpha>-d$, and $\theta>-(d+1)$.
\item Almost surely, for $2H+\alpha-1>0$, $\theta>0$ y $\frac{\theta}{d+1}>\gamma>1+\frac{\theta}{d+1}$. In particular condition $2H+\alpha-1>0$ hold for $H>1/2$ and $\alpha\geq 0$.
\end{itemize}
\end{proof}

\section{Simulation study}
This section reviews the theoretical results presented in the previous sections, this part of the work was performed using the software R. To represent the spatial location, we considered a grid defined on $\mathbb{R}^{2} \in [0,1]$; as points we defined the center of each pixel, i.e., ordered pairs defined by $(0.05, 0.05), \dots, (0.95, 0.95)$, a graphical representation of the locations can be seen in the following figure.

\begin{figure}[h!]
\centering
\includegraphics[width=0.35\textwidth]{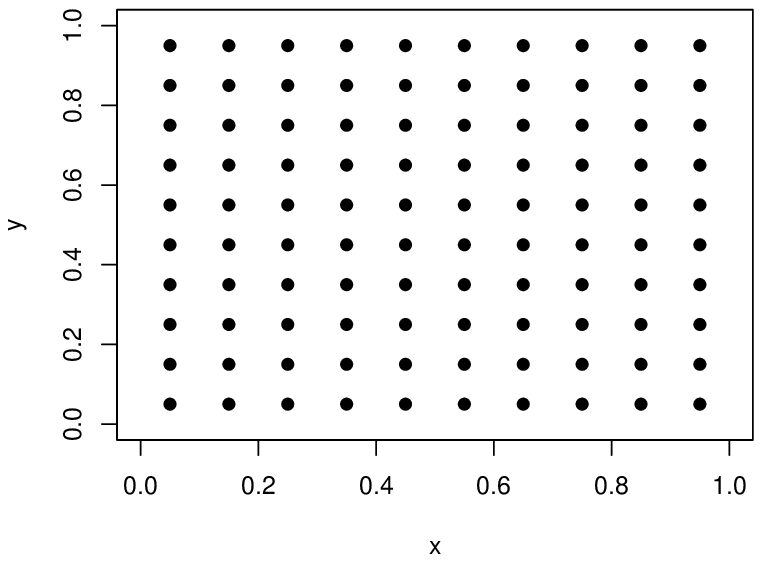}
         \caption{Spatial location.}
         \label{sites}
\end{figure}

We will start by representing the Colored noise in space and time, which represents the noise of our model. Then the four different models studied will be presented, along with the estimation of the response surface to analyze the residuals for the different models considered.

\subsubsection*{Colored noise in space and time}
For the noise, the following values of $H$ were considered: $H_{s}=0.40$ for space, $H_{t}=0.65$ and $H_{t}=0.90$. A representation for different time instants, $t_{1}$, $t_{50}$ and $t_{100}$, is shown in the figure below.
\begin{figure}[h!]
     \centering
     \begin{subfigure}[h]{0.32\textwidth}
         \centering
         \includegraphics[width=1\textwidth]{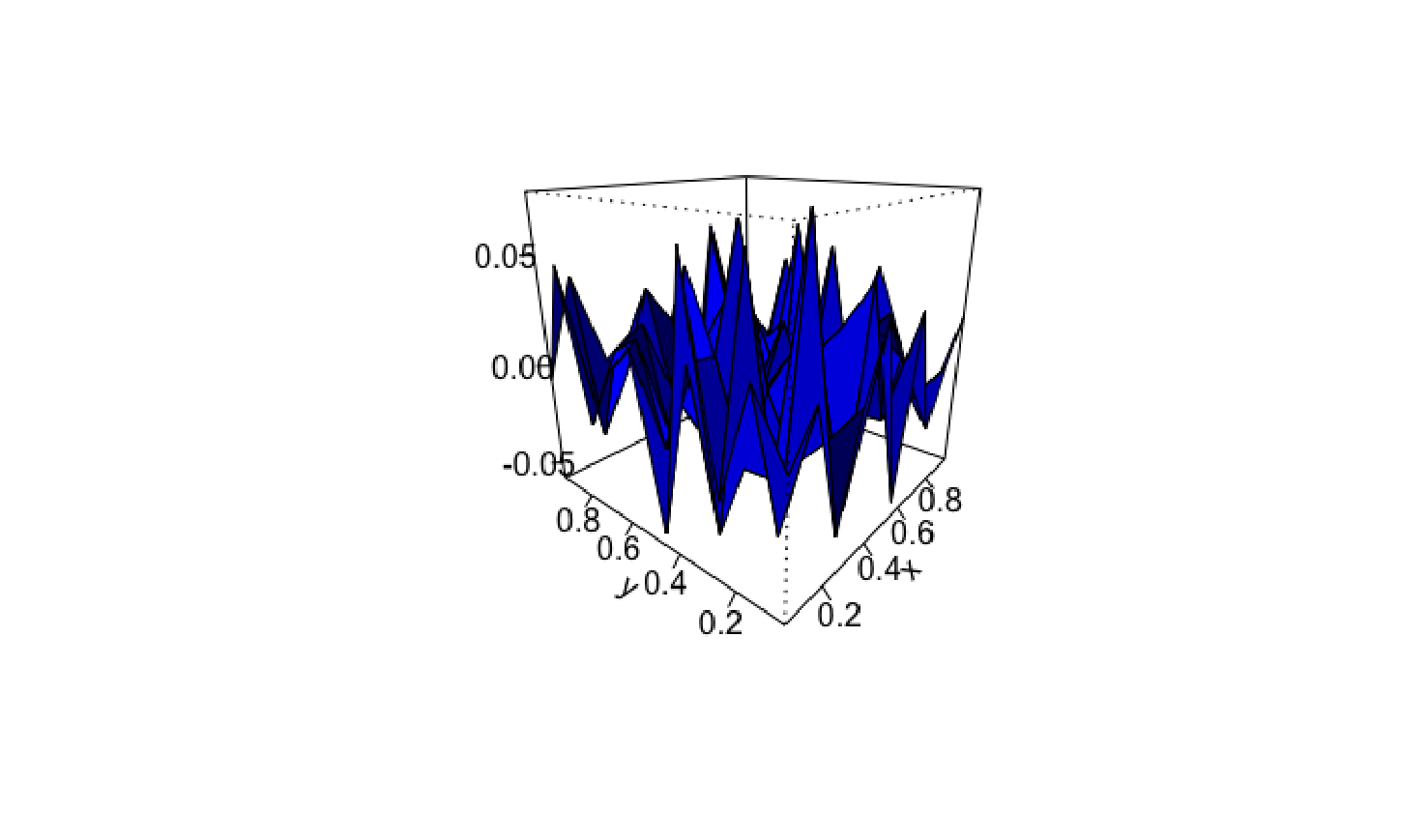}
         \caption{$t_{1}$}
         \label{Ht065t1}
     \end{subfigure}
     \hfill
     \begin{subfigure}[h]{0.32\textwidth}
         \centering
         \includegraphics[width=1\textwidth]{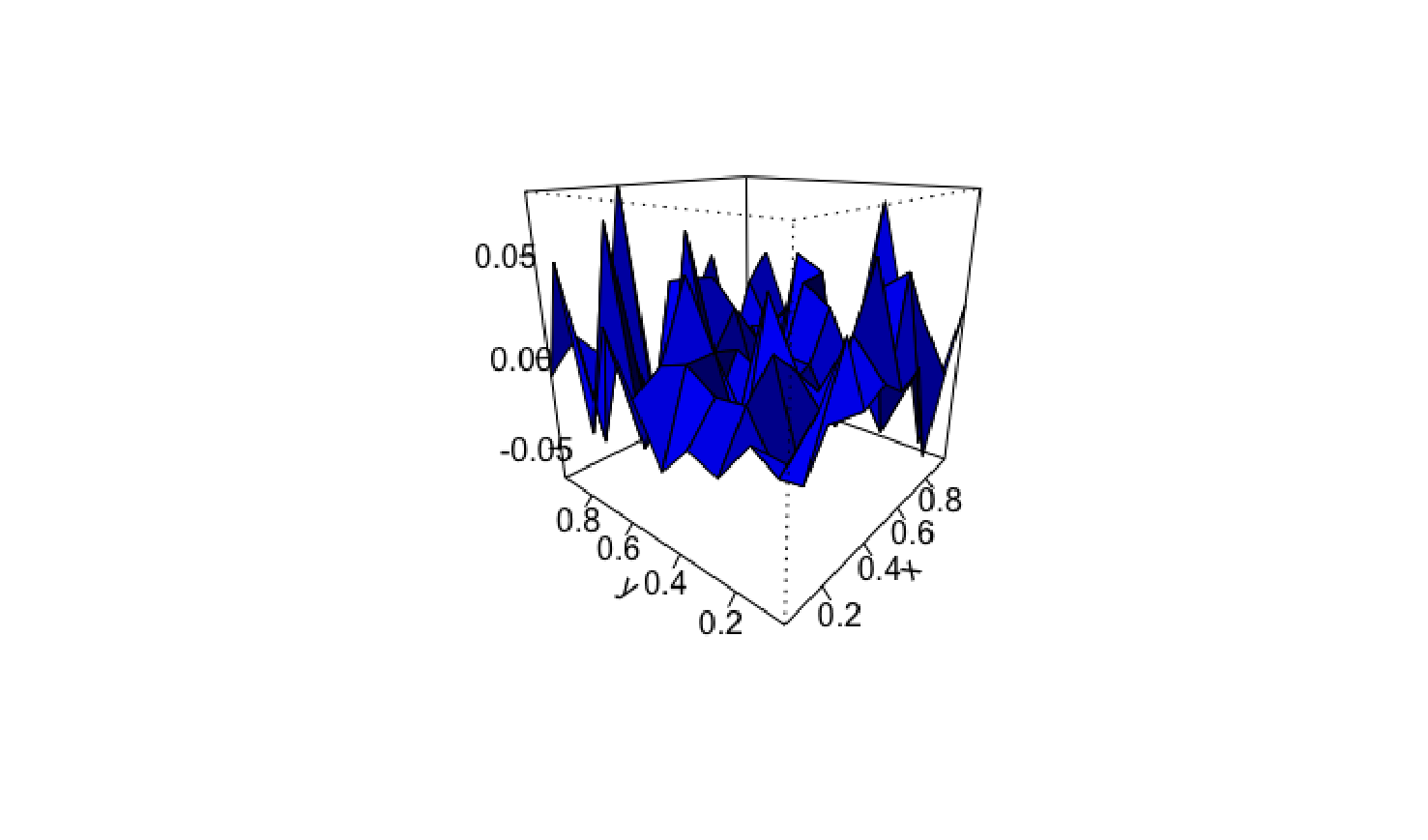}
         \caption{$t_{50}$}
         \label{Ht065t50}
     \end{subfigure}
     \hfill
     \begin{subfigure}[h]{0.32\textwidth}
         \centering
         \includegraphics[width=1\textwidth]{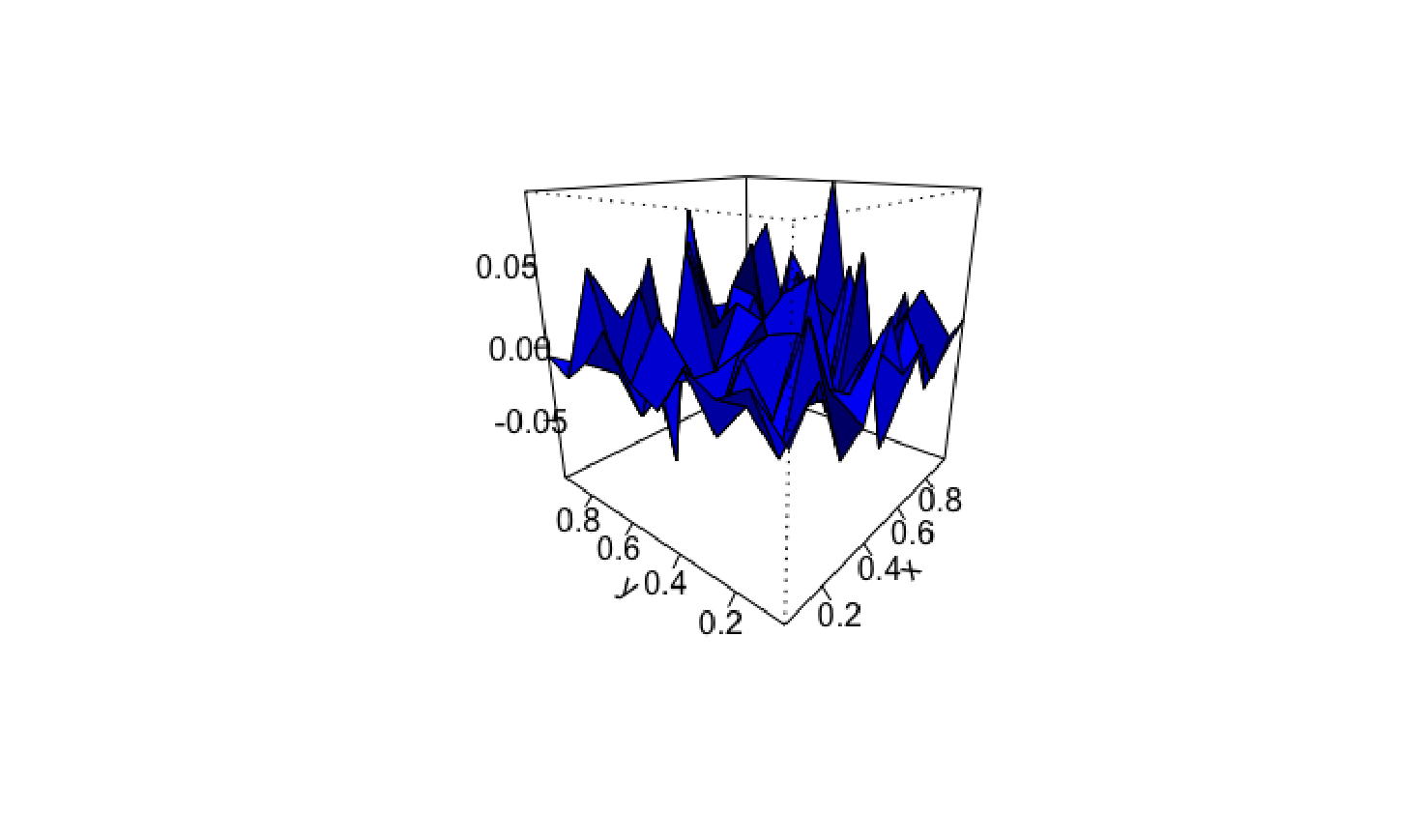}
         \caption{$t_{100}$}
         \label{Ht065t100}
     \end{subfigure}
	\hfill
     \begin{subfigure}[h]{0.32\textwidth}
         \centering
         \includegraphics[width=1\textwidth]{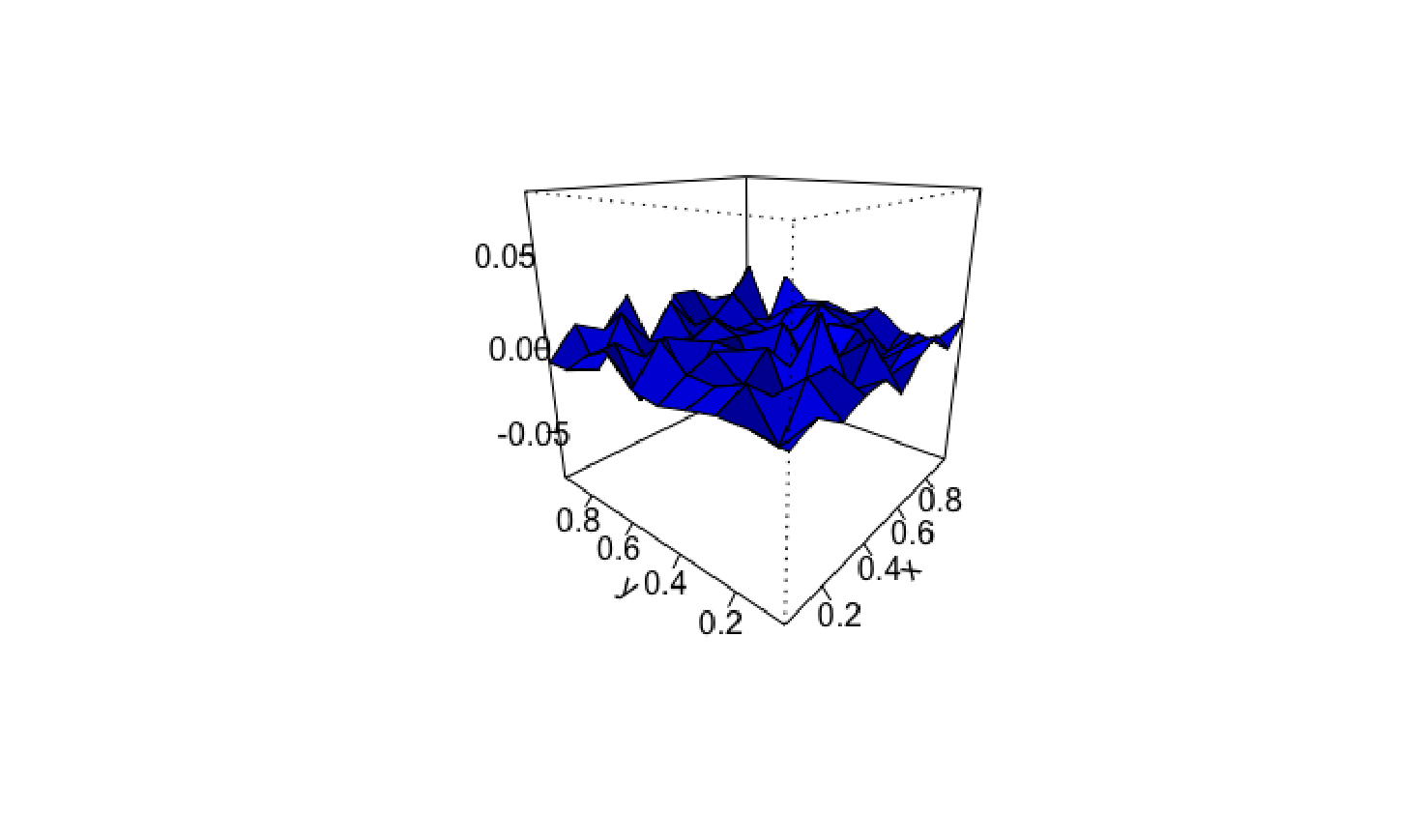}
         \caption{$t_{1}$}
         \label{Ht090t1}
     \end{subfigure}
     \hfill
     \begin{subfigure}[h]{0.32\textwidth}
         \centering
         \includegraphics[width=1\textwidth]{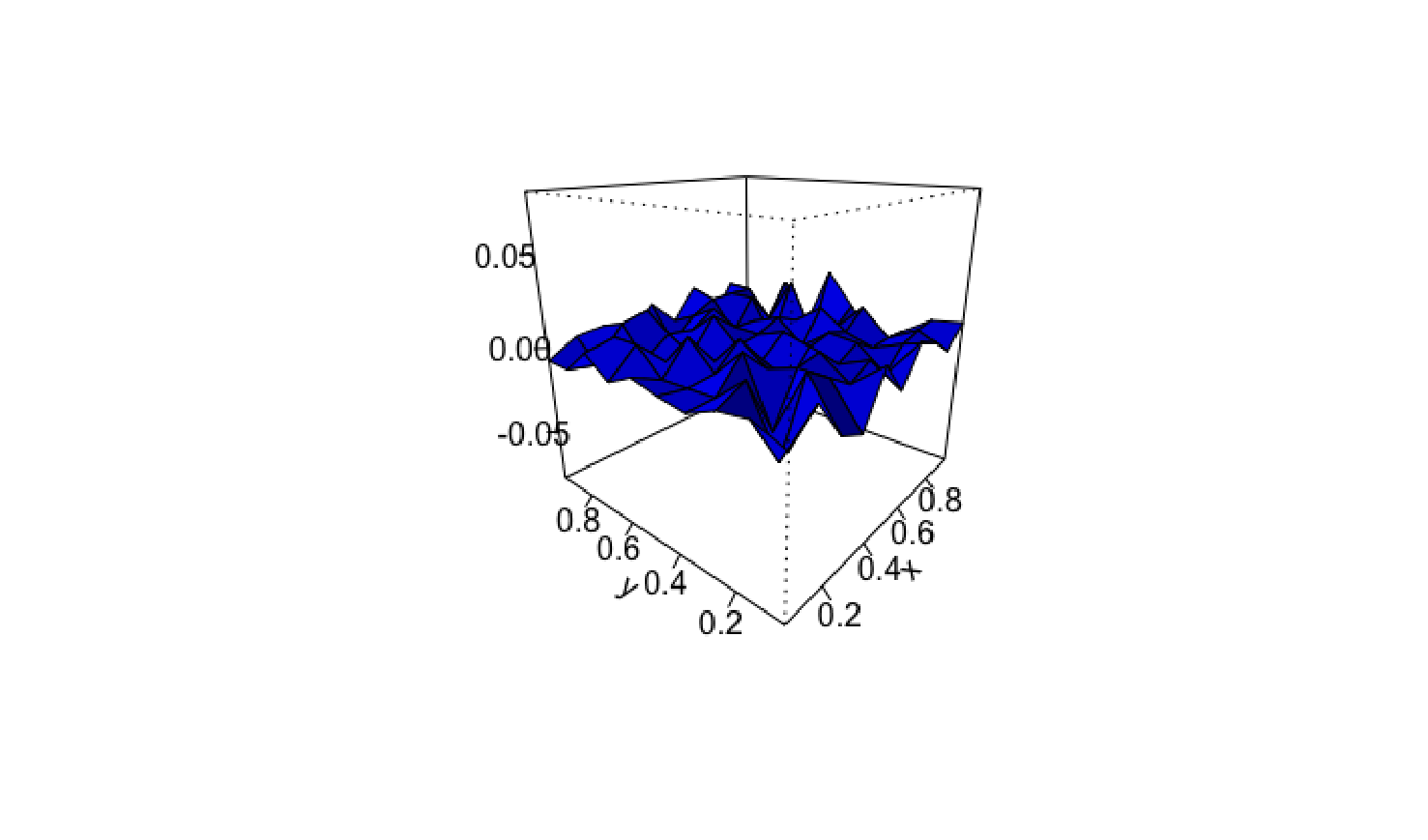}
         \caption{$t_{50}$}
         \label{Ht090t50}
     \end{subfigure}
     \hfill
     \begin{subfigure}[h]{0.32\textwidth}
         \centering
         \includegraphics[width=1\textwidth]{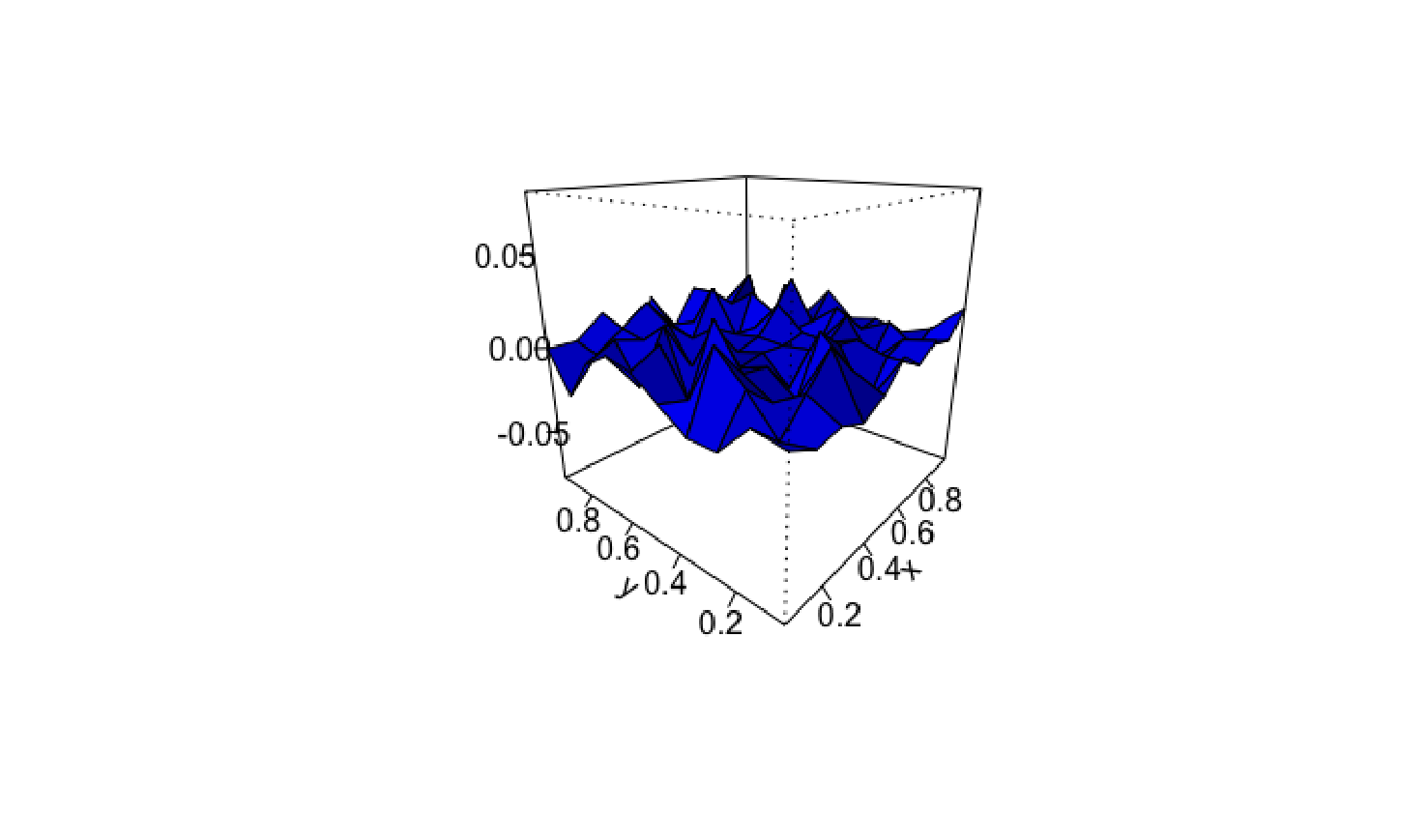}
         \caption{$t_{100}$}
         \label{Ht090t100}
     \end{subfigure}
        \caption{Colored noise, in space with $H_{s}=0.40$ and time with $H_{t}=0.65$ (a, b, c) and $H_{t}=0.90$  (d, e, f) in three different instants of time.}
        \label{colnoiseH04}
\end{figure}

It is possible to appreciate in Figure \ref{colnoiseH04} that for values of $H_{t}$ close to 1 (\ref{Ht090t1}, \ref{Ht090t50} and \ref{Ht090t100}) a lower roughness is observed on the simulated surface. In contrast to the case when the value of $H_{t}$ is close to 0.5 ((\ref{Ht065t1}, \ref{Ht065t50} and \ref{Ht065t100})). This is a common behavior of fractional Brownian motion.

\subsubsection*{Model}
We simulate four different versions of the model presented in \eqref{model}, where we consider that the intercept $\beta_{0}$ is zero, and a covariate represented by a Spatio - Temporal Auto Regressive Moving Average (STARMA) sampled at  the sites defined in \ref{sites}. To represent the covariates, we have decided to use the STARMA models since they have attracted great interest due to their flexibility to represent the relationship between observation sites and their neighbors; some of the research areas where the relevance of these models can be appreciated are renewable energies \cite{dambreville2014, zou2018}, environmental data \cite{deluna2005}, disease mapping \cite{martinez2008}, and regional studies \cite{ramajo2017}, among others. (for a detailed revision of STARMA we recommend to review \cite{pfei1980}, and to simulate this process we recommend the R package STARMA \cite{chey2015}). It is important to note that in the model \eqref{model}, the values of the coefficients $\beta$ accompanying the covariate, also depend on the location of the points.  As examples of different situations, we consider the values of $\beta$ presented in the work of \cite{que2020}. The models considered are the following:
\begin{align*}
&\text{Model 1:} Y_i=  \beta_{1}(z_i)X(z_{i}) + \epsilon_i^{H_{s}=0.40, H_{t}= 0.65},\quad i=1,\dots,n \\
&\text{Model 2:} Y_i=  \beta_{1}(z_i)X(z_{i}) + \epsilon_i^{H_{s}=0.40, H_{t}= 0.90},\quad i=1,\dots,n \\
&\text{Model 3:} Y_i=  \beta_{2}(z_i)X(z_{i}) + \epsilon_i^{H_{s}=0.40, H_{t}= 0.65},\quad i=1,\dots,n \\
&\text{Model 4:} Y_i=  \beta_{2}(z_i)X(z_{i}) + \epsilon_i^{H_{s}=0.40, H_{t}= 0.90},\quad i=1,\dots,n 
\end{align*}
where, $\beta_{1}=1 + (4(x+y)/12)$ represents a plane with a slight inclination
and $\beta_{2}=1 + (36- (6- (25x)/2)^2)(36- (6- (25y)/2)^2)/(324*8)$ a curved surface. $X(z_{i})$ corresponds to a Spatio - Temporal Auto Regressive model of order $(1,1)$. $\epsilon_i^{H_{s}=0.40, H_{t}= 0.65}$ and $\epsilon_i^{H_{s}=0.40, H_{t}= 0.90}$ corresponds to two different noises considered. In the following graphic we present three different times, $t_{1}$, $t_{50}$ and $t_{100}$ for the four different models considered

\begin{figure}[h!]
     \centering
     \begin{subfigure}[h]{0.32\textwidth}
         \centering
         \includegraphics[width=1\textwidth]{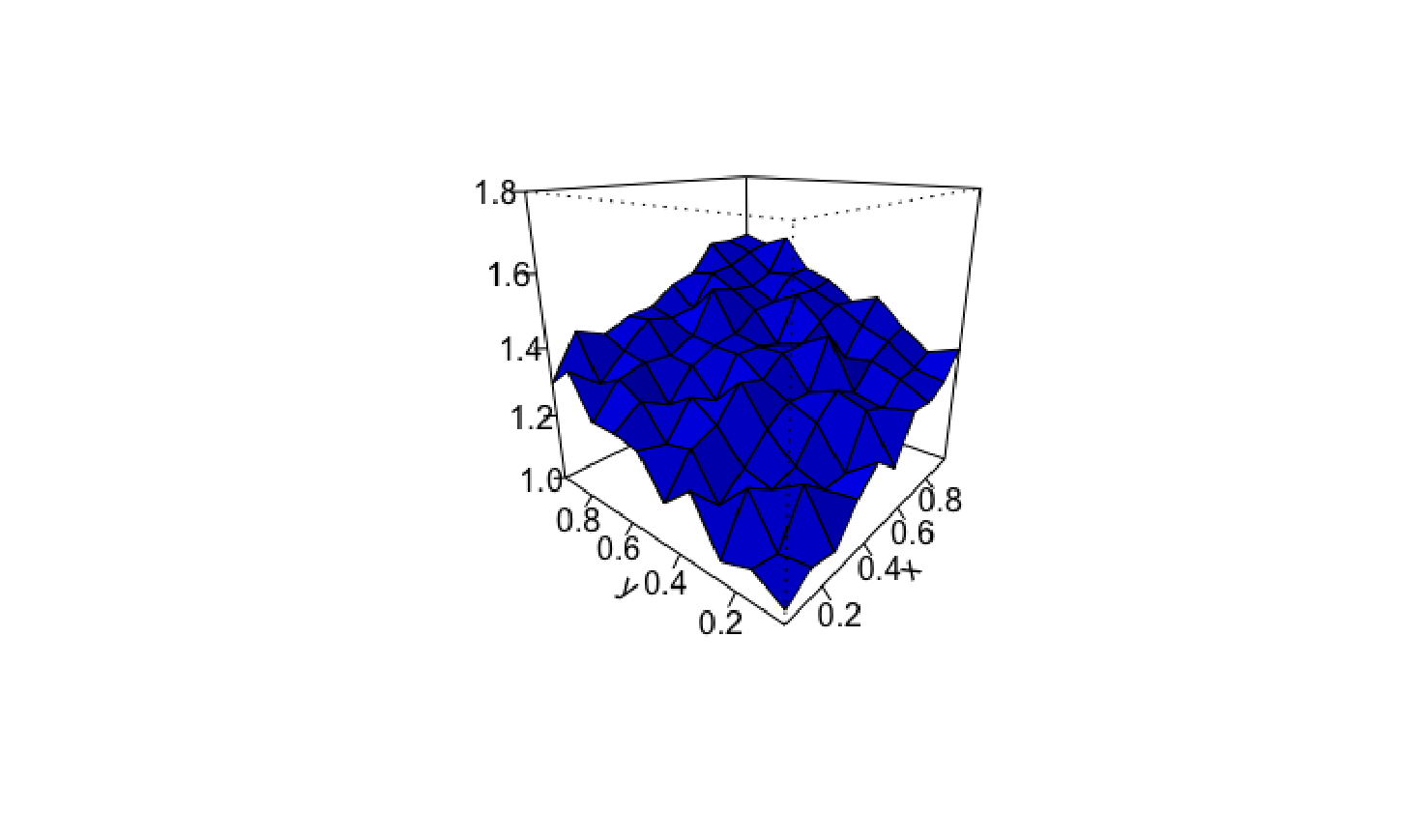}
         \caption{$t_{1}$}
         \label{m1t1}
     \end{subfigure}
     \hfill
     \begin{subfigure}[h]{0.32\textwidth}
         \centering
         \includegraphics[width=1\textwidth]{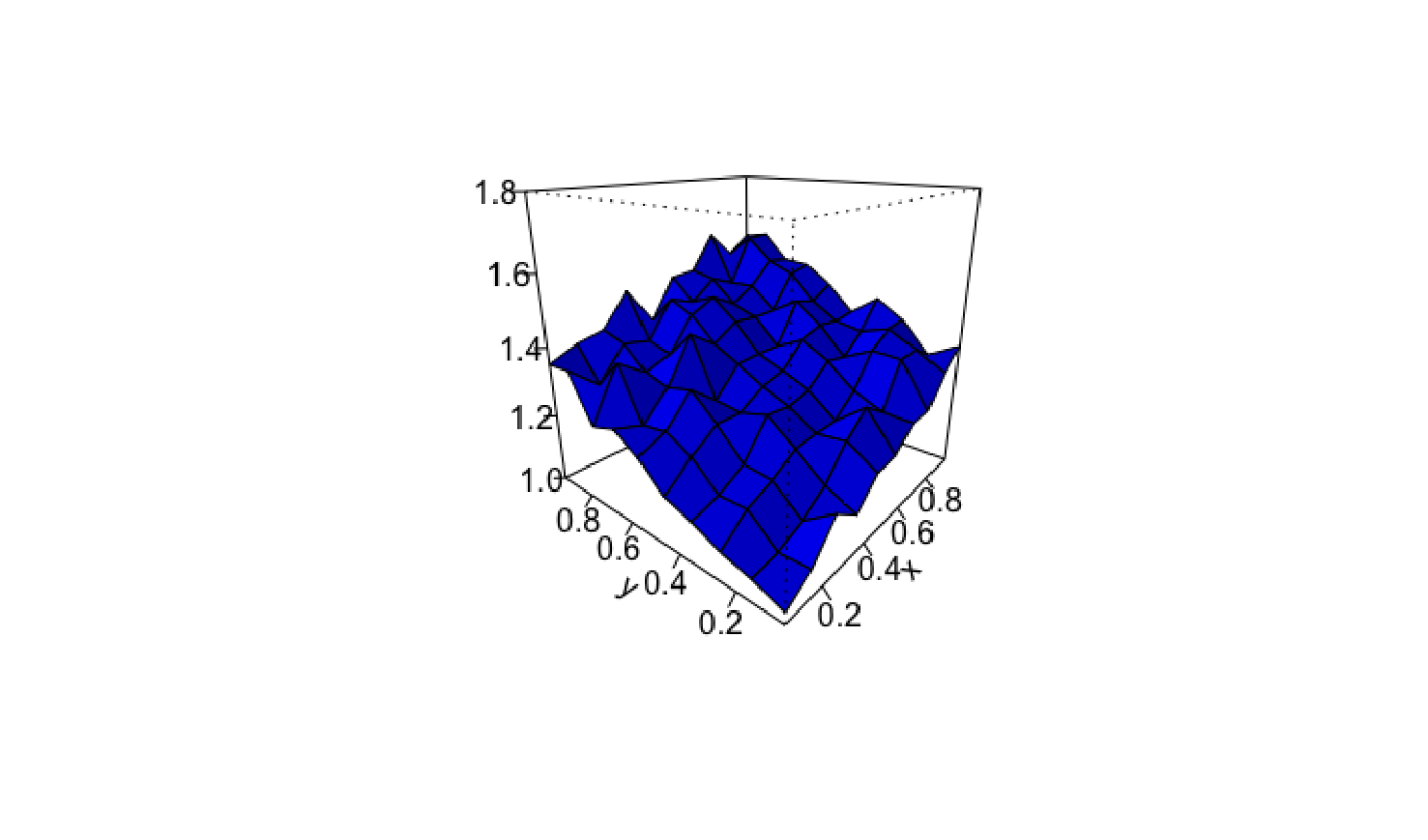}
         \caption{$t_{50}$}
         \label{m1t50}
     \end{subfigure}
     \hfill
     \begin{subfigure}[h]{0.32\textwidth}
         \centering
         \includegraphics[width=1\textwidth]{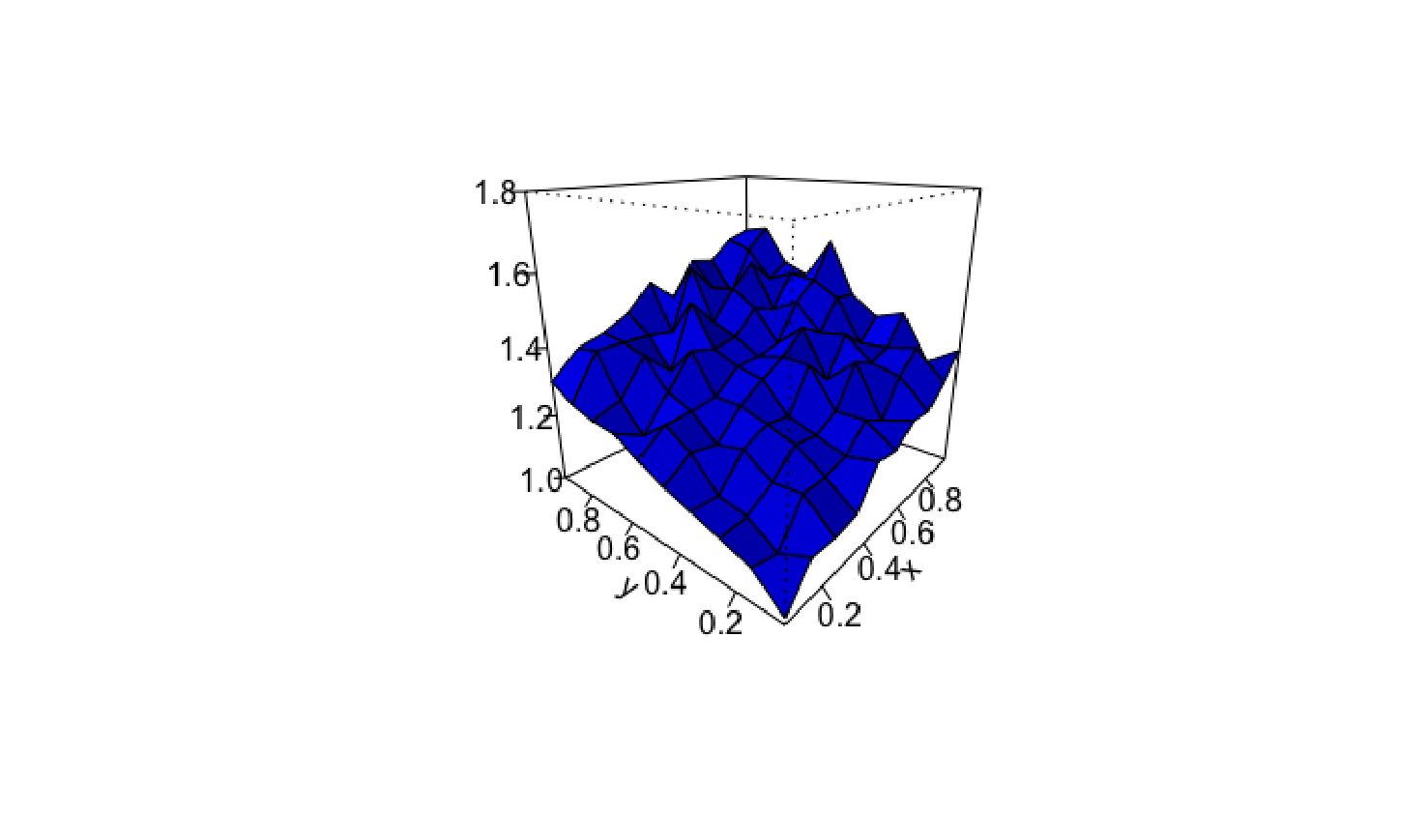}
         \caption{$t_{100}$}
         \label{m1t100}
     \end{subfigure}
     \hfill
     \begin{subfigure}[h]{0.32\textwidth}
         \centering
         \includegraphics[width=1\textwidth]{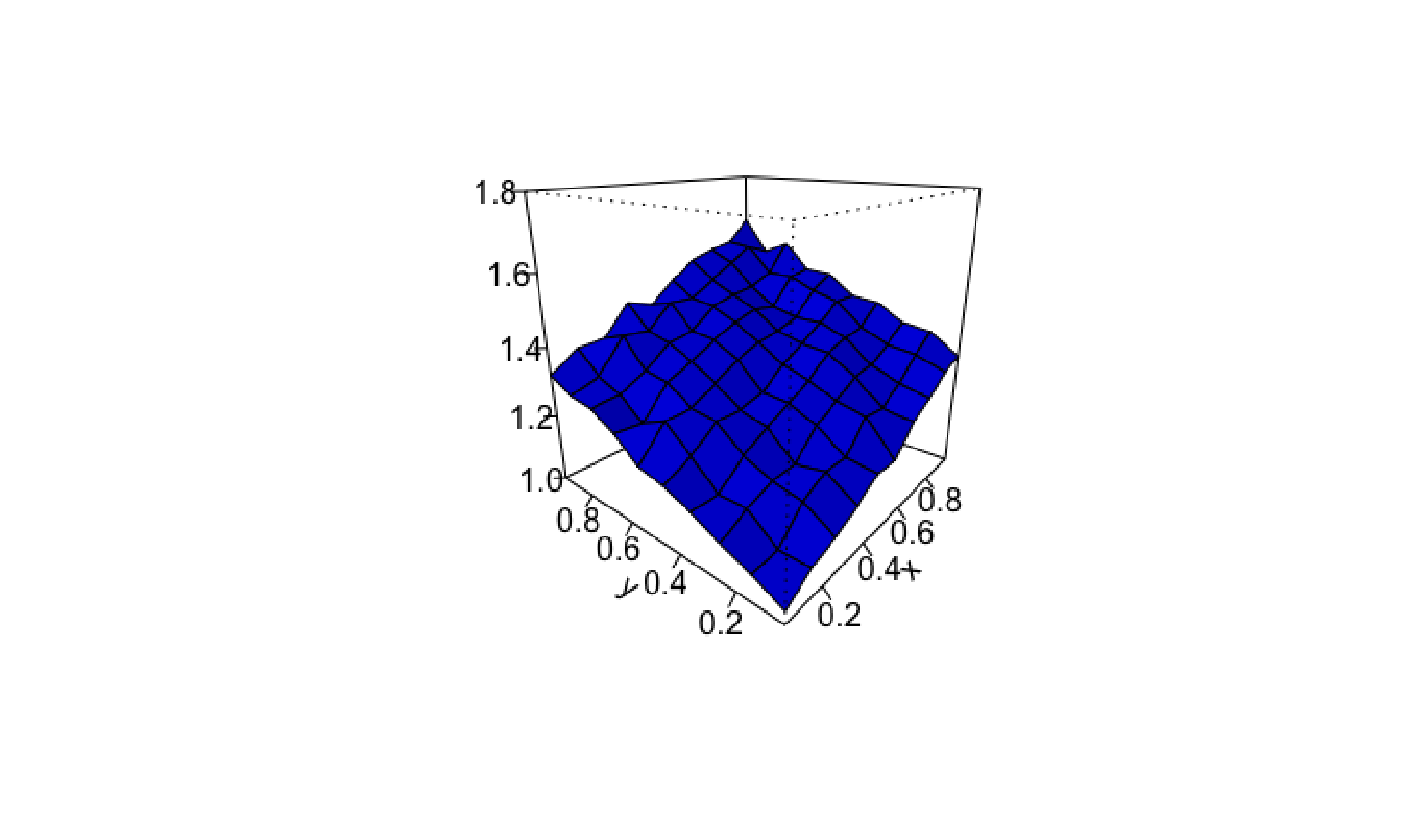}
         \caption{$t_{1}$}
         \label{m2t1}
     \end{subfigure}
     \hfill
     \begin{subfigure}[h]{0.32\textwidth}
         \centering
         \includegraphics[width=1\textwidth]{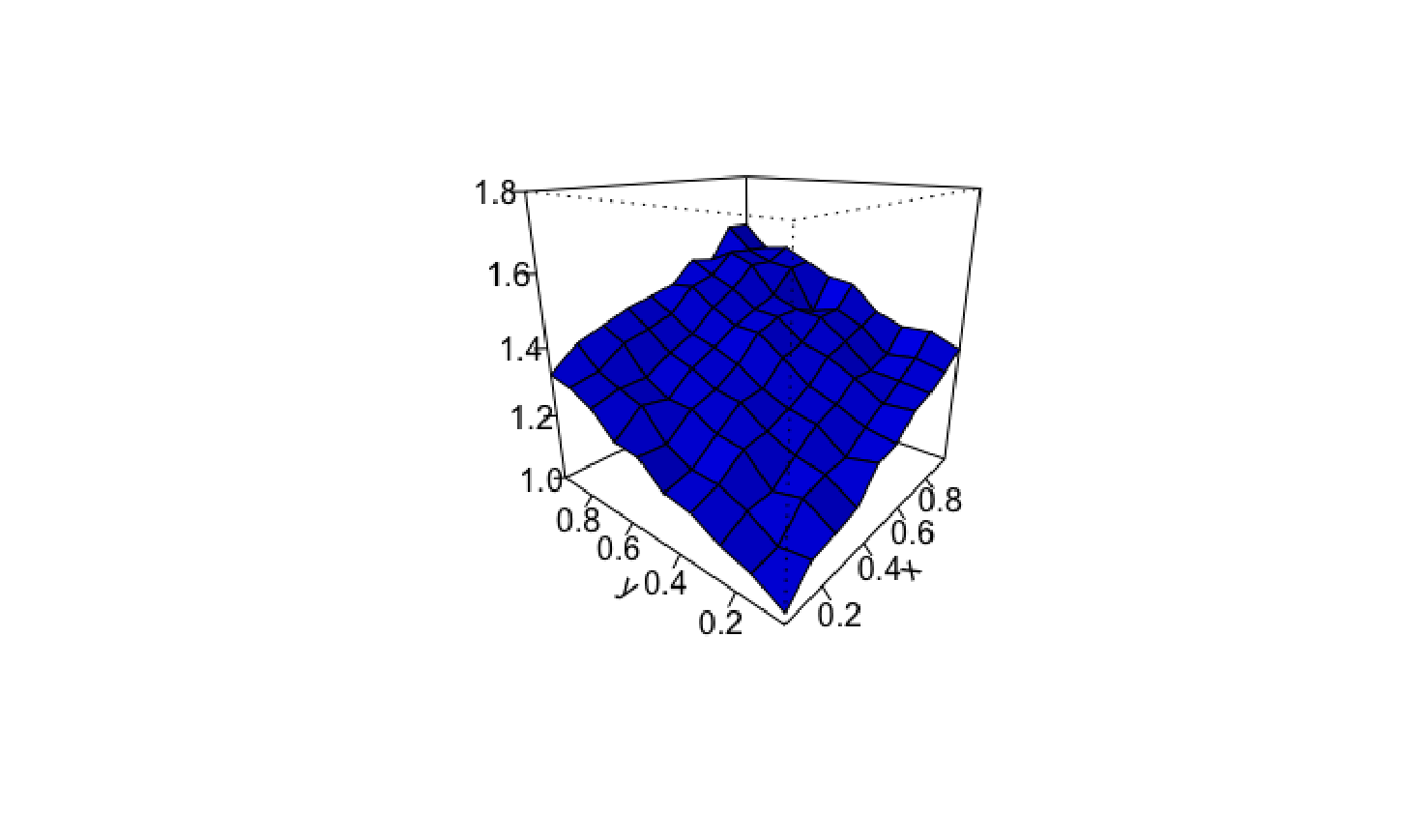}
         \caption{$t_{50}$}
         \label{m2t50}
     \end{subfigure}
     \hfill
     \begin{subfigure}[h]{0.32\textwidth}
         \centering
         \includegraphics[width=1\textwidth]{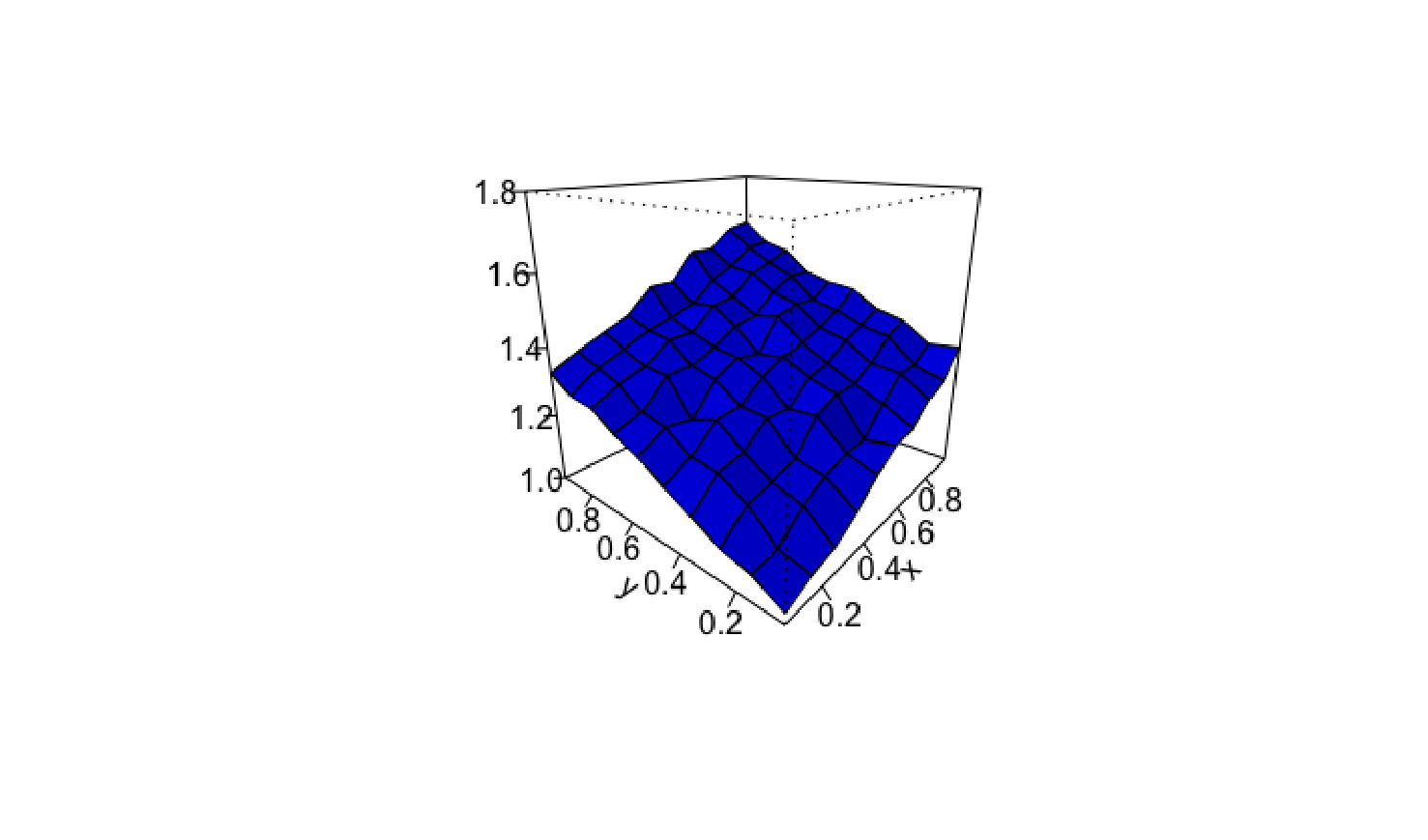}
         \caption{$t_{100}$}
         \label{m2t100}
     \end{subfigure}
        \caption{Model 1 and 2, in space with $H_{s}=0.40$ and time with $H_{t}=0.65$ (a, b, c) and $H_{t}=0.90$ (d, e, f) in three different instants of time.}
        \label{model1-2}
\end{figure}

\begin{figure}[h!]
     \centering
     \begin{subfigure}[h]{0.32\textwidth}
         \centering
         \includegraphics[width=1\textwidth]{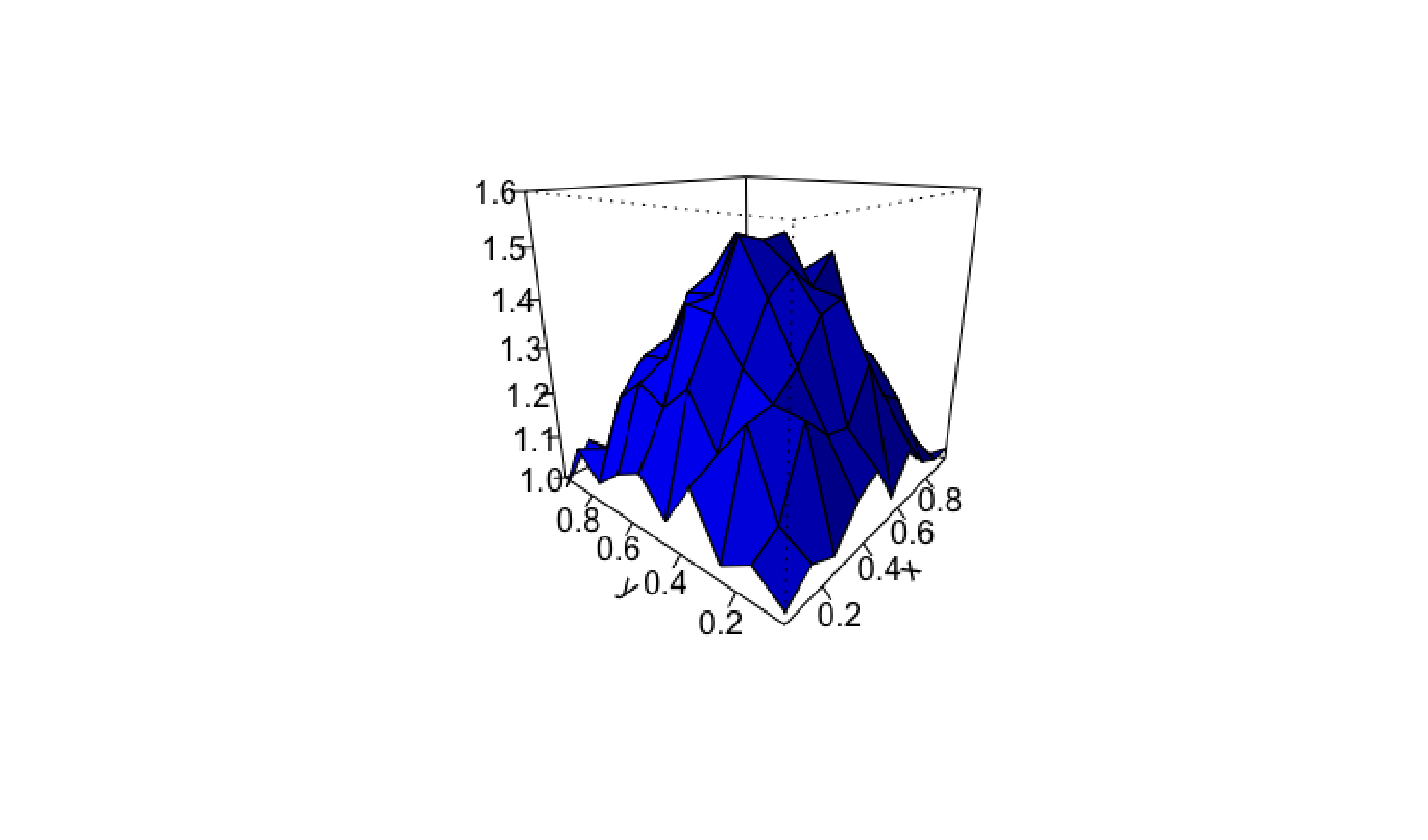}
         \caption{$t_{1}$}
         \label{m3t1}
     \end{subfigure}
     \hfill
     \begin{subfigure}[h]{0.32\textwidth}
         \centering
         \includegraphics[width=1\textwidth]{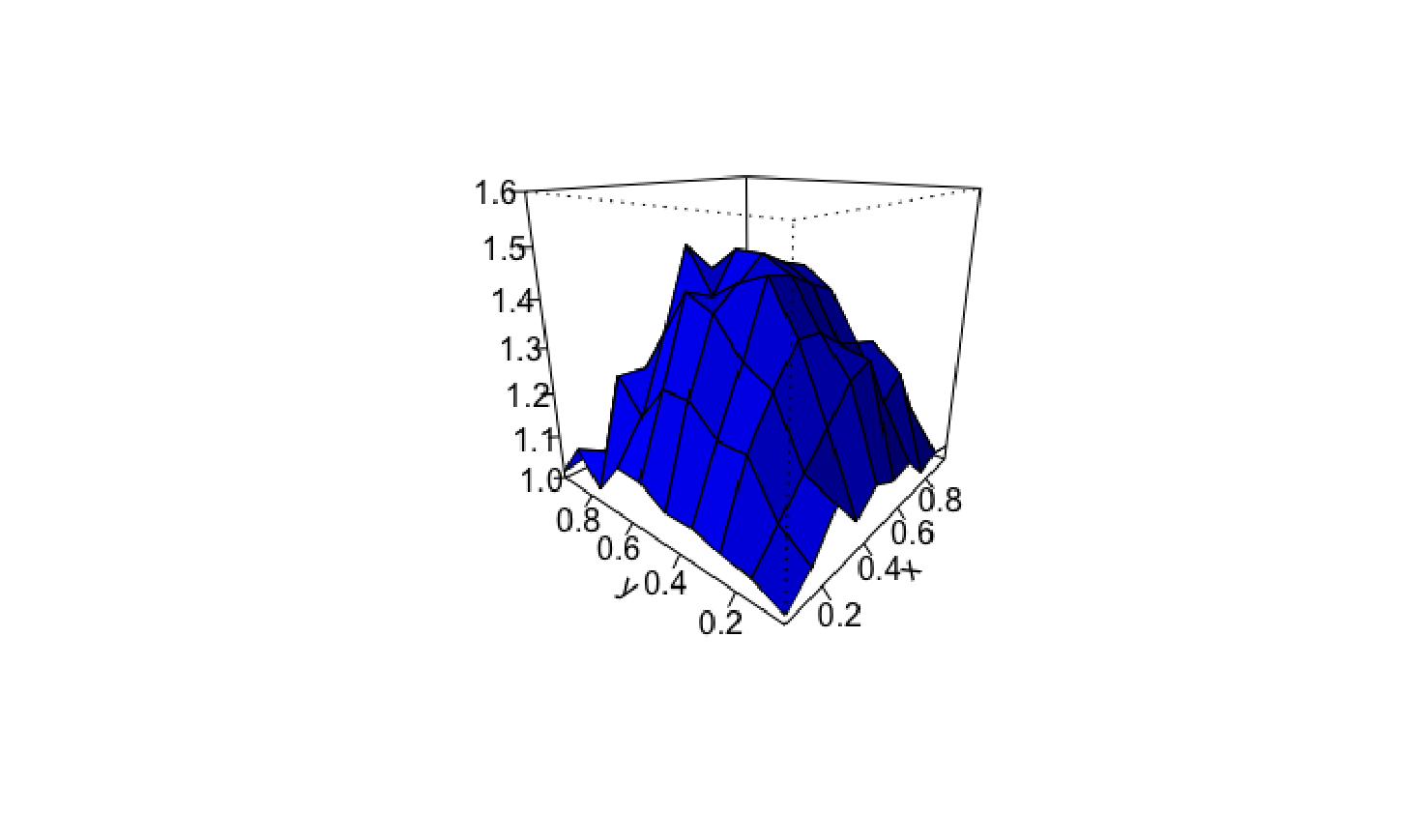}
         \caption{$t_{50}$}
         \label{m3t50}
     \end{subfigure}
     \hfill
     \begin{subfigure}[h]{0.32\textwidth}
         \centering
         \includegraphics[width=1\textwidth]{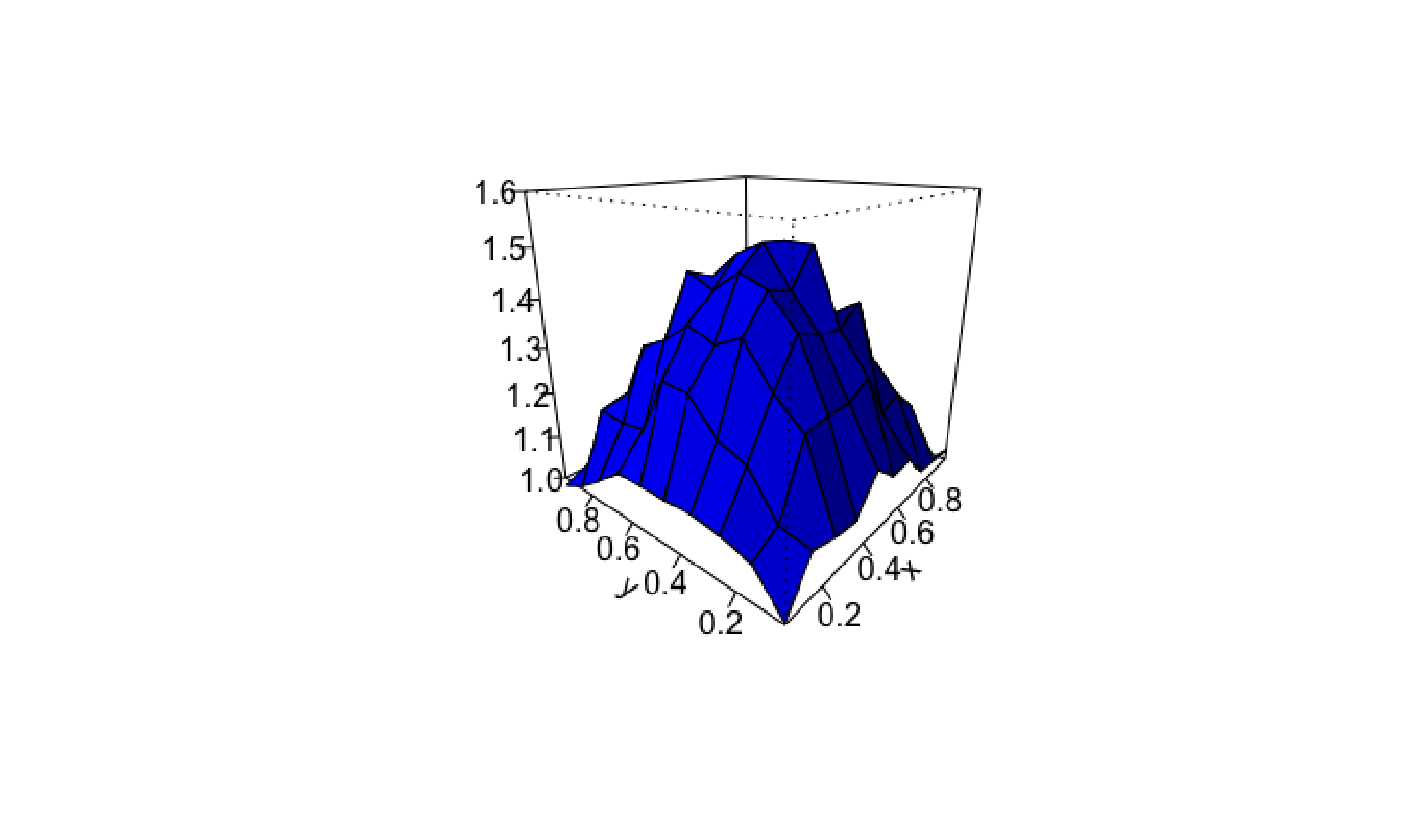}
         \caption{$t_{100}$}
         \label{m3t100}
     \end{subfigure}
\end{figure}
\clearpage
\begin{figure} \ContinuedFloat
     \begin{subfigure}[h]{0.32\textwidth}
         \centering
         \includegraphics[width=1\textwidth]{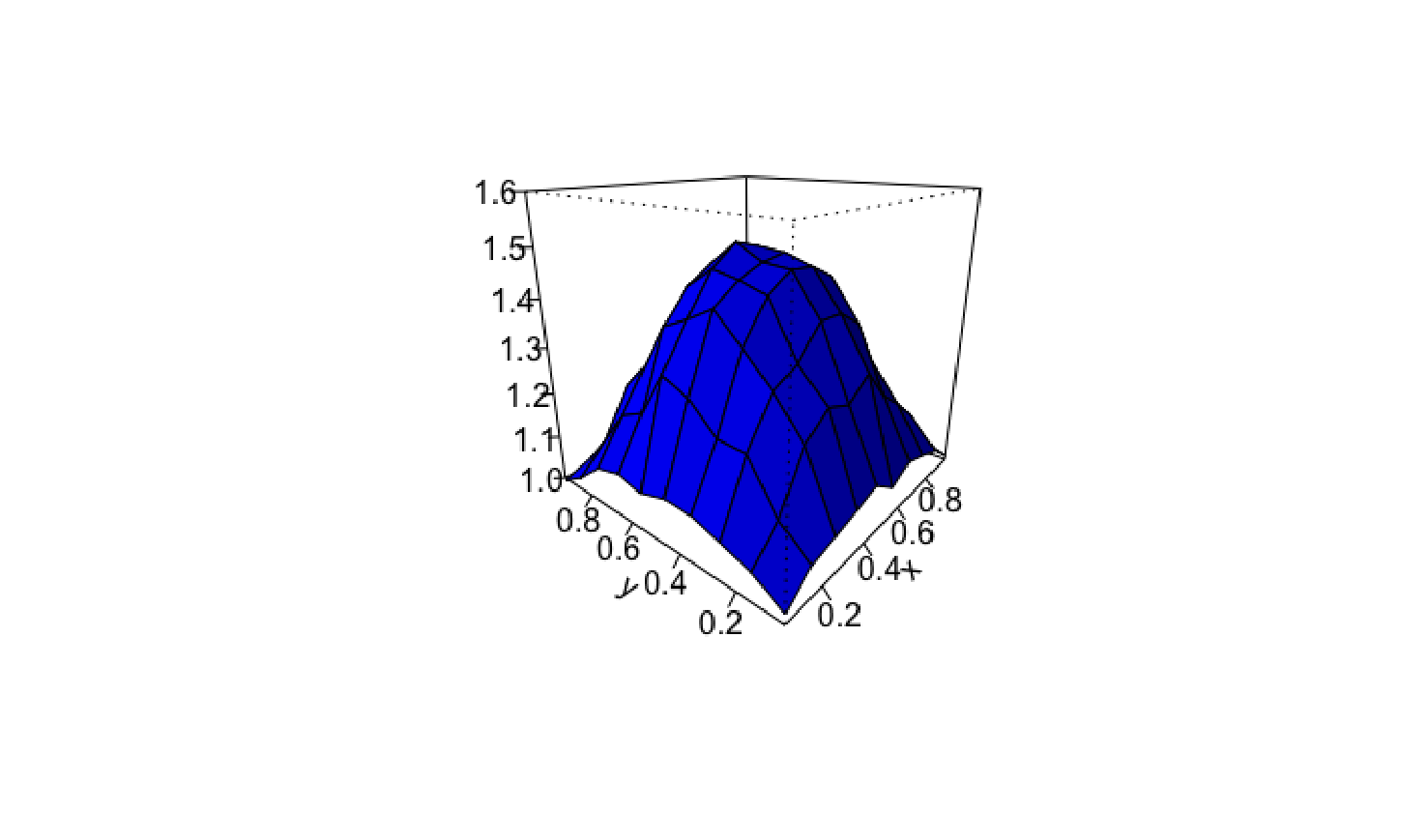}
         \caption{$t_{1}$}
         \label{m4t1}
     \end{subfigure}
     \hfill
     \begin{subfigure}[h]{0.32\textwidth}
         \centering
         \includegraphics[width=1\textwidth]{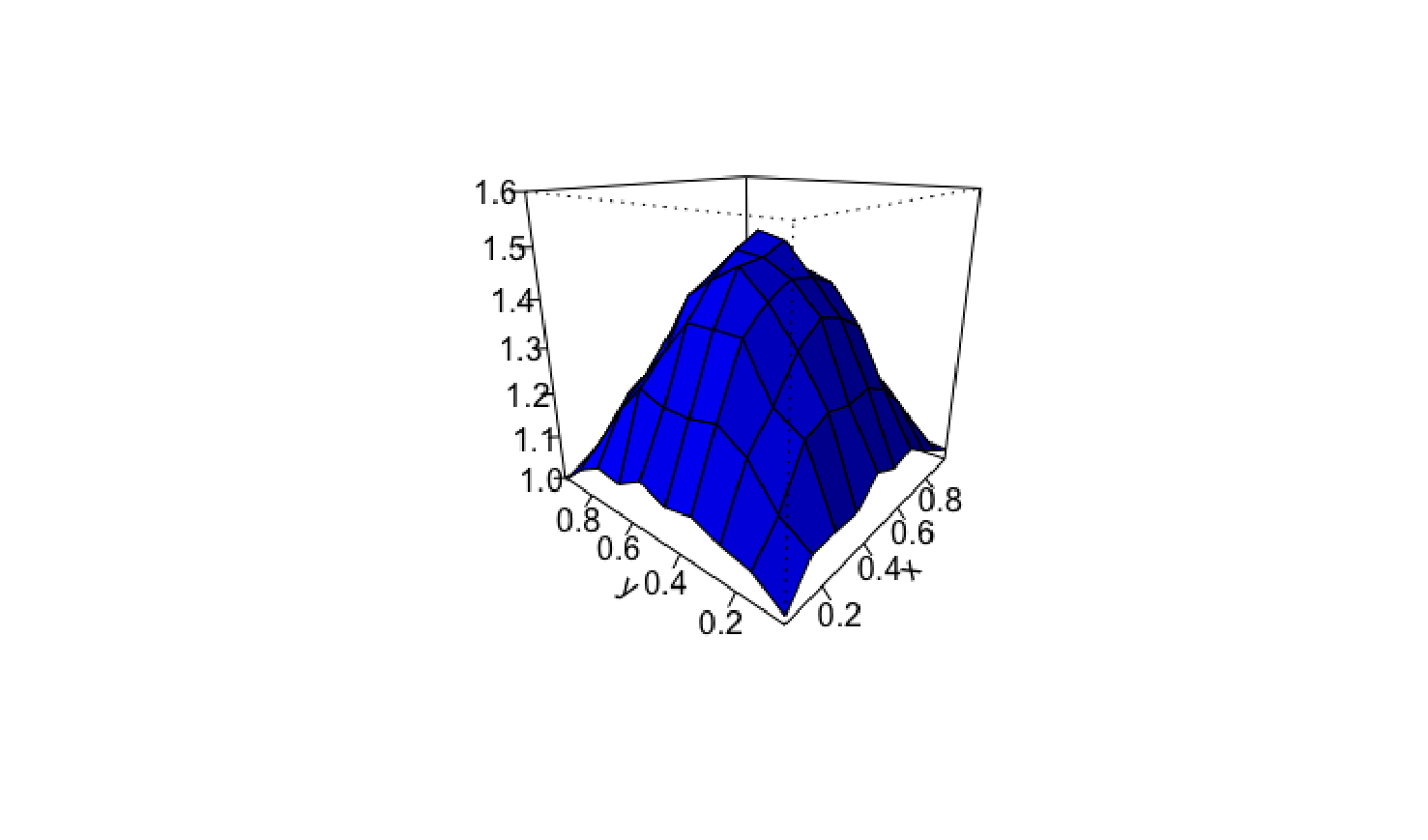}
         \caption{$t_{50}$}
         \label{m4t50}
     \end{subfigure}
     \hfill
     \begin{subfigure}[h]{0.32\textwidth}
         \centering
         \includegraphics[width=1\textwidth]{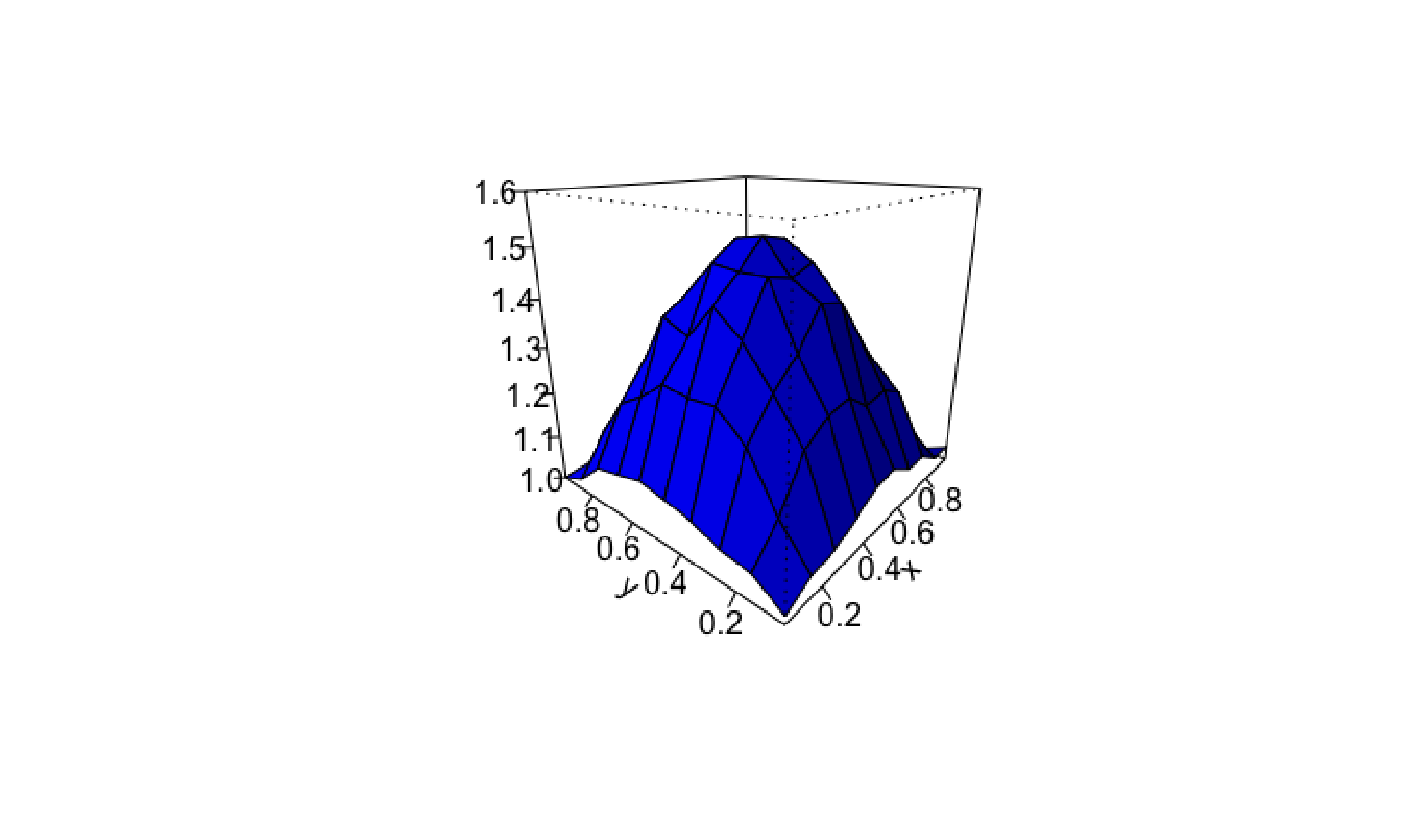}
         \caption{$t_{100}$}
         \label{m4t100}
     \end{subfigure}
        \caption{Model 3 and 4, in space with $H_{s}=0.40$ and time with $H_{t}=0.65$ (a, b, c) and $H_{t}=0.90$ (d, e, f) in three different instants of time.}
        \label{model3-4}
\end{figure}

Figures \ref{m1t1}, \ref{m1t50}, \ref{m1t100}, \ref{m3t1}, \ref{m3t50} and \ref{m3t100}, present different scenarios where a bigger variability, in time , is considered, this is a consequence of $H_{t}=0.65$. Meanwhile, figures \ref{m2t1}, \ref{m2t50}, \ref{m2t100}, \ref{m4t1}, \ref{m4t50} and \ref{m4t100} a decrease in variance is seen over time. On the other hand, regarding the spatial heterogeneity of the parameters, similar to the work of \cite{Foth2017}, in the models \ref{m1t1}, \ref{m1t50}, \ref{m1t100}, \ref{m2t1}, \ref{m2t50}, and \ref{m2t100} medium spatial heterogeneity is observed; in contrast to a high spatial heterogeneity for the models \ref{m3t1}, \ref{m3t50} and \ref{m3t100}, \ref{m4t1}, \ref{m4t50}, and \ref{m4t100}. These are the models that will be considered to estimate the parameter $\beta_{1}(z_{i})$ of the model \eqref{model}.

\subsection*{Estimator performance}
The estimation result, $\hat{Y}_{i} = \hat{\beta}_{1}(z_{i}) X (z_{i})$, in conjunction with the model defined by \eqref{model}, is presented below for the four different models simulated at $t_{1}, t_{50}$ and $t_{100}$.
\begin{figure}[h!]
     \centering
     \begin{subfigure}[h]{0.32\textwidth}
         \centering
         \includegraphics[width=1\textwidth]{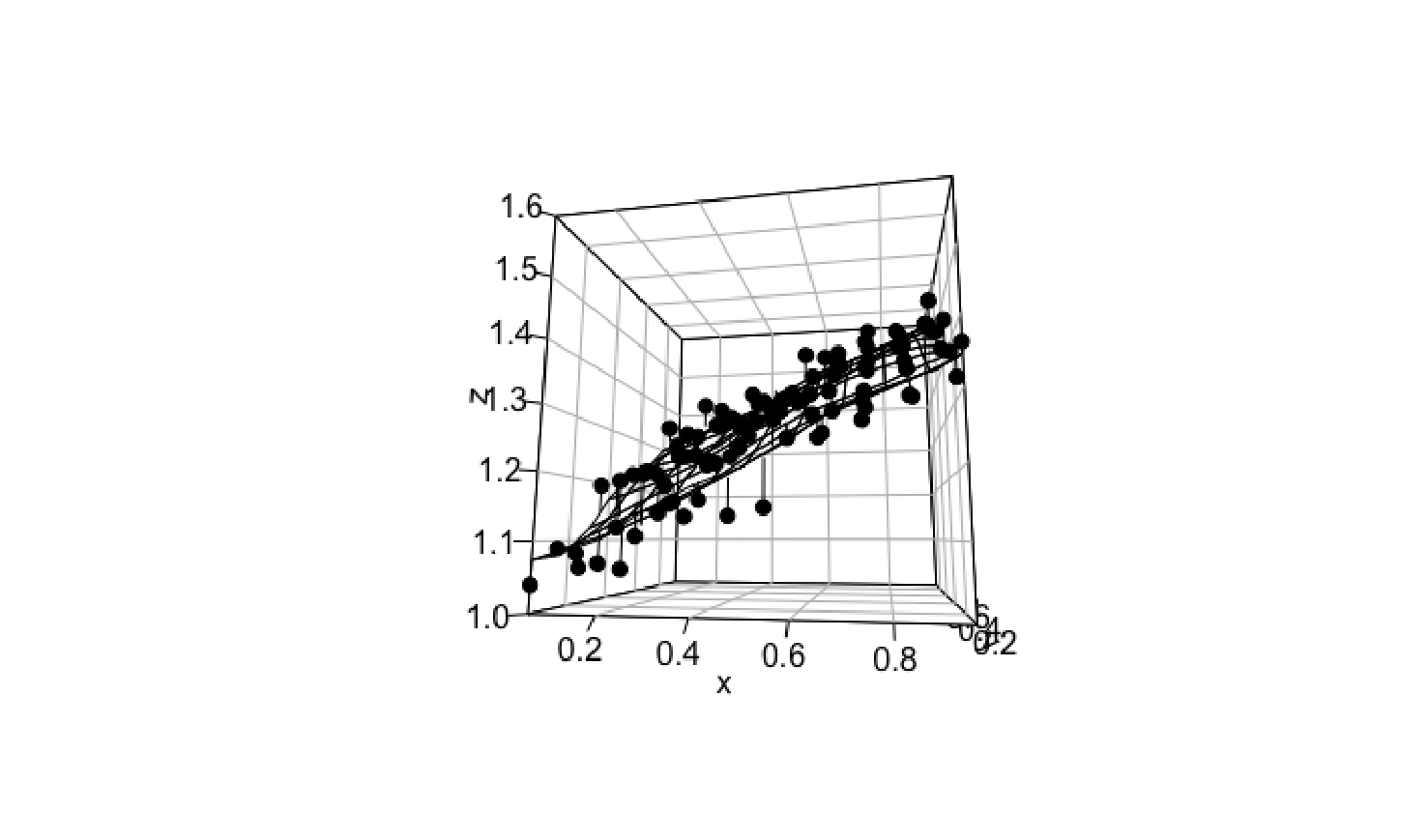}
         \caption{$t_{1}$}
         \label{estm1t1}
     \end{subfigure}
     \hfill
     \begin{subfigure}[h]{0.32\textwidth}
         \centering
         \includegraphics[width=1\textwidth]{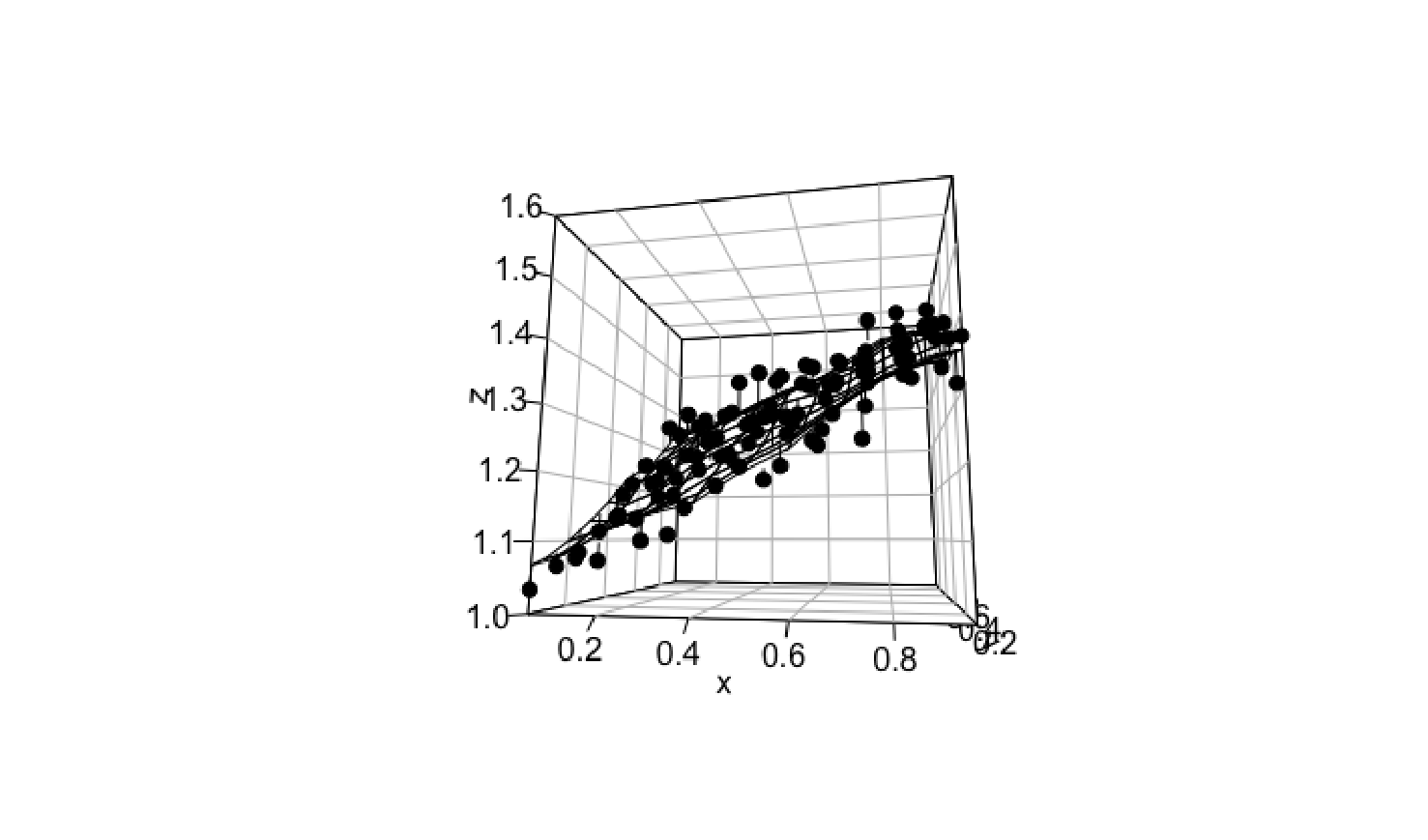}
         \caption{$t_{50}$}
         \label{estm1t50}
     \end{subfigure}
     \hfill
     \begin{subfigure}[h]{0.32\textwidth}
         \centering
         \includegraphics[width=1\textwidth]{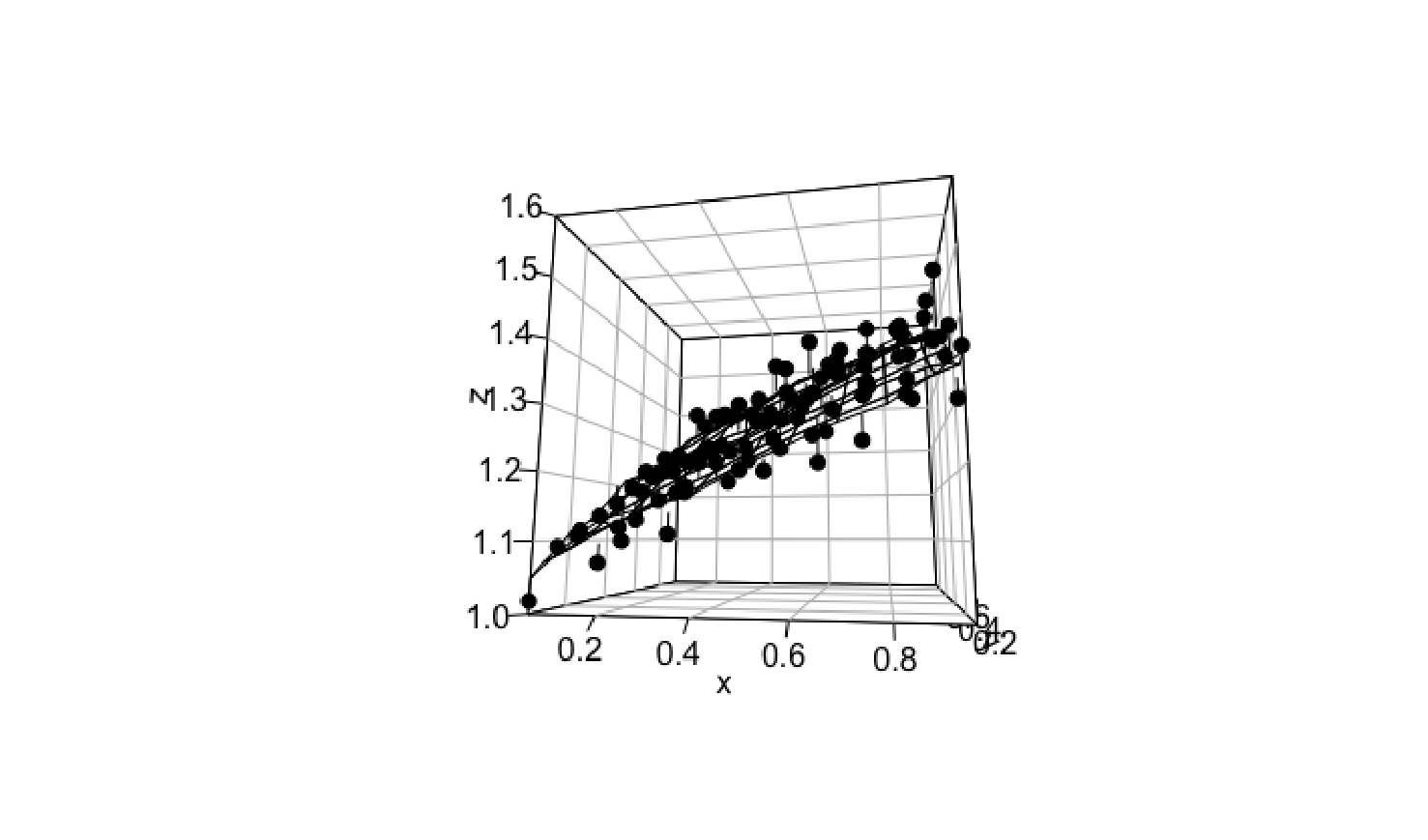}
         \caption{$t_{100}$}
         \label{estm1t100}
         \end{subfigure}
         \hfill
          \centering
     \begin{subfigure}[h]{0.32\textwidth}
         \centering
         \includegraphics[width=1\textwidth]{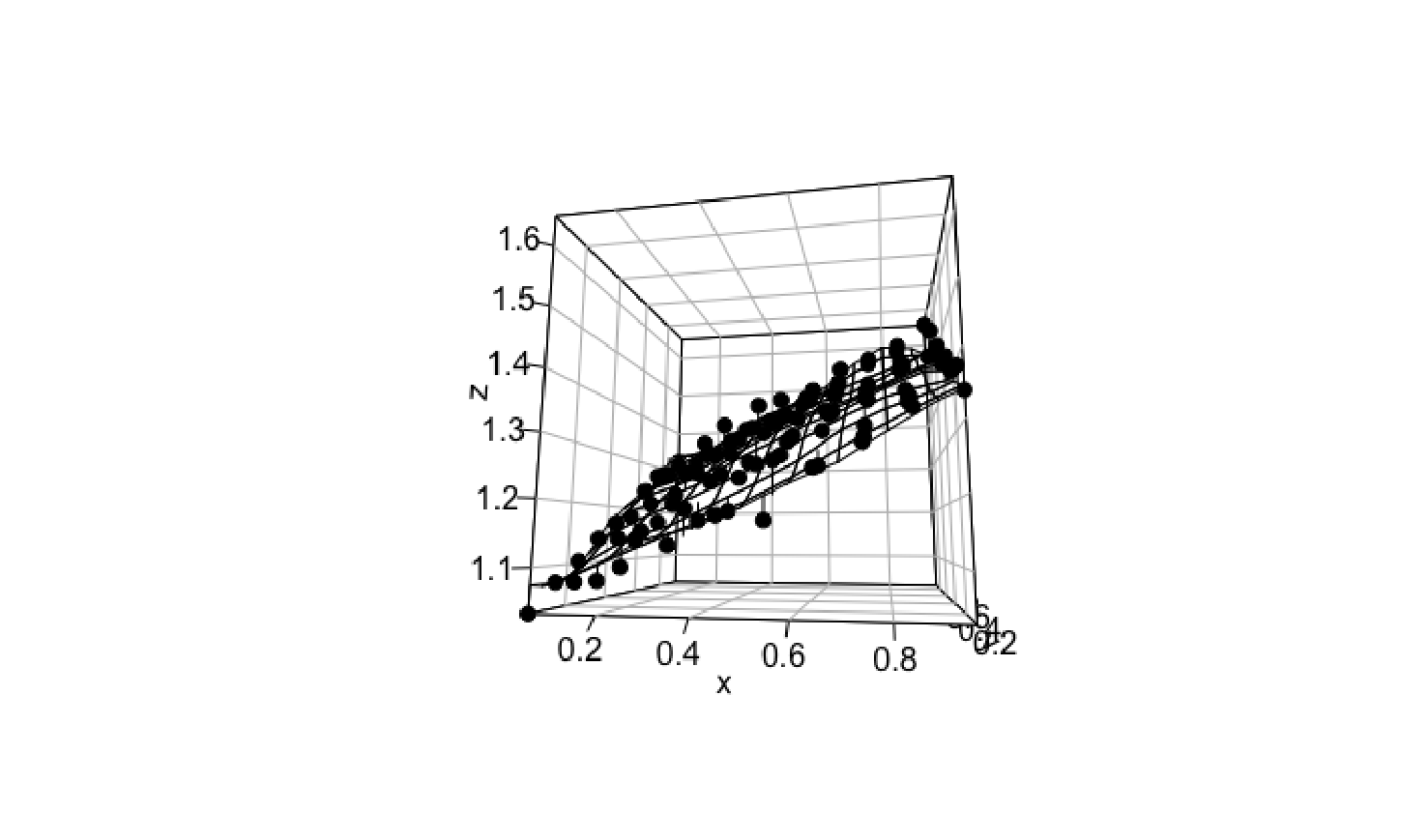}
         \caption{$t_{1}$}
         \label{estm2t1}
     \end{subfigure}
     \hfill
     \begin{subfigure}[h]{0.32\textwidth}
         \centering
         \includegraphics[width=1\textwidth]{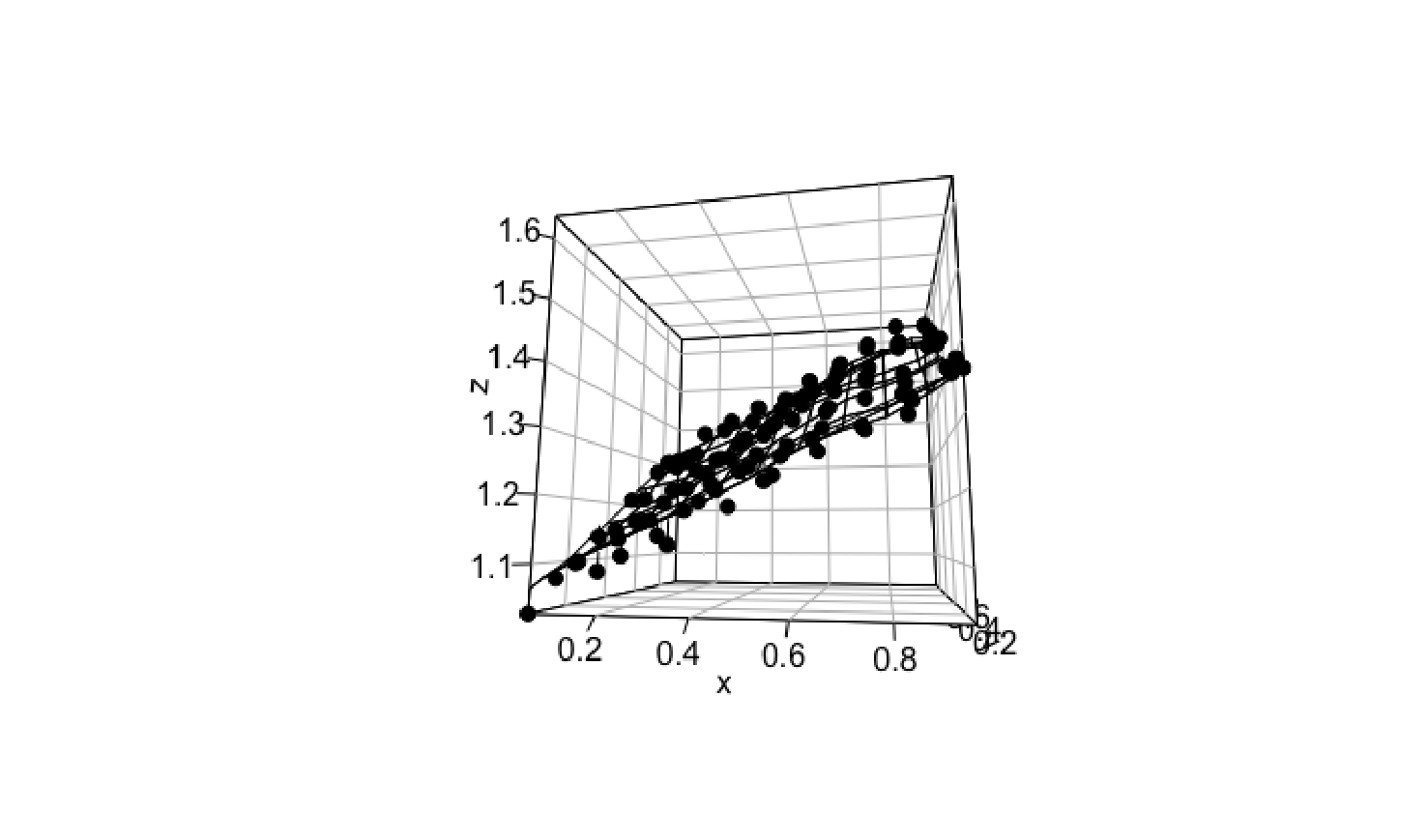}
         \caption{$t_{50}$}
         \label{estm2t50}
     \end{subfigure}
     \hfill
     \begin{subfigure}[h]{0.32\textwidth}
         \centering
         \includegraphics[width=1\textwidth]{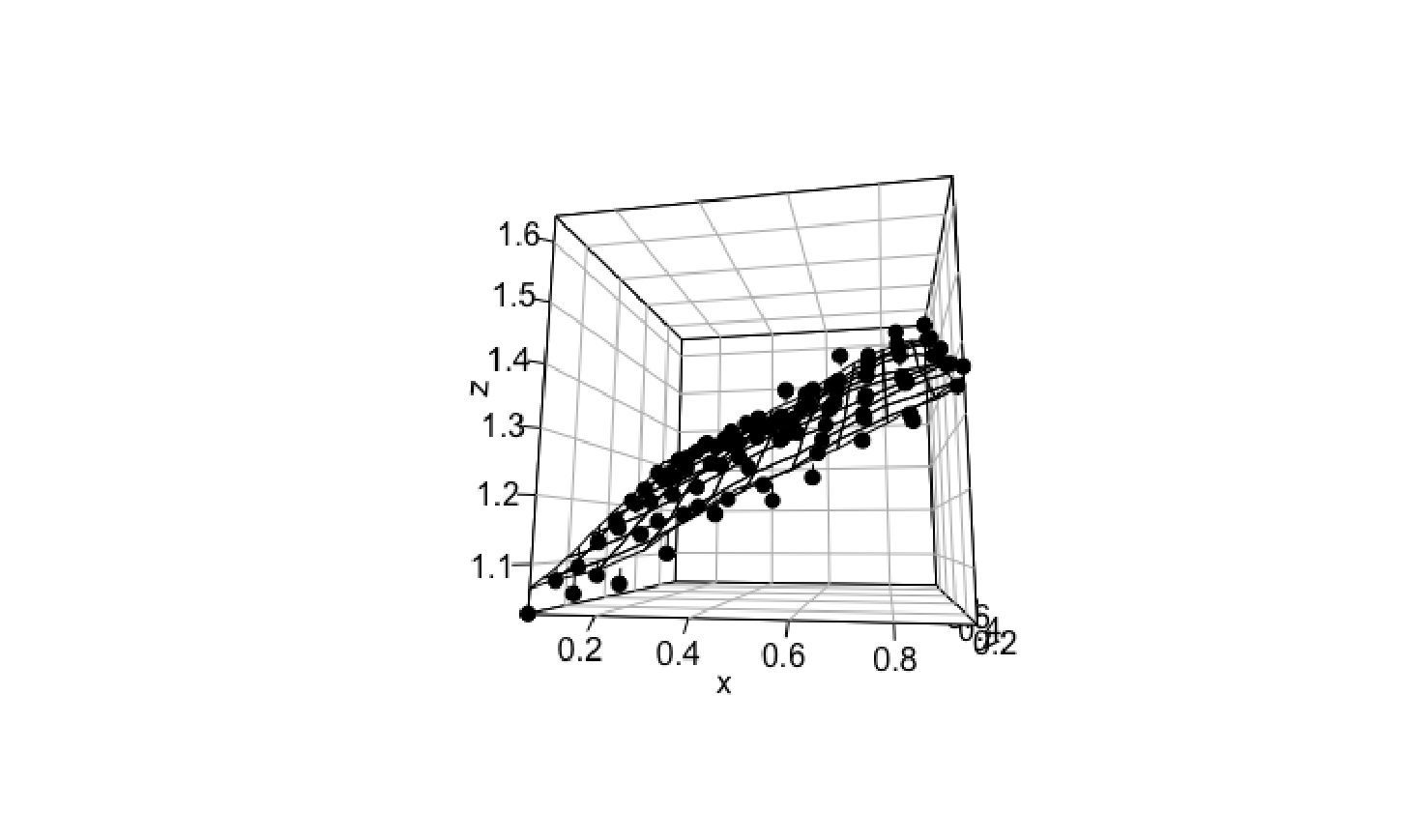}
         \caption{$t_{100}$}
         \label{estm2t100}
     \end{subfigure}
        \caption{Model 1 and 2, in space with $H_{s}=0.40$ and time with $H_{t}=0.65$ (a, b, c) and $H_{t}=0.90$ (d, e, f) in three different instants of time.}
        \label{model3-4}
\end{figure}
\clearpage
\begin{figure}[h!]
     \centering
     \begin{subfigure}[h]{0.32\textwidth}
         \centering
         \includegraphics[width=1\textwidth]{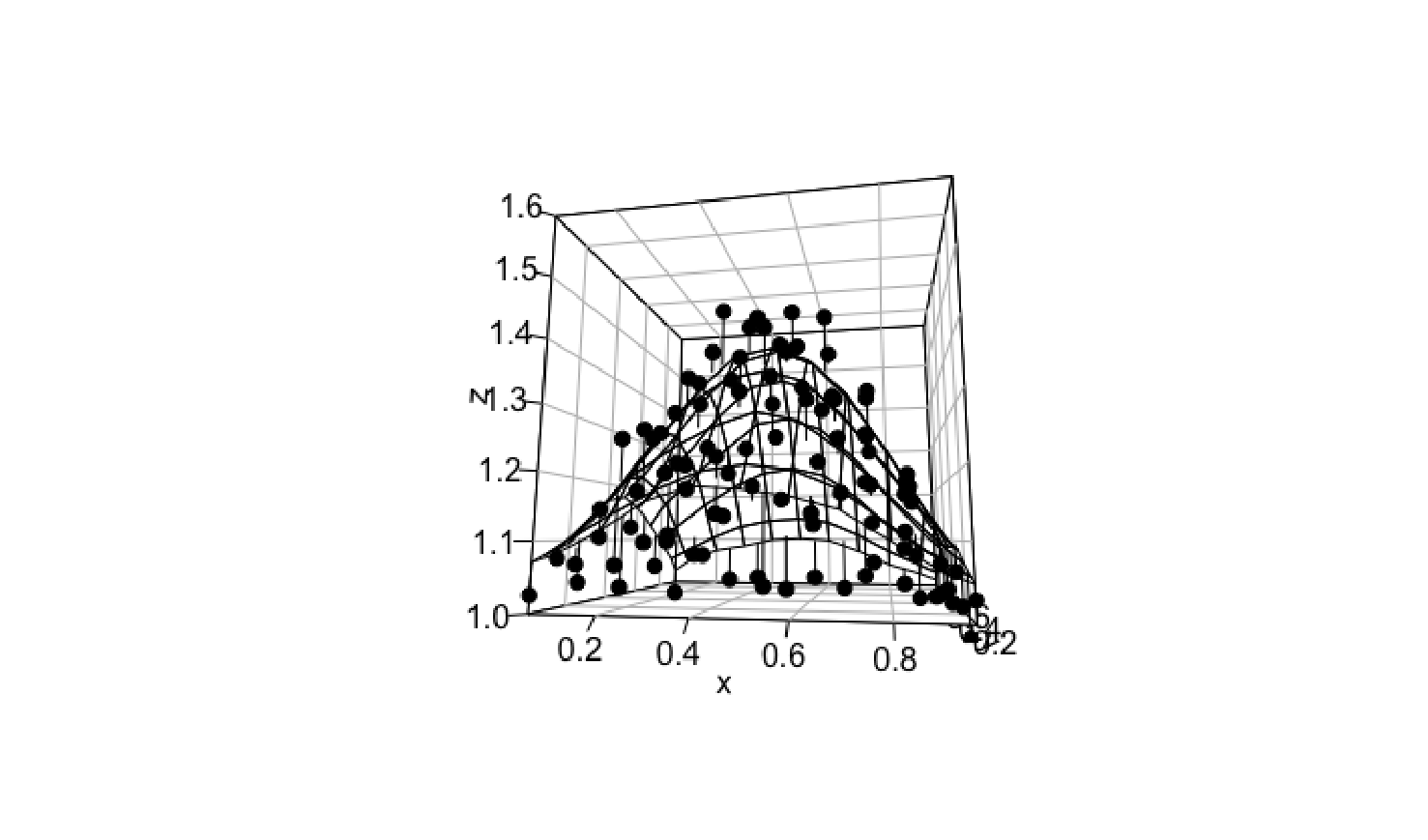}
         \caption{$t_{1}$}
         \label{estm1t1}
     \end{subfigure}
     \hfill
     \begin{subfigure}[h]{0.32\textwidth}
         \centering
         \includegraphics[width=1\textwidth]{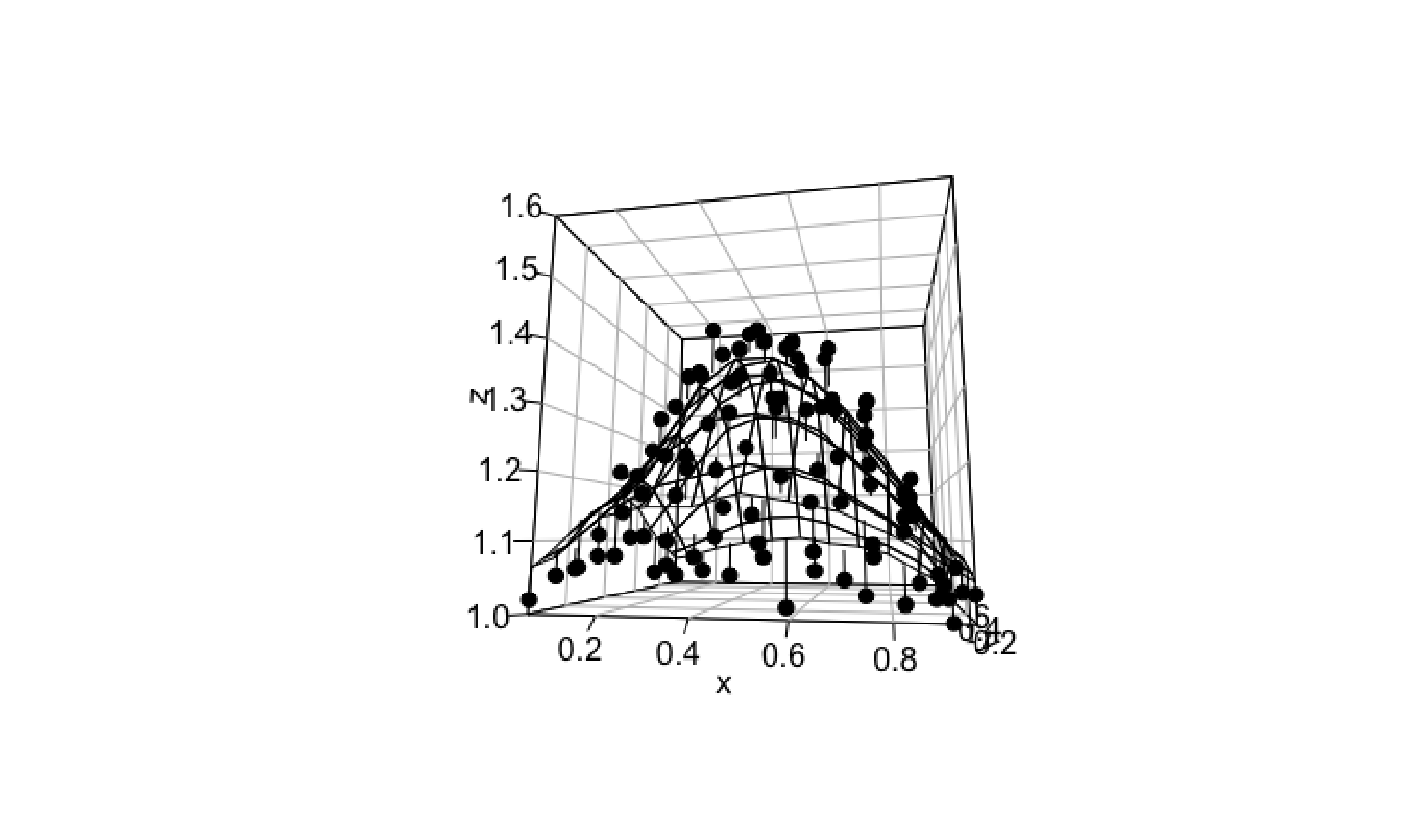}
         \caption{$t_{50}$}
         \label{estm1t50}
     \end{subfigure}
     \hfill
     \begin{subfigure}[h]{0.32\textwidth}
         \centering
         \includegraphics[width=1\textwidth]{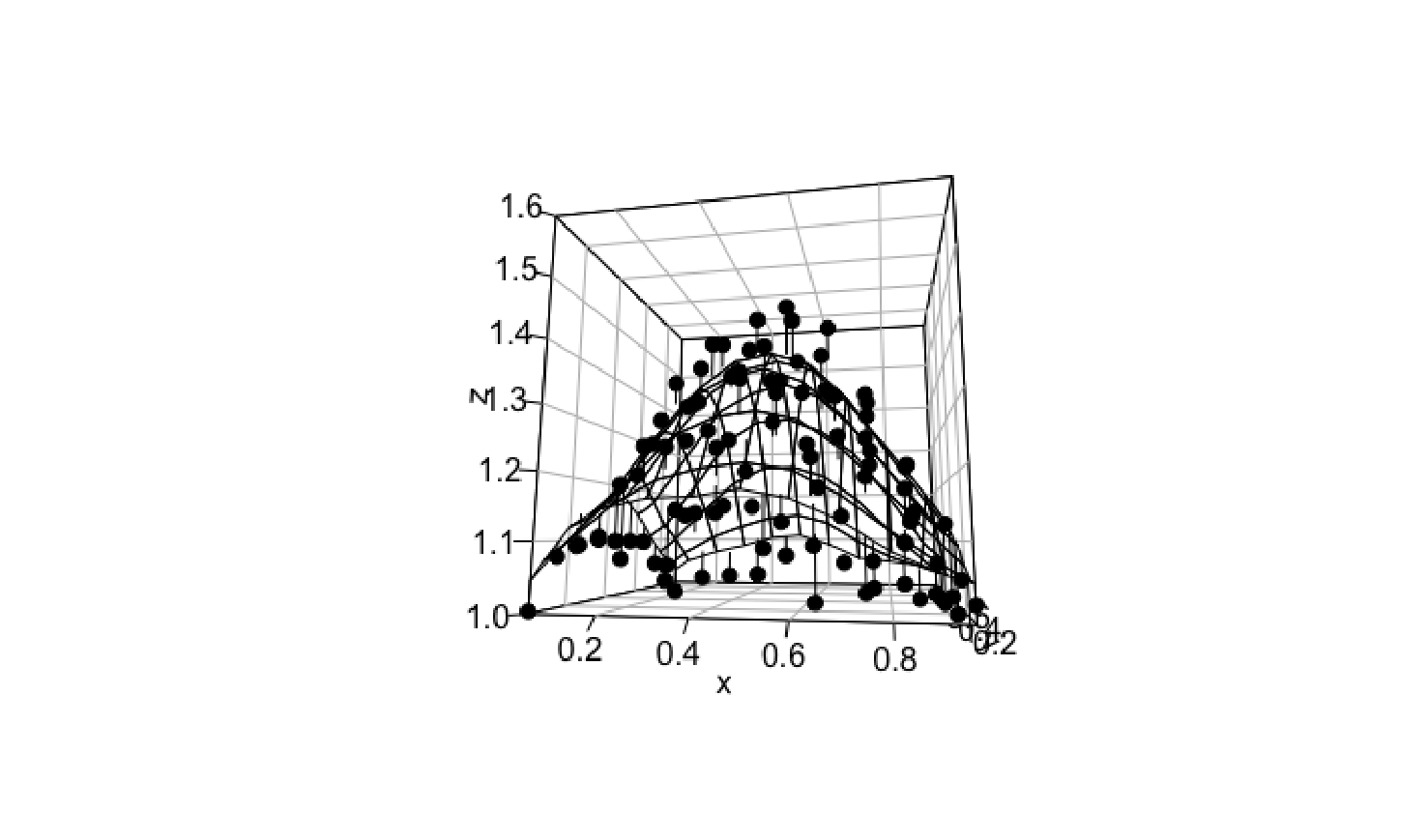}
         \caption{$t_{100}$}
         \label{estm1t100}
     \end{subfigure}
     \hfill
     \begin{subfigure}[h]{0.32\textwidth}
         \centering
         \includegraphics[width=1\textwidth]{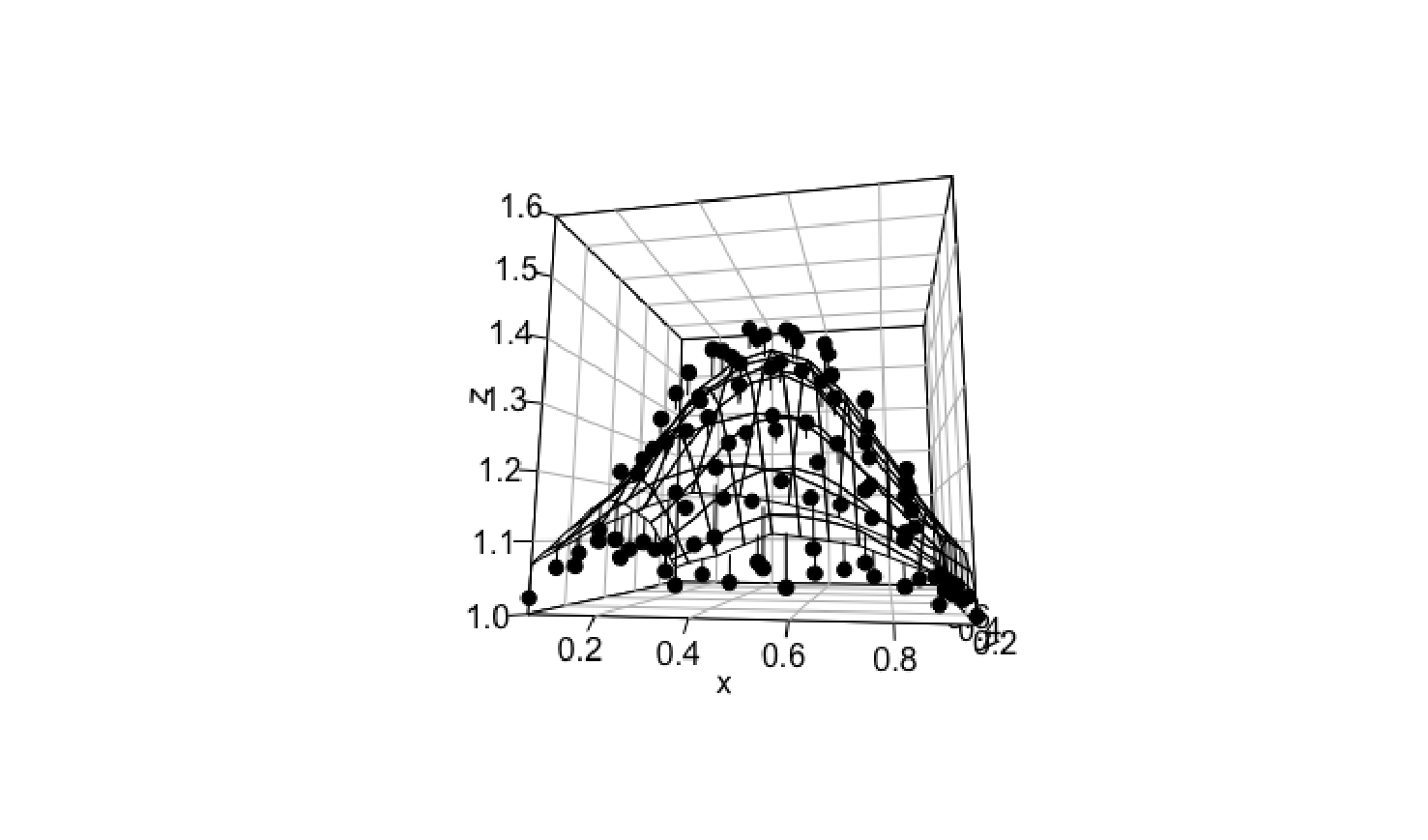}
         \caption{$t_{1}$}
         \label{estm2t1}
     \end{subfigure}
     \hfill
     \begin{subfigure}[h]{0.32\textwidth}
         \centering
         \includegraphics[width=1\textwidth]{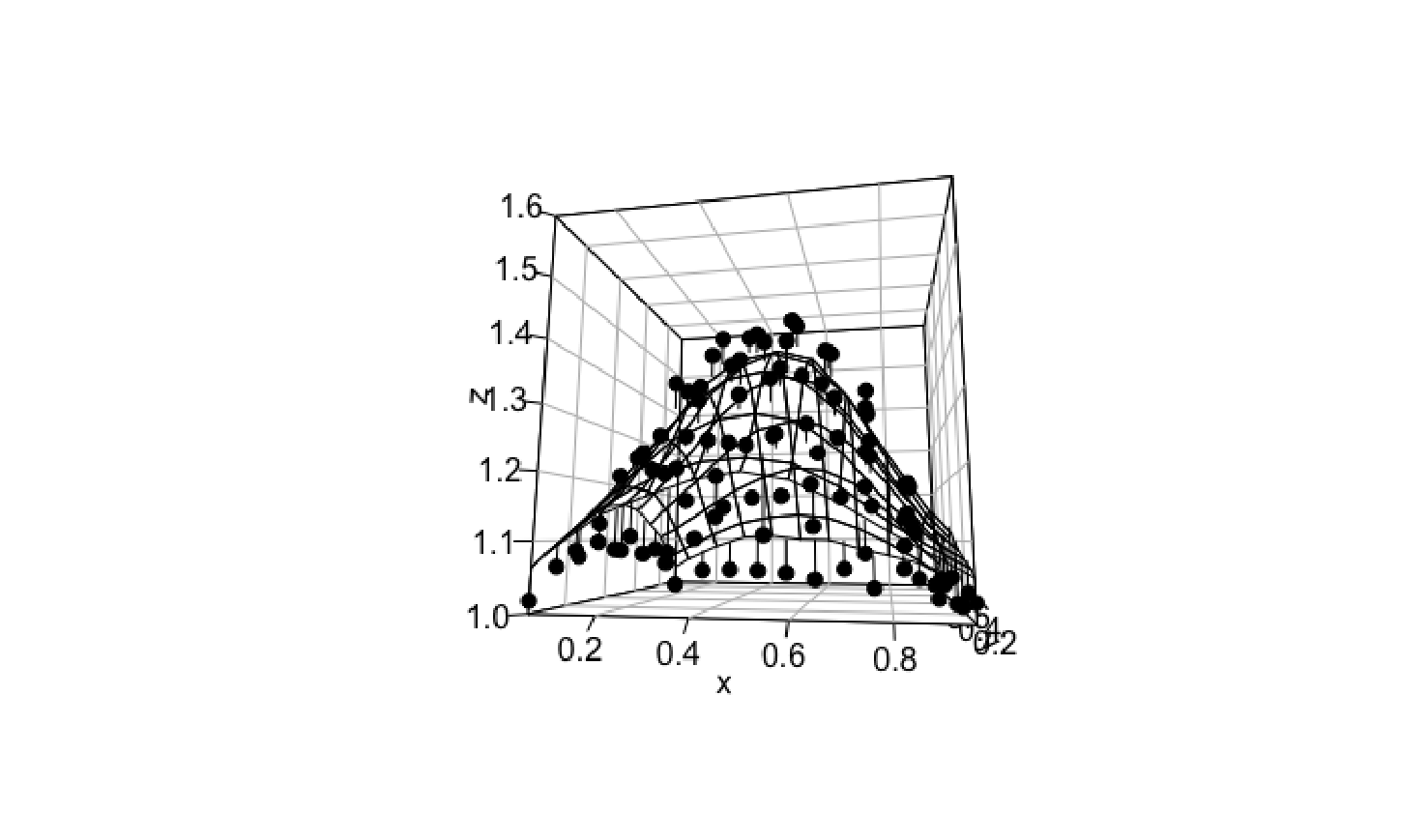}
         \caption{$t_{50}$}
         \label{estm2t50}
     \end{subfigure}
     \hfill
     \begin{subfigure}[h]{0.32\textwidth}
         \centering
         \includegraphics[width=1\textwidth]{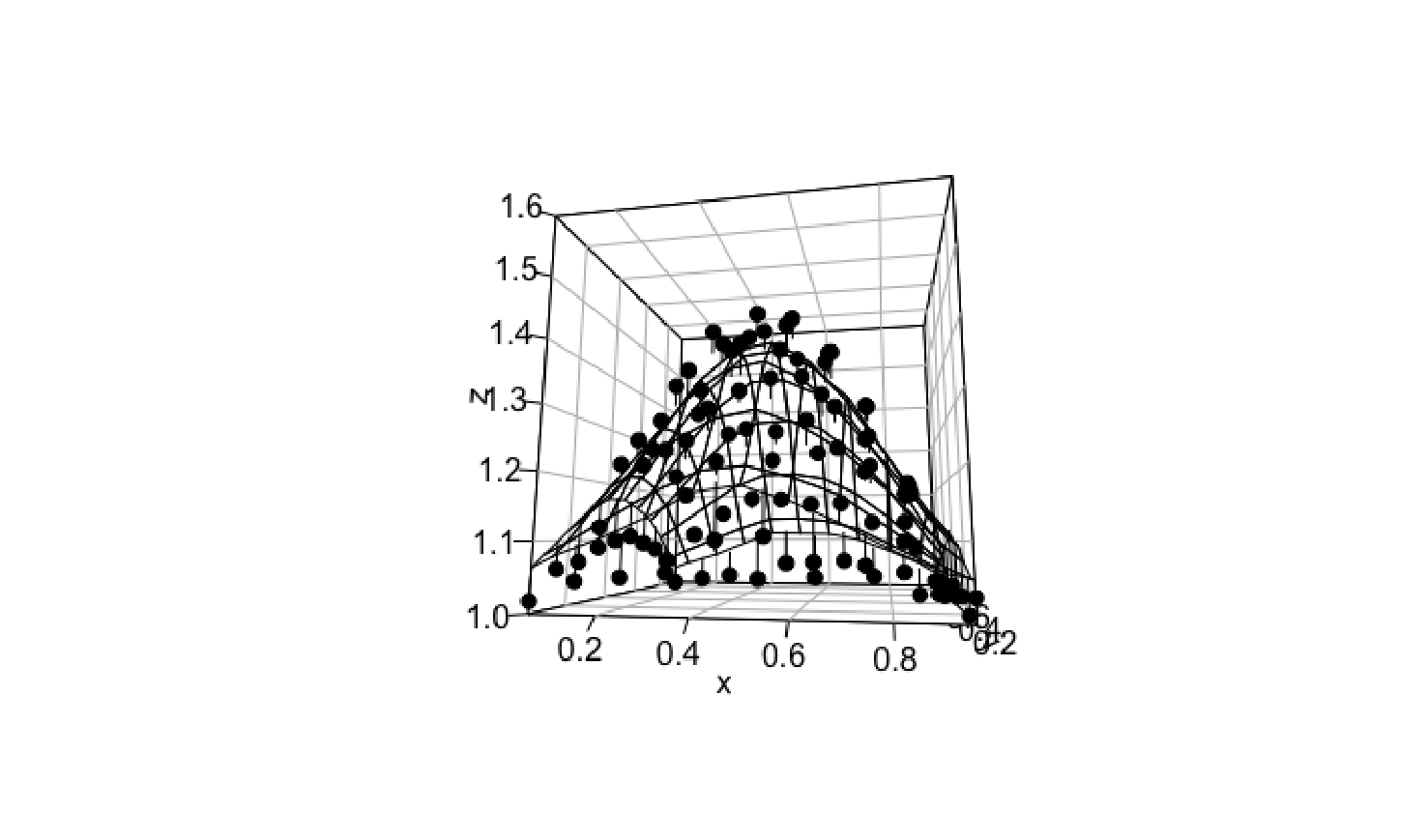}
         \caption{$t_{100}$}
         \label{estm2t100}
     \end{subfigure}
        \caption{Model 3 and 4, in space with $H_{s}=0.40$ and time with $H_{t}=0.65$ (a, b, c) and $H_{t}=0.90$ (d, e, f) in three different instants of time}
        \label{model3-4}
\end{figure}

In the figures above, the points represent $Y_{i}(z_{i})$, $i=1, \dots, 100$, while the surfaces represent $\hat{Y}_{i}(z_{i})$, $i=1, \dots, 100$. It is important to note that the estimation is slightly different for each time, in our example we consider $t_{1}, \dots t_{100}$ (for more details check Remark \ref{sim-rem}). For all the models considered, it is possible to appreciate a similarity between what was simulated and the estimation performed. To verify the performance of the proposed estimation, the following table presents indexes to quantify the goodness of fit.

\begin{table}[h!]
\centering
\begin{tabular}{llllll}
\hline
\multicolumn{6}{|c|}{Model 1}                                                                                                                                                                      \\ \hline
\multicolumn{1}{|l|}{Minimum} & \multicolumn{1}{l|}{$Q_{1}$} & \multicolumn{1}{l|}{Median} & \multicolumn{1}{l|}{$Q_{3}$} & \multicolumn{1}{l|}{Maximum} & \multicolumn{1}{l|}{Adjusted $R^{2}$} \\ \hline
\multicolumn{1}{|l|}{1.0667}  & \multicolumn{1}{l|}{1.2388}  & \multicolumn{1}{l|}{1.3331} & \multicolumn{1}{l|}{1.4284}  & \multicolumn{1}{l|}{1.5995}  & \multicolumn{1}{l|}{0.9419423}        \\ \hline
\multicolumn{6}{|c|}{Model 2}                                                                                                                                                                    \\ \hline
\multicolumn{1}{|l|}{Minimum} & \multicolumn{1}{l|}{$Q_{1}$} & \multicolumn{1}{l|}{Median} & \multicolumn{1}{l|}{$Q_{3}$} & \multicolumn{1}{l|}{Maximum} & \multicolumn{1}{l|}{Adjusted $R^{2}$} \\ \hline
\multicolumn{1}{|l|}{1.0683}  & \multicolumn{1}{l|}{1.2383}  & \multicolumn{1}{l|}{1.3330} & \multicolumn{1}{l|}{1.4291}  & \multicolumn{1}{l|}{1.5986}  & \multicolumn{1}{l|}{0.9877432}        \\ \hline
\multicolumn{6}{|c|}{Model 3}                                                                                                                                                                    \\ \hline
\multicolumn{1}{|l|}{Minimum} & \multicolumn{1}{l|}{$Q_{1}$} & \multicolumn{1}{l|}{Median} & \multicolumn{1}{l|}{$Q_{3}$} & \multicolumn{1}{l|}{Maximum} & \multicolumn{1}{l|}{Adjusted $R^{2}$} \\ \hline
\multicolumn{1}{|l|}{1.0234}  & \multicolumn{1}{l|}{1.1055}  & \multicolumn{1}{l|}{1.1872} & \multicolumn{1}{l|}{1.3032}  & \multicolumn{1}{l|}{1.4518}  & \multicolumn{1}{l|}{0.8695787}        \\ \hline
\multicolumn{6}{|c|}{Model 4}                                                                                                                                                                    \\ \hline
\multicolumn{1}{|l|}{Minimum} & \multicolumn{1}{l|}{$Q_{1}$} & \multicolumn{1}{l|}{Median} & \multicolumn{1}{l|}{$Q_{3}$} & \multicolumn{1}{l|}{Maximum} & \multicolumn{1}{l|}{Adjusted $R^{2}$} \\ \hline
\multicolumn{1}{|l|}{1.0244}  & \multicolumn{1}{l|}{1.1084}  & \multicolumn{1}{l|}{1.1882} & \multicolumn{1}{l|}{1.3043}  & \multicolumn{1}{l|}{1.4534}  & \multicolumn{1}{l|}{0.9182242}        \\ \hline
\end{tabular}
\caption{Minimums, Quartiles 1, 2 (Median), and 3, maximums and Adjusted $R^{2}$}
\end{table}

Considering that $\beta$ is estimated for each time instant, we can see that the values of the minima, the respective quartiles, and maxima, presented for model 1 together with model 2, and for model 3 together with model 4, are very similar. As for the adjusted $R^{2}$, which is an index of the goodness of fit, and which indicates the amount of variability explained by the explanatory variable, it is possible to notice that it decreases in models 2 and 4, with respect to models 1 and 3, respectively. These results are a consequence of the lower variability in models 2 and 4.

\subsection*{Quadratic Mean Error - QME}
The way to build them was by iterations where the observations were accumulated according to time, i.e., in the first iteration the parameter was estimated with 100 observations (first regular grid or $t_1$). The second iteration considered 200 points ($t_1$ and $t_2$) and so on until the 10000 observations were reached (100 observations for each of the 100 different times considered).

\clearpage
\begin{figure}[h!]
     \begin{subfigure}[h]{0.5\textwidth}
         \centering
         \includegraphics[width=1\textwidth]{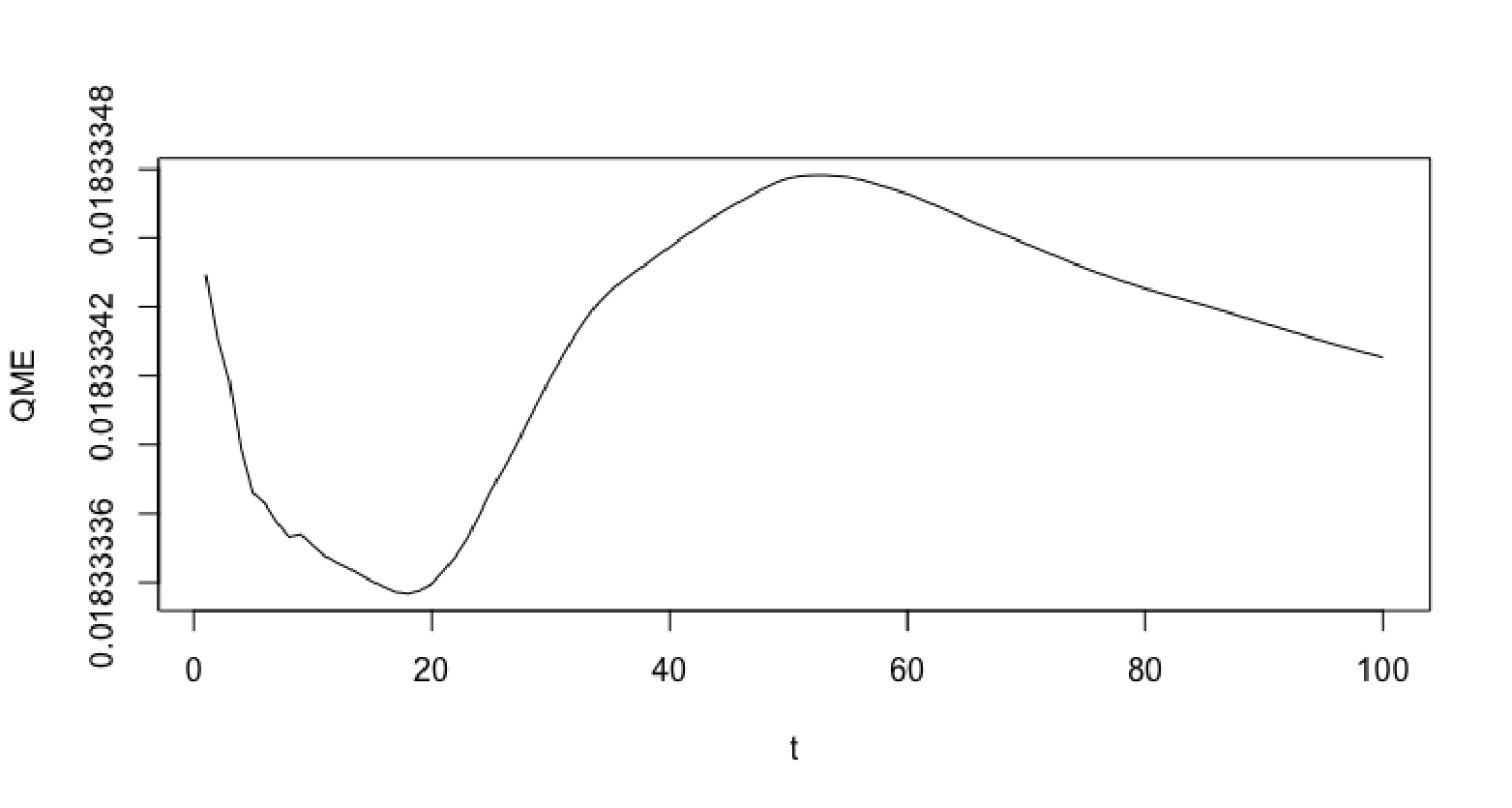}
         \caption{QME for model 1}
         \label{ecm-m1}
     \end{subfigure}
	\hfill
     \begin{subfigure}[h]{0.5\textwidth}
         \centering
         \includegraphics[width=1\textwidth]{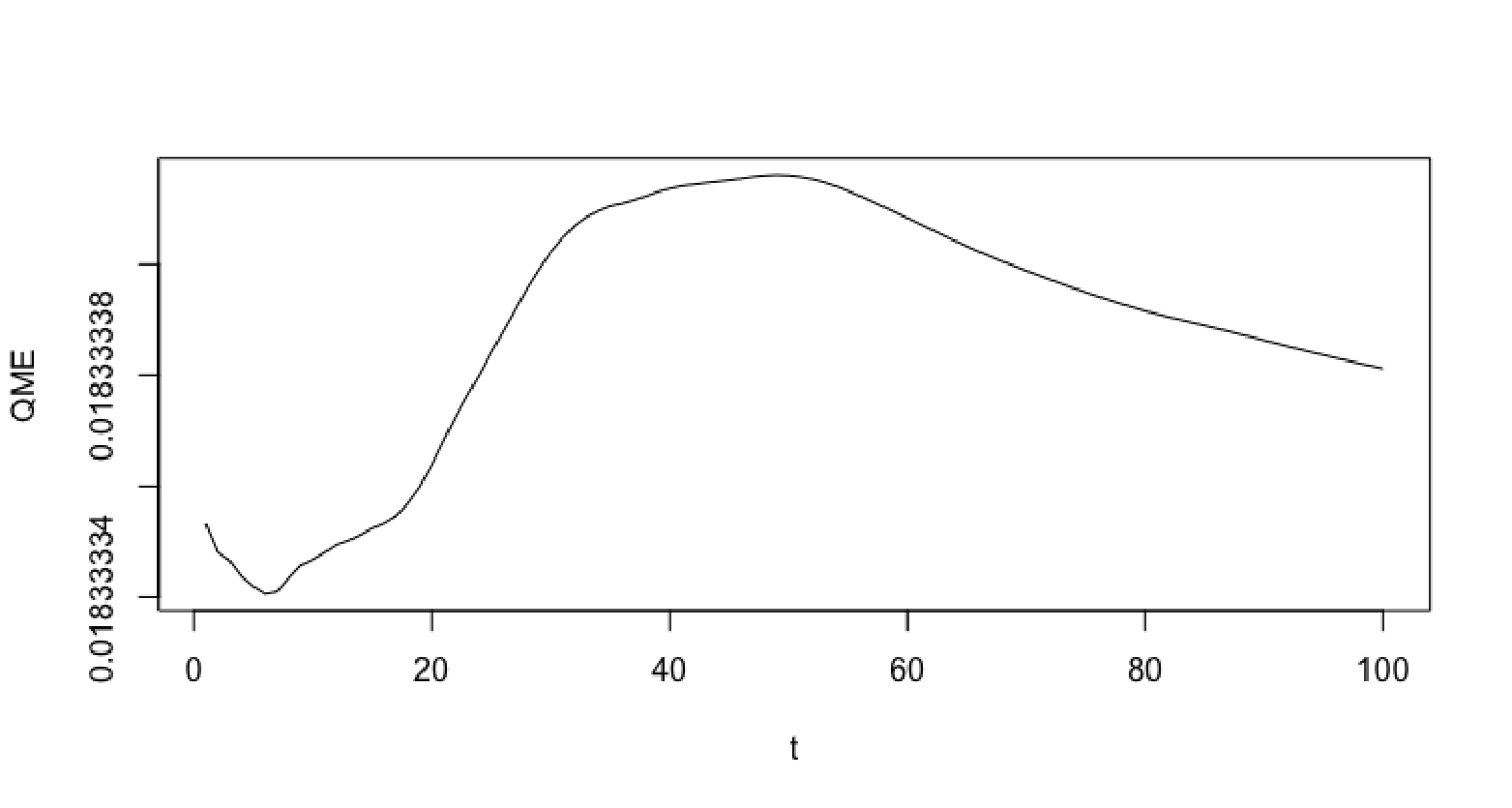}
         \caption{QME for model 2}
         \label{ecm-m2}
     \end{subfigure}
     \hfill
     \begin{subfigure}[h]{0.5\textwidth}
         \centering
         \includegraphics[width=1\textwidth]{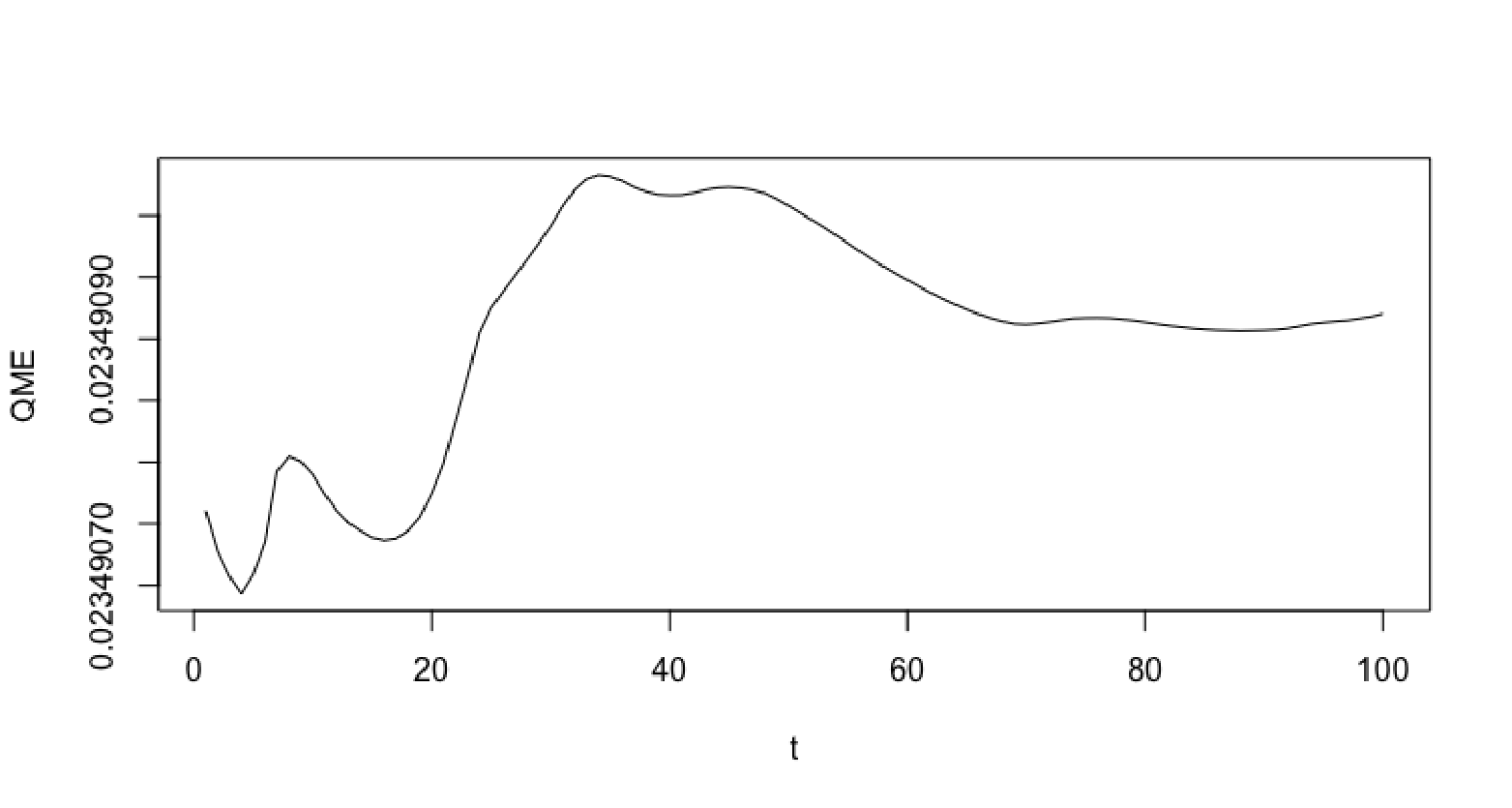}
         \caption{QME for model 3}
         \label{ecm-m3}
     \end{subfigure}
	\hfill
     \begin{subfigure}[h]{0.5\textwidth}
         \centering
         \includegraphics[width=1\textwidth]{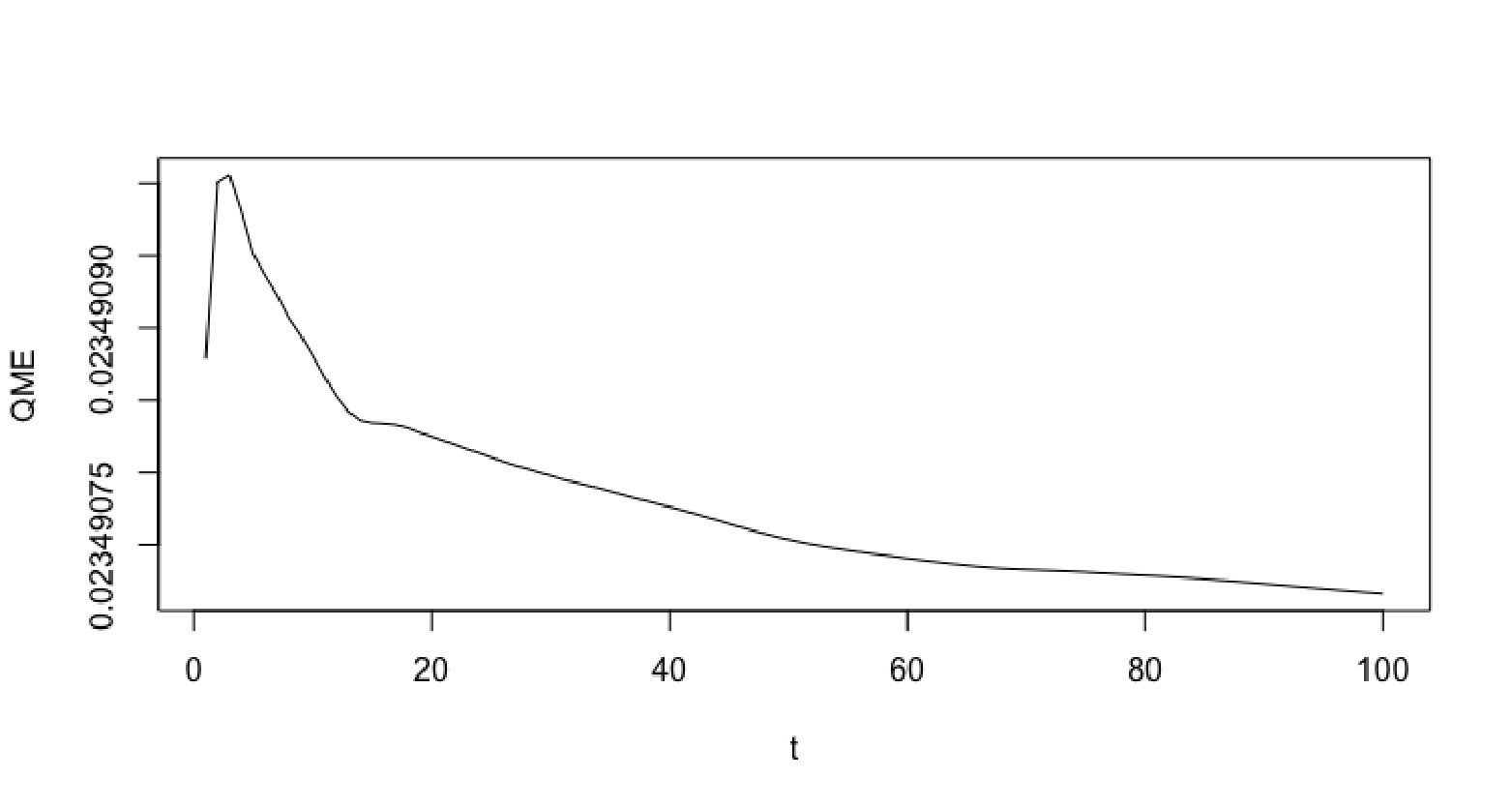}
         \caption{QME for model 4}
         \label{ecm-m4}
     \end{subfigure}
        \caption{Quadratic Mean Error for each model simulated.}
        \label{qme}
\end{figure}

The above graphs show the variation of QME as a function of the number of observations over time. The first thing to note is that the range of the QME is quite small, indicating that, on average, the quadratic difference between the estimated parameter and the true parameter is very small. The second is that the behavior is quite similar for models 1, 2, and 3, where  around $t_{60}$ the behavior of the QME starts to stabilize. Meanwhile, in model 4, around $t_{5}$ the QME values are bigger, and then decrease as the number of observations increases.

\begin{remark} \label{sim-rem}
A GIF file of $t_{1}, t_{2}, \dots, t_{100}$ for each figure presented, can be found in the following links \url{https://github.com/TaniaRoaRojas/GTWR-Simulations}

\end{remark}

\appendix
\section{Appendix}\label{appendix}
\subsection{Proof of Lemma \ref{cov_noise}} \label{ap-cov-noise}
\begin{proof}
From equations \eqref{cov_noise_0} and \eqref{def_noise} we have that the covariance function of $\epsilon=(\epsilon_l)_{l=1:n}$ is 

\begin{equation}\label{cov_epsilon}
\begin{split}
\mathbb{E}(\epsilon_l \epsilon_{l'}) & =\mathbb{E}\left( W^H(V(z_l)) W^H(V(z_{l'})) \right)\\
 & =\mathbb{E}\left( \left( W^H_{t_l^+}(V(u_l)) - W^H_{t_l^-}(V(u_l)) \right) 
\left( W^H_{t_{l'}^+}(V(u_{l'})) - W^H_{t_{l'}^-}(V(u_{l'})) \right) \right) \\
 & = \frac{1}{2}\left[ |t_l -t_{l'} + 2\delta_n|^{2H}+ |t_l -t_{l'} - 2\delta_n|^{2H} - 2|t_l -t_{l'}|^{2H} \right]\\
 & \times  \int_{\Rd} \int_{\mathbb{R}^d} \car_{V(u_l)}(u) f(u-v) \car_{V(u_{l'})}(v) du dv\\
& =\frac{1}{2}\left( \int_{t^{-}_l}^{t_l^+} \int_{t_{l'}^-}^{t_{l'}^+} 2H(2H-1)|t-t'|^{2H-2} dt' dt \right)\\
& \times \left(\int_{\Rd} \int_{\mathbb{R}^d} \car_{V(u_l)}(u) \gamma_{\alpha,d}\|u-v\|^{-d+\alpha} \car_{V(u_{l'})}(v) du dv\right)
 \end{split}
\end{equation}

From \eqref{def_Fourier} we have that the spatial covariance function can be rewritten as 
\begin{equation}\label{cov_colored}
\begin{split}
\lefteqn{\int_{\Rd} \int_{\mathbb{R}^d} \car_{V(u_l)}(u) f(u-v) \car_{V(u_{l'})}(v) du dv}\\ 
& = \int_{\Rd}  \mathcal{F}\car_{V(u_l)}(\xi) \overline{\mathcal{F}\car_{V(u_{l'})}(\xi)} \mu(d\xi)\\
&= \int_{\Rd}  
\left(\int_{\|u-u_l\|\leq \delta_n}e^{-i\xi\cdot u} du\right)
\left(\int_{\|v-u_{l'}\|\leq \delta_n}e^{i\xi\cdot v} dv\right)
\mu(d\xi)\\
&=
\int_{\Rd}  
\left((\delta_n)^d\int_{\|u\|\leq 1}e^{-i\xi\cdot (u_l+\delta_n u)} du\right)
\left((\delta_n)^d\int_{\|v\|\leq 1}e^{i\xi\cdot (u_{l'}+\delta_n v)} dv\right)
\mu(d\xi)\\
&=
(\delta_n)^{2d} 
\int_{\Rd} 
e^{-i\xi\cdot (u_l-u_{l'})}
\left(\int_{\|u\|\leq 1}e^{-i\xi\cdot\delta_n u} du\right)
\left(\int_{\|v\|\leq 1}e^{i\xi\cdot \delta_n v} dv\right)
\mu(d\xi)\\
&=
(\delta_n)^{2d}
\int_{\Rd}  
e^{-i\xi (u_l-u_{l'})/\delta_n}
\left(\int_{\|u\|\leq 1}e^{-i\xi u} du\right)
\left(\int_{\|v\|\leq 1}e^{i\xi v} dv\right)
\left\|\frac{\xi}{\delta_n}\right\|^{-\alpha}(\delta_n)^{-d}d\xi\\
&=
(\delta_n)^{d+\alpha}
\int_{\Rd}  
e^{-i\xi (u_l-u_{l'})/\delta_n}
\left(\int_{\|u\|\leq 1}e^{-i\xi u} du\right)
\left(\int_{\|v\|\leq 1}e^{i\xi v} dv\right)
\left\|\xi\right\|^{-\alpha} d\xi\\
&=
(\delta_n)^{d+\alpha}
\int_{\Rd}  
\left(\int_{\|u-u_l/\delta_n\|\leq 1}e^{-i\xi u} du\right)
\left(\int_{\|v-u_{l'}/\delta_n\|\leq 1}e^{i\xi v} dv\right)\mu(d\xi)\\
& = (\delta_n)^{d+\alpha} Cov\left( W^H \left(\car_{\{\|u-u_l/\delta_n\|\leq 1\}} \right), W^H \left( \car_{\{\|u-u_{l'}/\delta_n\|\leq 1\}}  \right) \right)
\end{split}
\end{equation}

In particular, for $l=l'$ we obtain the spatial variance 
\begin{equation}\label{cov_colored}
\begin{split}
\lefteqn{\int_{\Rd} \int_{\mathbb{R}^d} \car_{V(u_l)}(u) f(u-v) \car_{V(u_{l})}(v) du dv}\\ 
&=
(\delta_n)^{d+\alpha}
\int_{\Rd}  
\left(\int_{\|u\|\leq 1}e^{-i\xi u} du\right)
\left(\int_{\|v\|\leq 1}e^{i\xi v} dv\right)
\left\|\xi\right\|^{-\alpha} d\xi\\
&=
(\delta_n)^{d+\alpha}
\int_{\Rd}  
\left(\int_{\|u\|\leq 1}e^{-i\xi u} du\right)
\left(\int_{\|v\|\leq 1}e^{i\xi v} dv\right)\mu(d\xi)\\
& = (\delta_n)^{d+\alpha} Var\left( W^H \left( \car_{\{\|u\|\leq 1\}} \right) \right)\\
& = \sigma^2 (\delta_n)^{d+\alpha},
\end{split}
\end{equation}
where $\sigma^2=Var\left( W^H \left( \car_{\{\|u\|\leq 1\}} \right) \right)$.
Thus, the variance of the fractional colored noise is
\begin{equation}\label{var_epsilon}
\begin{split}
\mathbb{E}(\epsilon_l^2) 
 & =\sigma^2  2^{2H} (\delta_n)^{2H+d+\alpha} .
 \end{split}
\end{equation}
\end{proof}
\vspace{-2cm}

\subsection{Proof of Lemma \ref{conv_denom}} \label{ap-conv-denom}
\begin{proof}
The $jk_{th}$ component of the matrix $X^T\W(z_i)X$ is
\begin{equation}\label{XWXjk}
  (X^{T}\W(z_i)X)_{jk}= \sum_{l=1}^{n}X_{lj}X_{lk}\W_{il}.
\end{equation}

We study the asymptotic expectation of \eqref{XWXjk}, from assumption \ref{C1} we obtain 

\begin{equation}\label{Exp_XWXjk}
\begin{split}
\lefteqn{\frac{1}{nh^{d+1}}\E\left(\left(X^{T}\W(z_i)X\right)_{jk}\right)}\\
& = \frac{1}{nh^{d+1}}\sum_{l=1}^{n}\ \chi_{jk}(z_l,z_l) K_h\left(z_{l}-z_{i}\right)\\
& = \frac{1}{h^{d+1}}\int_{\mathbb{R}^{d+1}}\sum_{l=1}^{n} \chi_{jk}(z_l,z_l) K_h\left(z_{l}-z_{i}\right)\car_{V(z_l)}(z) dz\\
&\approx \; \frac{1}{h^{d+1}}\int_{\mathbb{R}^{d+1}}\chi_{jk}(z,z) K_h\left(z-z_{i}\right)dz\\
&= \int_{\mathbb{R}^{d+1}}\chi_{jk}(z_i+hz,z_i+hz) K\left(z\right)dz\\
&\approx \;  \chi_{jk}(z_i,z_i) + \mathcal{O}\left(|h|^{\alpha_{\chi}}\right).
\end{split}
\end{equation}


\begin{remark}
Note that condition \eqref{C2} implies that the covariance matrix $\chi(z_i,z_i)=\left(\chi_{jk}(z_i,z_i)\right)_{j,k=1:n}$ is an invertible matrix.
\end{remark}

Continuing, we calculate the variance of \eqref{XWXjk}.

\begin{eqnarray}\label{split_var_D}
Var\left(\left(X^{T}\W(z_{i})X\right)_{jk}\right)  &=& \sum_{l,l'=1}^n \Gamma_{jk}(z_l,z_{l'})K_h(z_l-z_i) K_h(z_{l'}-z_i)\nonumber\\
    &=& \sum_{l=1}^n \Gamma_{jk}(z_l,z_{l})\W_{il}^2\nonumber\\
    &+& \sum_{\substack{1\leq l\neq l'\leq n \\ \|z_l-z_{l'}\|\leq k\delta_n}} \Gamma_{jk}(z_l,z_{l'})\W_{il}\W_{il'} \nonumber\\ 
    &+& \sum_{\substack{1\leq l\neq l'\leq n \\ \|z_l-z_{l'}\| > k\delta_n}} \Gamma_{jk}(z_l,z_{l'})\W_{il}\W_{il'}\nonumber\\
    &:=& D^{(1)}_{jk,n}(z_i) + D^{(2)}_{jk,n}(z_i)  + D^{(3)}_{jk,n}(z_i),
\end{eqnarray}

First, we study the term $D^{(1)}_{jk,n}(z_i) $ in \eqref{split_var_D}. Let us consider the case $0<\|z_l-z_{i}\| < \delta_n$ 

\begin{equation*}
 \frac{1}{n^{2}h^{2(d+1)}} D^{(1,1)}_{jk,n}(z_i)  = \frac{1}{n^{2}h^{2(d+1)}}  \sum_{\substack{1\leq l \leq n \\ \|z_l-z_{i}\| < \delta_n}}  \Gamma_{jk}(z_l,z_l)K^2_h\left(z_l-z_i\right),
\end{equation*}

Using \ref{C2} (iv) and Remark \ref{grilla}, we can obtain
\begin{eqnarray}
\frac{1}{n^{2}h^{2(d+1)}} D^{(1,1)}_{jk,n}(z_i) & \preceq & 2^{(d+1)} \frac{C_{k,d}}{n^{2}h^{2(d+1)}}   \int_{\mathbb{R}^{d+1}}   K^2_h\left(z -z_i\right)dz \nonumber \\
 &=& 2^{(d+1)} \frac{C_{k,d}}{n^{2}h^{(d+1)}}  \|K\|_{2}^{2} \label{d12}
\end{eqnarray}\\

Now, we consider the case $0<\|z_l-z_{i}\|\geq \delta_n$. Using  Assumption \ref{K1} (iv), we obtain

\begin{align}
 \frac{1}{n^{2}h^{2(d+1)}} D^{(1,2)}_{jk,n}(z_i)  &= \frac{1}{n^{2}h^{2(d+1)}}  \sum_{\substack{1\leq l \leq n \\ \|z_l-z_{i}\| \geq \delta_n}}  \Gamma_{jk}(z_l,z_l)K^2_h\left(z_l-z_i\right) \nonumber \\
&=   \frac{1}{n^{2}h^{2(d+1)}}   \sum_{\substack{1\leq l \leq n \\ \|z_l-z_{i}\| \geq \delta_n}}  \Gamma_{jk}(z_l,z_l)f_K\left( \Vert z_l-z_i \Vert \right) K_h\left(z_l-z_i\right)  \nonumber \\
&\leq     \frac{1}{n^{2}h^{2(d+1)}} \dfrac{L(n)}{n^{\gamma}}   \sum_{\substack{1\leq l \leq n \\ \|z_l-z_{i}\| \geq \delta_n}}  \Gamma_{jk}(z_l,z_l) K_h\left(z_l-z_i\right) \nonumber \\
&\leq    \frac{L(n)}{n^{1+\gamma}} \frac{1}{ h^{2(d+1)}} \int_{\mathbb{R}^{d+1}}  \sum_{l=1}^{n} \left\vert  \Gamma_{jk}(z_l,z_l) \right\vert  K_h\left(z_l-z_i\right)  \car_{V(z_l)}(z) dz
 \nonumber \\
&\approx   \frac{L(n)}{n^{1+\gamma}} \frac{1}{ h^{2(d+1)}}  \int_{\mathbb{R}^{d+1}} \left\vert  \Gamma_{jk}(z,z)  \right\vert  K_h\left(z-z_{i}\right)dz  \nonumber \\
&\preceq \frac{L(n)}{n^{1+\gamma}} \frac{C_{k,d}}{ h^{(d+1)}}  \int_{\mathbb{R}^{d+1}}  K\left(z \right)dz =  \frac{L(n)}{n^{1+\gamma}} \frac{C_{k,d}}{ h^{(d+1)}},   \label{d13}
\end{align}
where in the last inequality we use \ref{C2} (iv).  Consequently,  by \eqref{d12} and \eqref{d13}, we can get 
\begin{equation}
 \frac{1}{n^{2}h^{2(d+1)}} D^{(1,1)}_{jk,n}(z_i)  \preceq  C \frac{L(n)}{n^{1+\gamma}} \frac{1}{ h^{(d+1)}} = \mathcal{O} \left( \frac{L(n)}{n^{1+\gamma}} \frac{1}{ h^{(d+1)}} \right) \label{d01}
 \end{equation}

Secondly, we consider the term $D^{(2)}_{jk,n}(z_i)$ in \eqref{split_var_D}, i.e. when $0<\|z_l-z_{l'}\|\leq k\delta_n$

\begin{equation}\label{D2_0}
\begin{split}
\lefteqn{   \frac{1}{nh^{(d+1)}} D^{(2)}_{jk,n}(z_i) }\\
 & =   \frac{1}{n h^{(d+1)}} \sum_{\substack{1\leq l\neq l'\leq n \\ \|z_l-z_{l'}\|\leq k\delta_n}} \Gamma_{jk}(z_l,z_{l'})K_h(z_l-z_i)K_h(z_{l'}-z_i) \\
 & =    \frac{1}{n h^{(d+1)}} \left[ \sum_{\substack{1\leq l\neq l'\leq n \\ \|z_l-z_{l'}\|\leq k\delta_n}} \Gamma_{jk}(z_l,z_{l})K_h(z_l-z_i)K_h(z_{l}-z_i)\right.\\
 & + \sum_{\substack{1\leq l\neq l'\leq n \\ \|z_l-z_{l'}\|\leq k\delta_n}}\Gamma_{jk}(z_l,z_{l}) K_h(z_l-z_i)\left(K_h(z_{l'}-z_i) - K_h(z_{l}-z_i)\right)\\
& + \sum_{\substack{1\leq l\neq l'\leq n \\ \|z_l-z_{l'}\|\leq k\delta_n}} \left(\Gamma_{jk}(z_l,z_{l'})-\Gamma_{jk}(z_l,z_l)\right)K_h(z_l-z_i)K_h(z_{l}-z_i)\\
&+ \left.\sum_{\substack{1\leq l\neq l'\leq n \\ \|z_l-z_{l'}\|\leq k\delta_n}} \left(\Gamma_{jk}(z_l,z_{l'})-\Gamma_{jk}(z_l,z_l)\right)K_h(z_l-z_i)\left(K_h(z_{l'}-z_i)-K_h(z_{l}-z_i)\right)\right]
\end{split}
\end{equation}

From regularity condition \eqref{C2} and \eqref{K1}

\begin{equation}\label{D2}
\begin{split}
\lefteqn{  \frac{1}{n h^{(d+1)}}D^{(2)}_{jk,n}(z_i) }\\
& \leq   \frac{1}{n h^{(d+1)}}  \left[ \sum_{\substack{1\leq l\neq l'\leq n \\ \|z_l-z_{l'}\|\leq k\delta_n}} \Gamma_{jk}(z_l,z_{l})K^2_h(z_l-z_i)\right.
 \\
 & +  C_K\sum_{\substack{1\leq l\neq l'\leq n \\ \|z_l-z_{l'}\|\leq k\delta_n}}\Gamma_{jk}(z_l,z_{l}) K_h(z_l-z_i)\left\|z_{l} - z_{l'}\right\|^{\alpha_{K}}\\
& + C_{\Gamma} \sum_{\substack{1\leq l\neq l'\leq n \\ \|z_l-z_{l'}\|\leq k\delta_n}} \left\|z_l-z_{l'}\right\|^{\alpha_{\Gamma}}K_h^2(z_l-z_i)\\
& +C_k C_{\Gamma}\left. \sum_{\substack{1\leq l\neq l'\leq n \\ \|z_l-z_{l'}\|\leq k\delta_n}} \left\|z_l-z_{l'}\right\|^{\alpha_{\Gamma}}K_h(z_l-z_i)\left\|z_{l} - z_{l'}\right\|^{\alpha_{K}}\right]\\
&= D^{(2,1)}_{jk,n}+D^{(2,2)}_{jk,n}+D^{(2,3)}_{jk,n}+D^{(2,4)}_{jk,n}.
\end{split}
\end{equation}

Note that, similarly to \eqref{volume_V_zl}  we have 
$\frac{1}{n}\sum_{l'=1}^n\car_{\{0<\|z_l-z_{l'}\|\leq k \delta_n\}} \appn \frac{k^{d+1}}{n}$.
Therefore, by \eqref{D2}, we can get 
\begin{equation}\label{D21}
\begin{split}
D^{(2,1)}_{jk,n} & =  \frac{1}{n h^{(d+1)}} \sum_{l=1}^n \Gamma_{jk}(z_l,z_{l})K^2_h(z_l-z_i)\sum_{l'=1}^n\car_{\{0<\|z_l-z_{l'}\|\leq k\delta_n\}}\\
& \approx  \frac{k^{d+1}}{n h^{(d+1)}} \sum_{l=1}^n \Gamma_{jk}(z_l,z_{l})K^2_h(z_l-z_i)
\end{split}
\end{equation}
We split the sum in two cases $0<\|z_l-z_{i}\| < \delta_n$ and $0<\|z_l-z_{i}\| \geq  \delta_n$, then  the same arguments as in the case of the term $D^{(1)}_{jk,n} $,  allow us to obtain
   
\begin{equation}
\begin{split}
 \frac{1}{n h^{(d+1)}} D^{(2,1)}_{jk,n} \; & \preceq  \frac{k^{d+1}}{n^{2} h^{2(d+1)}} \sum_{l=1}^n \Gamma_{jk}(z_l,z_{l})K^2_h(z_l-z_i) \\
 & \preceq \frac{k^{d+1}}{n^{2}h^{2(d+1)}}  \sum_{\substack{1\leq l \leq n \\ \|z_l-z_{i}\| < \delta_n}}  \Gamma_{jk}(z_l,z_l)K^2_h\left(z_l-z_i\right) \nonumber \\
 & + \frac{k^{d+1}}{n^{2}h^{2(d+1)}} \dfrac{L(n)}{n^{\gamma}}   \sum_{\substack{1\leq l \leq n \\ \|z_l-z_{i}\| \geq \delta_n}}  \Gamma_{jk}(z_l,z_l) K_h\left(z_l-z_i\right) \nonumber \\
& \preceq 2^{(d+1)} \frac{C_{k,d}}{n^{2}h^{2(d+1)}}   \int_{\mathbb{R}^{d+1}}   K^2_h\left(z -z_i\right)dz \nonumber \\  
& + \frac{L(n)}{n^{1+\gamma}} \frac{k^{d+1}}{ h^{2(d+1)}} \int_{\mathbb{R}^{d+1}}  \sum_{l=1}^{n} \left\vert  \Gamma_{jk}(z_l,z_l) \right\vert  K_h\left(z_l-z_i\right)  \car_{V(z_l)}(z) dz \nonumber \\ 
 & \preceq \; 2^{(d+1)} \frac{C_{k,d}}{n^{2}h^{(d+1)}}  \|K\|_{2}^{2}  + C_{k,d} \frac{L(n)}{n^{1+\gamma}} \frac{k^{d+1}}{ h^{2(d+1)}}  \int_{\mathbb{R}^{d+1}}   K_h\left(z-z_{i}\right)dz  \nonumber \\
& \leq C \frac{L(n)}{n^{1+\gamma}} \frac{k^{d+1}}{ h^{(d+1)}} 
\end{split}
\end{equation}

Continuing, we have that for $D^{(2,2)}_{jk,n}$ similarly to the previous terms 

\begin{equation}\label{D22}
\begin{split}
\frac{1}{(\delta_n)^{\alpha_k}}D^{(2,2)}_{jk,n}& = \frac{C_K}{nh^{d+1}(\delta_n)^{\alpha_k}}\sum_{l=1}^n \Gamma_{jk}(z_l,z_{l}) K_h(z_l-z_i)\sum_{l'=1}^n\left\|z_{l} - z_{l'}\right\|^{\alpha_{K}}\car_{\{0<\|z_l-z_{l'}\|\leq k\delta_n\}}\\
&\leq \frac{ C_K (k\delta_n)^{\alpha_k}}{nh^{d+1}(\delta_n)^{\alpha_k}}\sum_{l=1}^n \Gamma_{jk}(z_l,z_{l}) K_h(z_l-z_i)\sum_{l'=1}^n\car_{\{0<\|z_l-z_{l'}\|\leq k\delta_n\}}\\
& \preceq  \frac{ k^{\alpha_K+d+1}C_K}{h^{d+1}} \int_{\mathbb{R}^{d+1}} \Gamma_{jk}(z,z) K_h(z-z_i) dz\\
& =  k^{\alpha_K+d+1}C_K\int_{\mathbb{R}^{d+1}} \Gamma_{jk}(z_i+hz,z_i+hz) K(z) dz\\
& \approx   k^{\alpha_K+d+1}C_K \Gamma_{jk}(z_i,z_i) +
\mathcal{O}(|h|^{\alpha_{\Gamma}}).
\end{split}
\end{equation}

\begin{equation}\label{D23}
\begin{split}
\frac{1}{(\delta_n)^{\alpha_{\Gamma}}} D^{(2,3)}_{jk,n} & = \frac{C_{\Gamma}}{nh^{d+1}(\delta_n)^{\alpha_{\Gamma}}} \sum_{l=1}^n K_h^2(z_l-z_i) \sum_{l'=1}^n \left\|z_{l} - z_{l'}\right\|^{\alpha_{\Gamma}}\car_{\{0<\|z_l-z_{l'}\|\leq k\delta_n\}}\\
&\leq \frac{C_{\Gamma}(k \delta_n)^{\alpha_{\Gamma}} }{nh^{d+1}(\delta_n)^{\alpha_{\Gamma}}} \sum_{l=1}^n K_h^2(z_l-z_i) \sum_{l'=1}^n \car_{\{0<\|z_l-z_{l'}\|\leq k\delta_n\}}\\
& \approx \frac{k^{\alpha_{\Gamma}+d+1}C_{\Gamma}}{h^{d+1}} \int_{\mathbb{R}^{d+1}} K^2_h(z-z_i) dz\\
&=  k^{\alpha_{\Gamma}+d+1}C_{\Gamma} \|K\|_\infty^2.
\end{split}
\end{equation}

\begin{equation}\label{D24}
\begin{split}
\frac{1}{(\delta_n)^{\alpha_{\Gamma}+\alpha_k}} D^{(2,4)}_{jk,n}& = \frac{ C_K C_{\Gamma}}{nh^{d+1}(\delta_n)^{\alpha_{\Gamma}+\alpha_K}} \sum_{l=1}^n K_h(z_l-z_i) \sum_{l'=1}^n \left\|z_{l} - z_{l'}\right\|^{\alpha_{\Gamma}+\alpha_K}\car_{\{0<\|z_l-z_{l'}\|\leq k\delta_n\}}\\
&\leq \frac{ C_K C_{\Gamma}(k\delta_n)^{\alpha_{\Gamma}+\alpha_K} }{nh^{d+1}(\delta_n)^{\alpha_{\Gamma}+\alpha_K}} \sum_{l=1}^n K_h(z_l-z_i) \sum_{l'=1}^n \car_{\{0<\|z_l-z_{l'}\|\leq k\delta_n\}}\\
& \approx  \frac{k^{\alpha_{\Gamma}+\alpha_K+d+1}C_K C_{\Gamma}
}{h^{d+1}} \int_{\mathbb{R}^{d+1}} K_h(z-z_i) dz\\
&=  k^{\alpha_{\Gamma}+\alpha_K+d+1}C_K C_{\Gamma}  \|K\|_\infty.
\end{split}
\end{equation}

Thus, from \eqref{D21}, \eqref{D22}, \eqref{D23}, and \eqref{D24} we have
\begin{equation}\label{D2_2}
\begin{split}
\lefteqn{\frac{1}{n^{2}h^{2(d+1)}} D^{(2)}_{jk,n}(z_i) }\\
& \approx  C \frac{L(n)}{n^{1+\gamma}} \frac{k^{d+1}}{ h^{(d+1)}} 
  + \mathcal{O}\left(|h|^{\alpha_{\Gamma}}\vee (\delta_n)^{\alpha_K}\vee (\delta_n)^{\alpha_{\Gamma}}\right).
 \end{split}
\end{equation}

Finally we consider the term $D^{(3)}_{jk,n}(z_i)$, using Assumptions \eqref{C2} we obtain
\begin{equation}\label{D3}
\begin{split}
\lefteqn{ \frac{1}{n^2h^{2(d+1)}} D^{(3)}_{jk,n} }\\
& = \frac{1}{n^2h^{2(d+1)}} \sum_{1\leq l\neq l'\leq n} \Gamma_{jk}(z_l,z_{l'})K_h(z_l-z_i)K_h(z_{l'}-z_i)\car_{\{ \|z_l-z_{l'}\| > k\delta_n \}} \\
& \leq \frac{C_{k,d,\beta} (\delta_n)^{d+1+\beta} }{n^2h^{2(d+1)}} \sum_{1\leq l\neq l'\leq n} K_h(z_l-z_i)K_h(z_{l'}-z_i)\\
& \approx  \frac{ C_{k,d,\beta} (\delta_n)^{d+1+\beta} }{h^{2(d+1)}} \int_{\mathbb{R}^{d+1}} \int_{\mathbb{R}^{d+1}}K_h(z-z_i)K_h(z'-z_i)dzdz'\\
& = C_{k,d,\beta} (\delta_n)^{d+1+\beta}
\end{split}
\end{equation}

Substituting \eqref{d01}, \eqref{D2_2} and \eqref{D3} into the Equation \eqref{split_var_D}, and using that $2\lambda(\mathcal{S}^{d-1})(\delta_n)^{d+1}=1/n$, we obtain
\begin{align*}
\frac{1}{n^2h^{2(d+1)} }Var\left(\left(X^{T}\W(z_{i})X\right)_{jk}\right)  &\preceq  C \frac{\left(1+k^{(d+1)}\right) L(n) }{n^{1+ \gamma }h^{(d+1)}} + C_{k,d,\theta} (\delta_n)^{d+1+\theta} \\
&\approx \frac{C'}{n^{1+\nu}},
\end{align*}
 where $\gamma >\nu=\frac{\theta}{d+1}$, and $h$ such that  
 $L(n) n^{-1-\gamma}h^{(d+1)} = n^{-1-\theta}$.
 
Whether $\theta>0$ then $\nu>0$ and the $L^2$ rate of $\frac{1}{nh^{d+1}}\left(X^{T}\W(z_{i})X\right)_{jk}$ is faster than $1/n$, therefore Borell-Cantelli lemma allows us to obtain 
$$\frac{1}{nh^{d+1}}\left(X^{T}\W(z_{i})X\right)_{jk}\xrightarrow[n \to \infty]{a.s.} \chi_{j,k}(z_i,z_i).$$

Note that $\theta >-d-1$, thus $1+\nu>0$ and we obtain the $L^2$ convergence and therefore the convergence in probability when $\theta\leq 0$. 

\end{proof}

\begin{remark}
Let us note that the equality $L(n) n^{-1-\gamma}h^{(d+1)}= n^{-1-\theta}$ imposes a condition on the speed at which $h^{(d+1)}$ decreases to zero. In fact, we need that $h^{(d+1)} = \frac{L(n)}{n^{\gamma - \theta/(d+1)}}$ with $\gamma > \theta/(d+1)$. 
\end{remark}

\section*{Acknowledgments}
Héctor Araya was partially supported by FONDECYT 11230051 project.
Lisandro Fermín was partially supported by MathAmSud Tomcat 22-math-10.
Tania Roa was partially supported by FONDECYT 3220043 Postdoc project.
Soledad Torres was partially supported by Basal Project FB210005 and FONDECYT project 1221373.
Lisandro Fermín and Soledad Torres were partially supported by FONDECYT projects 1230807.
Héctor Araya, Tania Roa and Soledad Torres  were partially supported by ECOS210037(C21E07) and Mathamsud AMSUD210023 projects.


\newpage
\textbf{Declaration of generative AI and AI-assisted technologies in the writing process} \\

During the preparation of this work the authors used Google translator and DeepL in order to check grammar. After using this tool/service, the authors reviewed and edited the content as needed and takes full responsibility for the content of the publication.

%
%
%
%
%

\end{document}